\documentclass{article}

\usepackage{hyperref}       

\usepackage{macros}
\usepackage[margin=1.5in]{geometry}

\usepackage[utf8]{inputenc} 
\usepackage[T1]{fontenc}    
\usepackage{url}            
\usepackage{booktabs}       
\usepackage{amsfonts}       
\usepackage{nicefrac}       
\usepackage{microtype}      
\usepackage{xcolor}         

\title{
$k$-Sliced Mutual Information:\\A Quantitative Study of Scalability with Dimension}

\author{
      Ziv Goldfeld\\
      Cornell University\\
      \texttt{goldfeld@cornell.edu} \\\\
      Kristjan Greenewald\\
  MIT-IBM Watson AI Lab\\
  \texttt{kristjan.h.greenewald@ibm.com}\\\\
  Theshani Nuradha\\
  Cornell University\\
  \texttt{pt388@cornell.edu} \\\\
  Galen Reeves\\
  Duke University\\
  \texttt{galen.reeves@duke.edu}
  }

\begin{document}

\maketitle
\allowdisplaybreaks
\begin{abstract}
Sliced mutual information (SMI) is defined as an average of mutual information (MI) terms between one-dimensional random projections of the random variables. It serves as a surrogate measure of dependence to classic MI that preserves many of its properties but is more scalable to high dimensions. However, a quantitative characterization of how SMI itself and estimation rates thereof depend on the ambient dimension, which is crucial to the understanding of scalability, remain obscure. 
This work provides a multifaceted account of the dependence of SMI on dimension, under a broader framework termed $k$-SMI, which considers projections to $k$-dimensional subspaces. Using a new result on the continuity of differential entropy in the 2-Wasserstein metric, we derive sharp bounds on the error of Monte Carlo (MC)-based estimates of $k$-SMI, with explicit dependence on $k$ and the ambient dimension, revealing their interplay with the number of samples. We then combine the MC integrator with the neural estimation framework to provide an end-to-end $k$-SMI estimator, for which optimal convergence rates are established. We also explore asymptotics of the population $k$-SMI as dimension grows, providing Gaussian approximation results with a residual that decays under appropriate moment bounds. All our results trivially apply to SMI by setting $k=1$. Our theory is validated with numerical experiments and is applied to sliced InfoGAN, which altogether provide a comprehensive quantitative account of the scalability question of $k$-SMI, including SMI as a special case when $k=1$. 
\end{abstract}

\section{Introduction}

Mutual information (MI) is a fundamental measure of dependence between random variables \cite{CovThom06,ElGamal2011}, with a myriad of applications in information theory, statistics, and more recently machine learning \cite{haussler1994bounds,battiti1994using,viola1997alignment,chen2016infogan,alemi2017deep,higgins2017beta,DNNs_Tishby2017,oord2018representation,achille2018information,gabrie2018entropy,ICML_Info_flow2019,goldfeld2020information}. Its appeal stems from the favorable structural properties it possesses, such as meaningful units (bits or nats), identification of independence, entropy decompositions, and convenient variational forms. However, modern learning applications require estimating MI between high-dimensional variables based on data, which is known to be notoriously hard with exponential in dimension sample complexity \cite{Paninski2003,mcallester2020formal}. To alleviate this impasse, sliced MI (SMI) was recently introduced by a subset of the authors as a surrogate dependence measure that preserves much of the classic structure while being more scalable for computation and estimations in high dimensions~\cite{goldfeld2021sliced}.

Inspired by slicing techniques for statistical divergences \cite{rabin2011wasserstein,vayer2019sliced,lin2021projection,nadjahi2020statistical}, SMI is defined as an average of MI terms between one-dimensional projections of the high-dimensional variables. Beyond showing that SMI inherits many properties of its classic counterpart, \cite{goldfeld2021sliced} demonstrated that it can be estimated with (optimal) parametric error rates in all dimensions by combining a MI estimator between scalar variables with a MC integrator. However, the bounds from \cite{goldfeld2021sliced} rely on high-level assumptions that may be hard to verify in practice and hide dimension-dependent constants whose characterization is crucial for understanding scalability in dimension. Furthermore, when projecting high-dimensional variables it is natural to ask what information can be extracted from more than just the real line, say, a subspace of dimension $k \geq 1$, but this extension was not considered in \cite{goldfeld2021sliced}. This work defines $k$-SMI (which employs projections to $k$-dimensional subspaces), and provides a comprehensive quantitative study of its dependence on dimension, encompassing the MC error, formal guarantees for neural estimators, and asymptotics of the population $k$-SMI as dimension increases. All our results trivially apply for the original SMI case (when $k=1$), thereby closing the aforementioned gaps in analysis from \cite{goldfeld2021sliced}.

\subsection{Contributions}

The objective of this work is provide a thorough quantitative study of the dependence of SMI on dimension. We do so under the slightly broader framework of $k$-SMI, which we define between random variables $X$ and $Y$ with values in $\RR^{d_x}$ and $\RR^{d_y}$ as
\begin{equation}
    \ksmi(X;Y):=\int_{\sti(k,d_x)}\int_{\sti(k,d_y)}\sI(\rA^\intercal X;\rB^\intercal Y)d\sigma_{k,d_x}(\rA)d\sigma_{k,d_y}(\rB),\label{EQ:kSMI_intro}
\end{equation}
where $\sti(k,d)$ is the Stiefel manifold of $d \times k$ matrices with orthonormal columns and $\sigma_{k,d}$ is~its uniform measure. $k$-SMI coincides with SMI when $k=1$, but to further support it as a natural~extension, we show that structural properties of SMI derived in \cite{goldfeld2021sliced} still hold for any $1\leq k\leq \min\{d_x,d_y\}$. We then move to study formal guarantees for $k$-SMI estimation, targeting explicit dependence on $(k,d_x,d_y)$. A key technical tool we employ is a new continuity result of differential entropy with respect to (w.r.t.) the 2-Wasserstein distance $\sW_2$, which we derive using the HWI inequality from \cite{otto2000generalization,gentil2020entropic}. Our continuity claim strengthens the one from \cite{PolyWu_Wasserstein2016} in two ways: (i) it replaces the $(c_1,c_2)$-regularity condition therein with the weaker requirement of finite Fisher information, and (ii) it sharpens the constant multiplying $\sW_2$ to be optimal. As a corollary, we show that the differential entropy of a projected variable, say $\sh(\rA^\tr X)$, is Lipschitz continuous w.r.t. the Frobenius norm on the $\sti(k,d)$.

Lipschitzness is pivotal for obtaining dimension-dependent bounds on MC-based estimates of $k$-SMI. We bound the MC error in terms of the variance of $\sI(\rA^\tr X;\rB^T Y)$ when $(\rA,\rB)$ are uniform over their respective Stiefel manifolds. Lipschitz~continuity of differential entropy implies Lipschitzness of this projected MI, which enables controlling its variance via a concentration argument over $\sti(k,d)$. The resulting bound scales as $O\big(\sqrt{k(1/d_x+1/d_y)/m}\big)$, where $m$ is the number of MC samples and the constant is explicitly expressed via basic characteristics of the $(X,Y)$ distribution (its covariance and Fisher information matrices). This result, which also applies to standard SMI, sharpens the bounds from \cite{goldfeld2021sliced}, characterizes the dependence on dimension, and holds under primitive assumptions on the joint distribution. Furthermore, the bound reveals that higher dimension can~shrink the error in some cases---a surprising observation which is also verified numerically on synthetic examples.

In addition to MC integration, the $k$-SMI estimator employs a generic MI estimator between $k$-dimensional variables. We instantiate this estimator via the neural estimation framework based on the Donsker-Varadhan (DV) variational form \cite{belghazi2018mine} (see also \cite{poole2018variational,song2019understanding,chan2019neural}). The neural estimator is realized by an $\ell$-neuron shallow ReLU network and the effective convergence rate of the resulting $k$-SMI estimate is explored. We lift the convergence rates derived in \cite{sreekumar2022neural} for neural estimators of $f$-divergences to the $k$-SMI problem. The resulting rate scales as $O\big(k^{1/2}(\ell^{-1/2}+m^{-1/2}+ kn^{-1/2})\big)$, where $\ell$ is the number of neurons, $m$ is the number of MC samples, and $n$ is the number of $(X,Y)$ samples. Equating $\ell$, $m$, and $n$ results in the (optimal) parametric rate. Our result also shows that neural estimation of $k$-SMI requires milder smoothness assumptions on the population distributions. Namely, we relax the smoothness level $\lfloor(d_x +d_y)/2 \rfloor+3$ imposed in \cite{sreekumar2022neural} to $k+3$, i.e., adapting to the projection dimension rather than the ambient one. This is a significant relaxation since we often have $d_x,d_y\gg k$.

To further understand the effect of the ambient dimension, we explore how $\ksmi(X;Y)$ behaves as $d_x,d_y\to\infty$. To that end, we first provide a full characterization of $\ksmi(X,Y)$ between jointly~Gaussian variables, revealing that it scales as $k^2/(d_xd_y)$ times the squared Frobenius norm of the cross-covariance matrix. We then show that general $k$-SMI can be decomposed into a Gaussian part plus a residual term that quantifies the average distance (over projections) from Gaussianity. The latter is intimately related to the conditional central limit theorem (CLT) phenomenon \cite{meckes2012approximation,meckes2012projections,reeves2017conditional}, and we use those ideas to identify approximate isotropy conditions under which the residual vanishes as $d_x,d_y\to\infty$. Lastly, we conduct an empirical study that validates our theory and explores applications to independence testing and sliced infoGAN. Specifically, we revisit the infoGAN generative model \cite{chen2016infogan} and replace the classic MI used therein with SMI. Training the model, we find that it successfully learns disentangled representations despite the low-dimensional projections, suggesting that SMI can replace classic MI even in applications with complex underlying structure.

\section{Background and Preliminaries}

\subsection{Notation and Definitions}

\paragraph{Notation.} For $d\geq 1$, $\| \cdot \|$ is the Euclidean norm in $\RR^d$, $\langle\cdot,\cdot\rangle$ is the inner product, while $\|\cdot\|_1$ is the $\ell^1$ norm. We use $\| \cdot \|_{\op}$ and $\|\cdot\|_{\rF}$ for the operator and Frobenius norms of matrices, respectively. Matrix inequalities are understood in the sense of (partial) semi-definite ordering, i.e., we write $\rA \succeq\rB$ when $\rA-\rB$ is positive semi-definite. The Stiefel manifold of $d \times k$ matrices with orthonormal columns is denoted by $\sti(k,d)$. For a $d \times k$ matrix $\rA$, we use $\proj^\rA:\RR^d\to\RR^k$ for the orthogonal projection onto the row space of $\rA$. 

Let $\calP(\R^d)$ denote the space of Borel probability measures on $\RR^d$, and set $\cP_2(\R^d):= \{ \mu \in \calP(\R^d) : \int \| x \|^2 d\mu(x) < \infty \}$ as the subset of distributions with finite 2nd absolute moment. For $\mu,\nu\in\cP(\RR^d)$, we use $\mu\otimes\nu$ to denote a product measure, while $\supp(\mu)$ designates the support of $\mu$. We use $\leb$ for the Lebesgue measure on $\RR^d$, 
and denote the subset of probability measures that are absolutely continuous w.r.t. $\leb$ by $\cP_\s{ac}(\RR^d)$. For a  measurable map $f$, the pushforward of $\mu$ under $f$ is denoted by $f_{\sharp}\mu = \mu \circ f^{-1}$, i.e., if $X \sim \mu$ then $f(X) \sim f_{\sharp}\mu$. For $a,b\in\RR$, we use the notation $a \wedge b = \min \{ a,b \}$ and $a \vee b = \max \{a,b \}$. We write $a\lesssim_{x} b$ when $a\leq C_x b$ for a constant $C_x$ that depends only on $x$ ($a\lesssim b$ means the constant is absolute). 

For a multi-index $\alpha=(\alpha_1, \ldots, \alpha_d) \in \ZZ^d_{\geq 0}$, the partial derivative operator of order $\|\alpha\|_1 $ is denoted by $D^\alpha= \frac{\partial^{\alpha_1}}{\partial^{\alpha_1} x_1} \ldots \frac{\partial^{\alpha_d}}{\partial^{\alpha_d} x_d} $. For an open set $\cU \subseteq \RR^d$ and integer $s \geq 0$, the class of functions whose partial derivatives up to order $s$ all exist and are continuous on $\cU$ is denoted by $\sC^s(\cU)$, and we define the subclass  $\sC^s_b(\cU):=\{f \in \sC^s(\cU): \max_{\alpha: \|\alpha\|_1 \leq s} \|D^\alpha  f \|_{\infty,\cU} \leq  b  \}$. The restriction of $f: \RR^d \to \RR$ to $\cX \subseteq \RR^d$ is denoted by~$f|_{\cX}$. For compact $\cX$, slightly abusing notation, we set $ \| \cX\| := \sup_{x\in \cX} \|x\|$.

\paragraph{Divergences and information measures.} 
Let $\mu,\nu\in\cP(\RR^d)$ satisfy $\mu\ll\nu$, i.e., $\mu$ is absolutely continuous w.r.t. $\nu$. 
The relative entropy and the relative Fisher information are defined, respectively, as
$\dkl(\mu\|\nu):=\int_{\RR^d} \log(d\mu/d\nu)d\mu$ and $\smash{\sJ(\mu \| \nu) := \int_{\RR^d} \big\| \nabla \log (d\mu/d\nu)\big \|^2 d\mu}$. The 2-Wasserstein distance between $\mu,\nu \in \calP_2(\R^d)$ is $\sW_2(\mu,\nu):=  \inf_{\pi \in \Pi(\mu,\nu)} \big (\int_{\R^d \times \R^d} \|x-y\|^2 \, d \pi(x,y) \big)^{1/2}$, where  $\Pi(\mu,\nu)$ is the set of couplings of $\mu$ and $\nu$. All three measures are divergences, i.e., non-negative and nullify if and only if (iff) $\mu=\nu$. In fact, $\sW_2$ is a metric on $\cP_2(\R^d)$, which metrizes weak convergence plus convergence of 2nd moments. 

MI and differential entropy are defined from the relative entropy as follows. Consider a pair of random variables $(X,Y)\sim\mu_{XY}\in\cP(\RR^{d_x}\times\RR^{d_y})$ and denote the corresponding marginal distributions by $\mu_X$ and $\mu_Y$. The MI between $X$ and $Y$ is given by $\sI(X;Y):=\dkl(\mu_{XY}\|\mu_X\otimes\mu_Y)$ and serves as a measure of dependence between those random variables. The differential entropy of $X$ is defined as $\sh(X)=\sh(\mu_X):=-\dkl(\mu_X\|\leb)$. 
MI between (jointly) continuous variables and differential entropy are related via $\sI(X;Y)=\sh(X)+\sh(Y)-\sh(X,Y)$; decompositions in terms of conditional entropies are also available~\cite{CovThom06}. The Fisher information of $X\sim\mu$ is $\sJ(\mu):=\sJ(\mu\|\leb)$. Denoting the density of $\mu$ by $f_\mu$, the Fisher information matrix of $\mu$ is $\rJ_\rF(\mu):= \EE\left[( \nabla \log f_\mu)( \nabla \log f_\mu)^\tr\right]$, and we have $\Tr\big(\rJ_\rF(\mu) \big)=\sJ(\mu)$. 



\subsection{Lipschitz Continuity of Projected Differential Entropy}\label{SUBSEC:continuity}



A key technical tool we use is a new continuity result of differential entropy w.r.t. the 2-Wasserstein distance. It strengthens an earlier version of this result from \cite{PolyWu_Wasserstein2016}, and may be of independent interest.

\begin{lemma}[Wasserstein continuity]\label{LEMMA:Wass_cont}
Let $\mu,\nu\in\cP_2(\RR^d)$ satisfy $\mu\ll\leb$ and $\sh(\mu),\sJ(\nu)<\infty$.~Then
\[
\sh(\mu)-\sh(\nu)\leq \sqrt{\sJ(\nu)}\sW_2(\mu,\nu),\]
and the constant above is optimal in the sense that \ \ $\sup_{\substack{\mu\neq\nu:\\\sh(\mu),\sJ(\nu)<\infty}}\frac{\sh(\mu)-\sh(\nu)}{\sqrt{\sJ(\nu)}\sW_2(\mu,\nu)}=1$.
\end{lemma}
The proof of the lemma, given in \cref{APPEN:continuity_proofs}, follows by invoking the HWI inequality for the difference of relative entropies \cite{otto2000generalization,gentil2020entropic} with an isotropic Gaussian reference measure $\gamma_\sigma=\cN(0,\sigma^2 \rI_d)$, re-expressing the relative entropy difference in terms of differential entropies, and taking the limit as $\sigma\to0$. 

\begin{remark}[Comparison to \cite{PolyWu_Wasserstein2016}]
Continuity of differential entropy w.r.t. the $\sW_2$ was previously derived in \cite[Proposition 1]{PolyWu_Wasserstein2016}, but via a different argument, under stronger conditions, and without an optimal constant. The inequality from \cite{PolyWu_Wasserstein2016} assumed $(c_1,c_2)$-regularity of the density of $\nu$ (i.e., that $\|\nabla \log f_\nu(x)\|\leq c_1\|x\|+c_2$, for all $x\in\RR^d$), which is stronger than $\sJ(\nu)<\infty$ when $\nu\in\cP_2(\RR^d)$. 
\end{remark}

A rather direct implication of \cref{LEMMA:Wass_cont} is the following Lipschitz continuity of projected entropy (also proven in \cref{APPEN:continuity_proofs}), which plays a key role in the subsequent analysis of $k$-SMI estimation.

\begin{proposition}[Lipschitzness of projected entropy]\label{PROP:ent_lip}
Let $\mu \in \cP_2(\RR^d)$ have covariance matrix $\Sigma_\mu$ and $\sJ(\mu)<\infty$. For any $\rA,\rB\in\sti(k,d)$, we have $\big| \sh(\proj^{\rA}_\sharp \mu) \mspace{-3mu} - \mspace{-3mu} \sh(\proj^{\rB}_\sharp \mu) \big| \mspace{-3mu}\le\mspace{-3mu}\sqrt{ k \|\rJ_\rF( \mu)\|_{\op} \|\Sigma_\mu\|_{\op}}  \|A\mspace{-1.5mu}-\mspace{-1.5mu}B\|_\rF$.

\end{proposition}

\section{{\boldmath $k$}--Sliced Mutual Information}

SMI was defined in \cite{goldfeld2021sliced} as an average of MI terms between one-dimensional projections of the considered random variables. As higher dimensional projections preserve more information about the original $(X,Y)$, we extend this definition to $k$-dimensional projections. 
\begin{definition}[$k$-sliced mutual information]
For $1\leq k\leq d_x\wedge d_y$, the $k$-SMI between $(X,Y)\sim \mu_{XY}\in \cP(\RR^{d_x}\times\RR^{d_y})$ is defined in \eqref{EQ:kSMI_intro},
where $\sigma_{k,d}$ is the uniform distribution on $\sti(d,k)$.
\end{definition}
$k$-SMI can be equivalently expressed in term of conditional (classic) MI as $\ksmi(X;Y)=\sI(\rA^\intercal X;\rB^\intercal Y|\rA,\rB)$, where $(\rA,\rB)\sim \sigma_{k,d_x}\otimes\sigma_{k,d_y}$, i.e., $(\rA,\rB)$ are independent and uniform over the respective Stiefel manifolds. $k$-SMI reduces to the SMI from \cite{goldfeld2021sliced} when $k=1$. Below we show that $\ksmi$ preserves the structural properties of SMI, as derived in \cite[Section 3]{goldfeld2021sliced}.

\begin{remark}[Related definitions]
$k$-SMI entropy decompositions and chain rule require defining $k$-sliced entropy and conditional $k$-SMI. For $(X,Y,Z)\sim \mu_{XYZ}\in\cP(\RR^{d_x}\times\RR^{d_y}\times \RR^{d_z})$ and $(\rA,\rB,\rC)\sim \sigma_{k,d_x}\otimes\sigma_{k,d_y}\otimes\sigma_{k,d_z}$, the $k$-sliced entropy of $X$ is $\ksh(X):=\sh(\rA^\tr X|\rA)$, while the conditional version given $Y$ is given by $\ksh(X|Y):=\sh(\rA^\tr X|\rA,\rB,\rB^\tr Y)$. The condition $k$-SMI between $X$ and $Y$ given $Z$ is $\ksmi(X;Y|Z):=\sI(\rA^\intercal X;\rB^\intercal Y|\rA,\rB,\rC,\rC^\tr Z)$. 
\end{remark}



\subsection{Structural Properties}
We verify that $k$-SMI preserves structural properties previously established in \cite{goldfeld2021sliced} for SMI.

\begin{proposition}[$k$-SMI properties]\label{PROP:kSMI_prop}
For any $1\leq k\leq d_x\wedge d_y$, the following properties hold:
\begin{enumerate}[leftmargin=.45cm]
    \item\label{prpty:NonNeg} \textbf{Identification of independence:} $\ksmi(X;Y)\geq 0$ with equality iff $X$ and $Y$ are independent. 
    \item\label{prpty:Bounds} \textbf{Bounds:} For integers $k_1\mspace{-2mu}<\mspace{-2mu}k_2$: $\smi_{k_1}(X;Y)\mspace{-1mu}\leq\mspace{-1mu} \smi_{k_2}(X;Y) \mspace{-1mu}\leq\mspace{-8mu} \sup\limits_{\substack{\rA\in \sti(k_2,d_x)\\ \rB\in \sti(k_2,d_y)}}\mspace{-5mu} \sI(\rA^\intercal  X; \rB^\intercal  Y) \mspace{-1mu}\leq\mspace{-1mu} \s{I}(X;Y)$.
    \item\label{prpty:KL} \textbf{Relative entropy and variational form:} Let $(\tilde X,\tilde Y)\mspace{-1mu}\sim\mspace{-1mu} \mu_X\mspace{-1mu}\otimes\mspace{-1mu}\mu_Y$ and $(\rA,\rB)\mspace{-1mu}\sim\mspace{-1mu}\sigma_{k,d_x}\mspace{-1mu}\otimes \mspace{-1mu}\sigma_{k,d_y}$, then  
    \begin{align*}
    \ksmi(X;Y)&=\s{D}_{\s{KL}}\big((\proj^\rA,\proj^\rB)_\sharp \mu_{XY}\big\|(\proj^\rA,\proj^\rB)_\sharp \mu_X\otimes\mu_Y\big| \rA,\rB\big)\\
    &=\sup_{f:\,\sti(k,d_x)\times\sti(k,d_y)\times \RR^{2k}\to\RR}\mspace{-10mu} \EE\big[f(\rA,\rB,\rA^\tr X,\rB^\tr Y)\big]-\log\left(\EE\left[e^{f(\rA,\rB,\rA^\tr\tilde{X},\rB^\tr\tilde{Y})}\right]\right),
    \end{align*}
    where  the supremum is over all measurable functions for which both expectations are finite. 
    \item\label{prpty:entropy} \textbf{Entropy decomposition:} $\ksmi(X;Y)=\ksh(X)-\ksh(X|Y)=\ksh(Y)-\ksh(Y|X)=\ksh(X)+\ksh(Y)-\ksh(X,Y)$, provided that all the relevant (joint / marginal / conditional) densities exist.
    \item\label{prpty:Chain} \textbf{Chain rule:} For any $X_1,\ldots,X_n,Y,Z$, we have $\ksmi(X_1,\ldots,X_n;Y)=\ksmi(X_1;Y)+\sum_{i=2}^n\ksmi(X_i;Y|X_1,\ldots,X_{i-1})$. In particular, $\ksmi(X,Y;Z)=\ksmi(X;Z)+\ksmi(Y;Z|X)$.
    \item\label{prpty:Tensor} \textbf{Tensorization:} For mutually independent $\{\mspace{-1mu}(X_i,\mspace{-3mu}Y_i)\mspace{-1mu}\}_{i=1}^n$, $\ksmi\big(\mspace{-1mu}\{X_i\}_{i=1}^n;\mspace{-2mu}\{Y_i\}_{i=1}^n\mspace{-1mu})\mspace{-2.5mu}=\mspace{-2.5mu} \sum\limits_{i=1}^n\mspace{-2mu} \ksmi(X_i;\mspace{-1mu} Y_i)$.
    \end{enumerate}
\end{proposition}
The proposition is proven in Appendix \ref{APPEN:kSMI_prop_proof} via a direct extension of the $k=1$ argument from \cite{goldfeld2021sliced}.


\section{Estimation and Asymptotics of {\boldmath $k$}-SMI in High Dimensions}

As shown in \cite{goldfeld2021sliced}, SMI can be estimated from high-dimensional data by combining a MI estimator between scalar random variables and a MC integration step. However, the bounds from \cite{goldfeld2021sliced} do not explicitly capture dependence on the ambient dimension, which is crucial for understanding scalability of the approach. We now extend the estimator from \cite{goldfeld2021sliced} to $k$-SMI and provide formal guarantees with explicit dependence on $k$, $d_x$, and $d_y$, thus closing the said gap.


To estimate $k$-SMI, let $\{(X_i,Y_i)\}_{i=1}^n$ be i.i.d. from $\mu_{XY}\in\cP(\RR^{d_x}\times\RR^{d_y})$ and proceed as follows:
\begin{enumerate}[leftmargin=.45cm]
    \item Draw $\{(\rA_j,\rB_j)\}_{j=1}^m$ i.i.d. from $\sigma_{k,d_x}\otimes\sigma_{k,d_y}$ (i.e., each pair is uniform on $\sti(k,d_x)\times \sti(k,d_y)$).\footnote{A simple approach for sampling the uniform distribution on $\sti(k,d)$ is to draw $kd$ random samples from $\cN(0,1)$, arrange them into an $d\times k$ matrix $\Lambda$, and compute $\Lambda(\Lambda^\tr \Lambda)^{-1/2}$ (cf. \cite[Theorem 2.2.1]{chikuse2003statistics}). A slightly more efficient approach is to first apply a QR decomposition to $\Lambda$ and then follow the aforementioned sampling method only to the Q matrix. Note that for $k = O(1)$, both computation times are linear in $d$ (QR decomposition via the Schwarz-Rutishauser algorithm is $O(d k^2)$) \cite{gander1980algorithms}.}
    \item Compute $\big\{(\rA_j^\tr X_i, \rB_j^\tr Y_i)\big\}_{j,i=1}^{m,n}$, which, for fixed $(\rA_j,\rB_j)$, are samples from $(\proj^\rA,\proj^\rB)_\sharp\mu_{XY}$. 
    \item For each $j=1,\ldots,m$, a MI estimator between $k$-dimensional random vectors is applied to the $n$ samples corresponding to $(\rA_j,\rB_j)$ to obtain an estimate $\hat{\sI}\big((\rA_j^\tr X)^n,(\rB_j^\tr Y)^n\big)$ of $\sI(\rA_j^\tr X;\rB_j^\tr Y)$, where $(\rA_j^\tr X)^n:=(\rA_j^\tr X_1,\ldots,\rA_j^\tr X_n)$ and $(\rB_j^\intercal Y)^n$ is defined similarly.
    \item Take a MC average of the above estimates, resulting in the $k$-SMI estimator:
\begin{equation}
\widehat{\smi}_k^{m,n}:=\frac{1}{m}\sum_{j=1}^m \hat{\sI}\big((\rA_j^\tr X)^n,(\rB_j^\tr Y)^n\big).\label{EQ:SMI_est}
\end{equation}
\end{enumerate}  

We provide formal guarantees for the quality of the $\widehat{\smi}_k^{m,n}$ estimator given a generic $k$-dimensional MI estimator $\hat{\sI}(\cdot,\cdot)$ in Step 3. Afterwards, we instantiate the latter as a neural MI estimator and provide explicit convergence rates. To get further insight into the dependence on dimension, we~study asymptotics of Gaussian $k$-SMI as $d_x,d_y\to\infty$ and corresponding Gaussian approximation arguments. 

\subsection{Error Bounds with Explicit Dimension Dependence}

Our analysis decomposes the overall error of $\widehat{\smi}_k^{m,n}$ into the MC error plus the error of the $k$-dimensional MI estimator $\hat{\sI}(\cdot,\cdot)$. 
We first consider an arbitrary estimator $\hat{\sI}(\cdot,\cdot)$ whose error is (implicitly) upper bounded by $\delta_k(n)$ and focus 
on analyzing the MC error, targeting explicit dependence on $k$, $d_x$, and~$d_y$. As in \cite{goldfeld2021sliced}, the statement relies on the following assumption on the $k$-dimensional estimator $\hat{\sI}(\cdot;\cdot)$.

\begin{assumption}\label{ASSUMP:scalar_MI}
$(X,Y)\sim\mu_{XY}\in\cP(\RR^{d_x}\times\RR^{d_y})$ is such that $\sI(\rA^\tr X;\rB^\tr Y)$ can be estimated by $\hat{\sI}\big((\rA^\tr X)^n,(\rB^\tr Y)^n\big)$ with error at most $\delta_k(n)$, uniformly over $(\rA,\rB)\in\sti(k,d_x)\times\sti(k,d_y)$.
\end{assumption}

\begin{theorem}[$k$-SMI estimation error]\label{THM:MC_error}
Let $\mu_{XY}\in\cP_2(\RR^{d_x}\times\RR^{d_y})$ satisfy \cref{ASSUMP:scalar_MI}, have~marginal covariance matrices $\Sigma_X$ and $\Sigma_Y$, and $\sJ(\mu_{XY}\mspace{-1mu})\mspace{-3mu}<\mspace{-3mu}\infty$. Then the estimator from \eqref{EQ:SMI_est} has error bounded~by
\begin{equation}
    \EE\left[ \left| \ksmi(X;Y) - \widehat{\smi}_k^{m,n} \right|\right] \leq C(\mu_{XY})\sqrt{\frac{k(d_x+d_y)}{d_xd_y}}m^{-\frac 12}+\delta_k(n),\label{EQ:MC_bound}
\end{equation}
where $C(\mu_{XY})=21\sqrt{\|\rJ_\rF(\mu_{XY})\|_{\op}\big(\|\Sigma_X\|_{\op}\vee\|\Sigma_Y\|_{\op}\big)}$. 

\end{theorem}


The proof of Theorem \ref{THM:MC_error} (in \cref{APPEN:MC_error_proof}) bounds the MC error by $\smash{\big(\Var\big(i_{XY}(\rA,\rB)\big)/m\big)^{\frac 12}}$, where $i_{XY}(\rA,\rB):=\sI(\rA^\tr X;\rB^\tr Y)$ and $(\rA,\rB)\sim\sigma_{k,d_x}\otimes\sigma_{k,d_y}$. We then use the continuity result from Proposition \ref{PROP:ent_lip} along with the entropy decomposition of $k$-SMI (Proposition \ref{PROP:kSMI_prop}, Claim 4) to show that $i_{XY}$ is Lipschitz continuous (w.r.t. the Frobenius norm) on $\sti(k,d_x)\times\sti(k,d_y)$. Concentration of Lipschitz functions on the Stiefel manifold and the Efron-Stein inequality then imply the above bound. This result clarifies the dependence of the MC error on $k$, $d_x$, and $d_y$, and reveals scaling rates of the parameters with $m$ for which (high-dimensional) convergence holds true.

\begin{remark}[Comparison to \cite{goldfeld2021sliced}]
Theorem 1 from \cite{goldfeld2021sliced} treats the $k=1$ case under stronger high-level assumptions and without identifying the dependence on dimension. Namely, assuming the uniform bound $\|i_{XY}\|_{L^\infty} \leq M$, they control the variance by $M^2/4$ to obtain the $m^{-1/2}$ rate, although $M$ generally depends on $(d_x,d_y)$. Herein, we rely on the finer observation that $i_{XY}$ is Lipschitz and use concentration results to get a dimension-dependent bound in terms of basic characteristics of~$(X,Y)$.
\end{remark}

\begin{remark}[Blessing of dimensionality]\label{REM:BoD}
The constant in the MC error may decay as dimension grows. For instance, if $X$ and $Y$ are both $d$-dimensional with identity covariance matrices, then $\|\Sigma_X\|_{\op},\|\Sigma_Y\|_{\op}$ are $O_d(1)$. For such $(X,Y)$, the MC bound decays to 0 as $d \to \infty$, assuming that $\|\rJ_\rF(\mu_{XY})\|_{\op}$ grows at most sublinearly with $d$. Also note that 
$C(\mu_{XY})$ has the same invariances as the $k$-SMI: it is invariant to translations and scalings of the form $(X,Y) \mapsto (sX,sY)$ for $s \ne 0$.  

\end{remark}

\subsection{Neural Estimation}\label{SUBSEC:SMI_NE}

We now instantiate the $k$-dimensional MI estimator via the neural estimation framework of \cite{sreekumar2021non,sreekumar2022neural}, and obtain an explicit bound on $\delta_k(n)$ in terms of $m$, $n$, $k$, and the size of the neural network. 

Neural estimation of MI relies on the DV variational form \[\sI(U;V)=\sup_{f:\RR^{d_u}\times\RR^{d_v}\to\RR} \EE[f(U,V)]-\log\left(e^{\EE[f(\tilde{U},\tilde{V})]}\right),\]
where $(U,V)\sim \mu_{UV}$, $(\tilde{U},\tilde{V})\sim \mu_U\otimes\mu_V$, and $f$ is a measurable function for which the expectations above are finite. Define the class of $\ell$-neuron ReLU network as
\[
\cG_{d_u, d_v}^\ell(a) :=\left\{ g:\mathbb{R}^{d_u + d_v} \rightarrow \mathbb{R}:\begin{aligned}
    &\qquad\qquad g(z)=\sum\nolimits_{i=1}^\ell \beta_i \phi\left(\langle w_i, z\rangle+b_i\right)+\langle w_0, z\rangle + b_0, 
    \\&\max_{1 \leq i \leq \ell}\|w_{i}\|_1 \vee |b_i| \leq 1,~  
 \max_{1 \leq i \leq \ell}|\beta_i| \leq \frac{a}{2\ell},~ |b_0|,\|w_0\|_1 \leq a \end{aligned}\right\},
\]
where $\phi(z)=z\vee0$ is the ReLU activation; set the shorthand $\cG_{d_u, d_v}^\ell=\cG_{d_u, d_v}^\ell(\log \log \ell \vee 1)$. Given i.i.d. data $(U_1,V_1),\ldots,(U_n,V_n)$ from $\mu_{UV}$, the neural estimator parameterizes the DV potential $f$ by the class $\cG_{d_u, d_v}^\ell$ and approximates expectations by sample means,\footnote{Negative samples, i.e., from $\mu_X\otimes\mu_Y$, can be obtained from the positive one via $(U_1,V_{\sigma(1)}),\ldots,(U_n,V_{\sigma(n)})$, where $\sigma\in S_n$ is a permutation such that $\sigma(i)\neq i$, for all $i=1,\ldots,n$.}  resulting in the estimate 
\[\hat{\sI}^{\ell}_{d_u,d_v}(U^n,V^n):=\sup_{g \in \cG_{d_u, d_v}^\ell} \frac 1n \sum_{i=1}^n g(U_i,V_i)- \log\left(\frac 1n \sum_{i=1}^n e^{g(U_i,V_{\sigma(i)})}\right).\]

For $k$-SMI neural estimation, we set \[\widehat{\smi}_{k,\s{NE}}^{\ell,m,n}:=\frac{1}{m}\sum_{j=1}^m \hat{\sI}^{\ell}_{k,k}\big((\rA_j^\intercal X)^n,(\rB_j^\intercal Y)^n\big),\]
i.e., we use $\hat{\sI}^{\ell}_{k,k}$ as the $k$-dimensional MI estimator in \eqref{EQ:SMI_est}. 
This estimator is readily implemented by parallelizing $m$ $\ell$-neuron ReLU nets with inputs in $\RR^{2k}$ and scalar outputs. We provide explicit convergence rates for it over an appropriate distribution~class, drawing upon the results of \cite{sreekumar2022neural} for neural estimation of $f$-divergences (see also \cite{sreekumar2021non}). For compact $ \cX \subset \RR^{d_x}$ and $\cY\subset \RR^{d_y}$, let $\cP_{\s{ac}}(\cX\times\cY):=\{\mu_{XY}\in\cP_{\s{ac}}(\RR^{d_x}\times\RR^{d_y}):\supp(\mu_{XY})\subseteq\cX\times\cY\}$, and denote the density of $\mu_{XY}$ by $f_{XY}$.
The distribution class of interest~is
\[\cF^k_{d_x,d_y}(M,b):=\mspace{-3mu}\left\{ \mu_{XY}\in \cP_{\s{ac}}(\cX \times\cY):\begin{aligned}
    &  \exists\, r \in\sC^{k+3}_{b}(\cU) \ \textnormal{for some open set }\cU \supset \cX \times \cY
    \\ & \  \textnormal{s.t.} \ \log f_{XY} = r|_{\cX \times \cY}, \   \sI(X;Y) \leq M \end{aligned}\right\},
\] 
which, in particular, contains distributions whose densities are bounded from above and below on $\cX\times\cY$ with a smooth extension to an open set covering $\cX\times\cY$. This includes uniform distributions, truncated Gaussians, truncated Cauchy distributions, etc. 

We next provide convergence rates for the $k$-SMI estimator from \eqref{EQ:SMI_est}, uniformly over~$\cF^k_{d_x,d_y}(M,b)$. 
\begin{theorem}[Neural estimation error]\label{THM:SMI_NE}
For any $M,b\geq 0$, we have 
\[
  \sup_{\mu_{X,Y}\in\cF^k_{d_x,d_y}(M,b)}\mathbb{E}\left[\left|\ksmi(X;Y)- \widehat{\smi}_{k,\s{NE}}^{\ell,m,n}\right|\right]  \lesssim_{M,b,k,d_x,d_y,\|\cX \times \cY\|}k^{\frac 12}\big(m^{-\frac 12}+\ell^{-\frac{1}{2}}+ kn^{-\frac 12}\big).
\]
The dependence on $d_x,d_y$ above is only through the MC bound \eqref{EQ:MC_bound} (explicit) and $\| \cX \times \cY\|$~(implicit). 

\end{theorem} 

\cref{THM:SMI_NE} is proven in \cref{APPEN:SMI_NE_proof} by combining the MC bound from  \cref{THM:MC_error} with the neural estimation error bound from \cite[Proposition 2]{sreekumar2022neural}. To apply that bound for each $\sI(\rA^\tr X;\rB^\tr Y)$, where $(\rA,\rB)\in\sti(k,d_x)\times\sti(k,d_y)$, we show that the existence of an extension $r$ of $\log f_{XY}$ with $k+3$ continuous and uniformly bounded derivatives implies that the density of $(\rA^\tr X,\rB^\tr Y)$ also has such an extension.

\begin{remark}[Parametric rate and optimality]
Taking $\ell\asymp m\asymp n$, the resulting rate in \cref{THM:SMI_NE} is parametric, and hence minimax optimal. This result implicitly assumes that $M$ is known when picking the neural net parameters. This assumption can be relaxed to mere existence of (an unknown) $M$, resulting in an extra $\mathrm{polylog}(\ell)$ factor multiplying the $n^{-1/2}$ term. 
\end{remark}

\begin{remark}[Comparison to \cite{sreekumar2022neural}]
Neural estimation of classic MI under the framework of \cite{sreekumar2022neural} requires the density to have H\"older smoothness $s\geq \lfloor(d_x +d_y)/2 \rfloor+3$. For $\ksmi(X;Y)$, smoothness of $k+3$ is sufficient (even though the ambient dimension is the same), 
which mean it can be estimated over a larger class of distributions. This is another virtue of slicing in addition to fast convergence rates. For SMI (i.e., $k=1$) as in \cite{goldfeld2021sliced}, a constant smoothness level suffices irrespective of~$(d_x,d_y)$. 
\end{remark}







\subsection{Characterization of and Approximation by Gaussian {\boldmath $k$}-SMI}

To gain further insight into the dependence of $k$-SMI on dimension, we fully characterize it in the Gaussian case. Afterwards, we show that general $k$-SMI decomposes into a Gaussian part plus a residual, and discuss conditions for the latter to decay as $d\to\infty$. As before, $\Sigma_X$ is~the covariance matrix of $X$ (similarly, for $Y$), while $\rC_{XY}:=\EE\big[(X-\EE[X])(Y-\EE[Y])^\tr\big]$ is the cross-covariance. 


\begin{theorem}[Gaussian $k$-SMI]\label{THM:Gaussian_ksmi} Let $(X,Y)\sim \gamma_{XY}=\cN(0,\Sigma_{XY})$ be jointly Gaussian random variables. Suppose that $\|\Sigma_X\|_{\op} \|\Sigma^{-1}_X\|_{\op}, \|\Sigma_Y\|_{\op} \|\Sigma^{-1}_Y\|_{\op}  \le \kappa$ and $\|\Sigma_X^{-1/2} \rC_{XY}\Sigma_Y^{-1}\|_{\op} \le \rho$ for some $\kappa \ge 1$ and $\rho < 1$. 
Then, for any fixed $k$, we have
\[
\ksmi(X;Y)  =  \frac{k^2 \| \rC_{XY} \|_\rF^2 }{ 2\Tr( \Sigma_X ) \Tr(\Sigma_Y)}\big(1+o(1)\big),
\]
as $d_x,d_y \to \infty$, where $o(1)$ denotes a quantity that converges to zero in the limit.
\end{theorem}
Theorem \ref{THM:Gaussian_ksmi} is proven in Appendix \ref{APPEN:Gaussian_decomp_proof}. It states that if $\Sigma_X$ and $\Sigma_Y$ have bounded condition numbers and the correlation, as quantified by $\|\Sigma_X^{-1/2} \rC_{XY}\Sigma_Y^{-1}\|_{\op}$, is less than 1, then the Gaussian $k$-SMI is asymptotically equivalent to the squared Frobenius norm $\rC_{XY}$, normalized by the traces of the marginal covariances. Since $\| \rC_{XY} \|_F^2 \le  (d_x \wedge d_y) \rho^2 \, \|\Sigma_X\|_{\op} \|\Sigma_Y\|_{\op}$~and $\Tr( \Sigma_X) \Tr(\Sigma_Y) \ge d_x d_y \|\Sigma_X^{-1} \|_{\op} \|\Sigma_{Y}^{-1} \|_{\op}$, we see that the $\ksmi(X;Y)$ typically decreases with dimension as $d_x^{-1} \wedge d_y^{-1}$. This rate is inline with the shrinkage with dimension of the MC bound from \eqref{EQ:MC_bound}, which renders that bound meaningful even when $k$-SMI is itself decaying, e.g., under the framework of Theorem \ref{THM:Gaussian_ksmi}. 


\paragraph{{\boldmath $k$}-SMI decomposition and Gaussian approximation.} Given the above result and the recent interest in Gaussian approximations of sliced Wasserstein distances \cite{nadjahi2021fast,rakotomamonjy2021statistical}, we present a decomposition of $k$-SMI into a Gaussian part plus a residual. For $(X,Y)\sim\mu_{XY}\in\cP(\RR^{d_x}\times\RR^{d_y})$, let $(X^*,Y^*)\sim\gamma_{XY}:=\cN(0,\Sigma_{XY})$ be jointly Gaussian with the same covariance as $(X,Y)$. The $k$-SMI satisfies
\begin{equation}
\ksmi(X;Y)=\ksmi(X^*;Y^*) + \EE\big[\delta_{XY}(A,B)\big],\label{EQ:kSMI_decomposition}
\end{equation}
where, for each $(\rA,\rB)\in\sti(k,d_x)\times \sti(k,d_y)$
\[\delta_{XY}(\rA,\rB) :=\dkl\big((\proj^\rA\mspace{-2mu},\proj^\rB)_\sharp\mu_{XY} \big\|(\proj^\rA\mspace{-2mu},\proj^\rB)_\sharp\gamma_{XY}\big)-\dkl\big((\proj^\rA\mspace{-2mu},\proj^\rB)_\sharp\mu_X\mspace{-2mu}\otimes\mspace{-2mu} \mu_Y \big\|(\proj^\rA\mspace{-2mu},\proj^\rB)_\sharp\gamma_X\mspace{-2mu}\otimes\mspace{-2mu} \gamma_Y\big).\] 
This decomposition is proven in Appendix \ref{APPEN:kSMI_decomposition_proof}. Theorem \ref{THM:Gaussian_ksmi} fully accounts for the first summand, which begs the questions of whether it is the leading term in the decomposition, and under what conditions? This question~is intimately related to the conditional CLT of low-dimensional projections under relative entropy \cite{reeves2017conditional}. This is a challenging and active research topic \cite{meckes2012approximation,meckes2012projections,reeves2017conditional}, for which sharp convergence rates remain unknown. As a first step towards a complete answer, in Appendix \ref{APPEN:residual_bound} we bound this residual term and identify mild isotropy conditions on the marginal distributions of $X$ and $Y$ that are sufficient for the residual term to vanish as $d_x,d_y\to\infty$.

\section{Experiments}

\begin{wrapfigure}{r}{0.6\textwidth}
\vspace{-4mm}
\centering
\tiny \hspace{.1in} Population $k$-SMI\hspace{.75in}  MC Standard Deviation \\
\vspace{-2mm}
\includegraphics[width=1.5in]{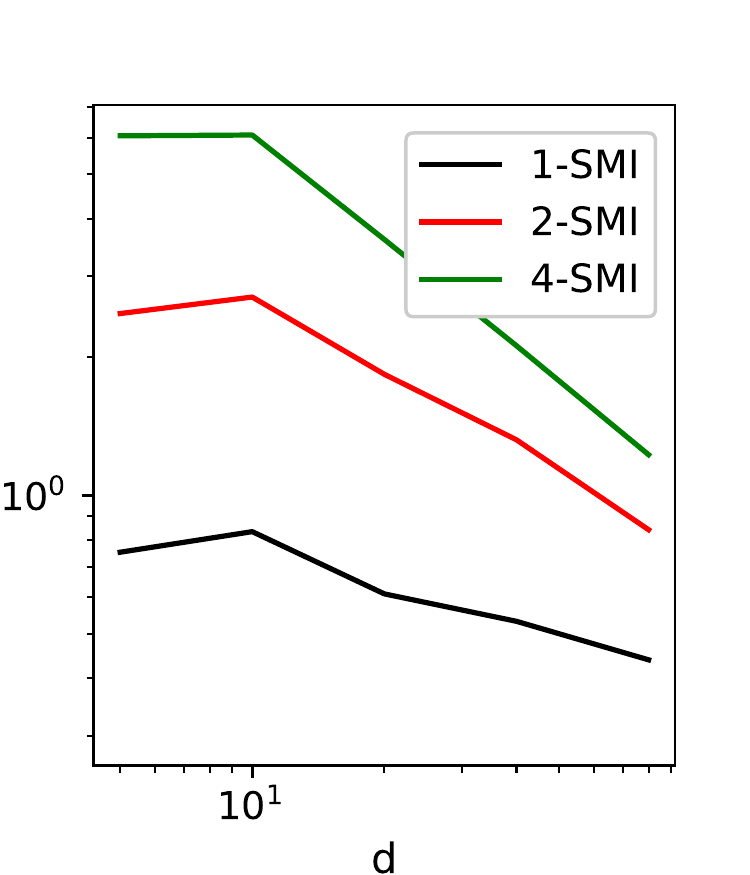}
\includegraphics[width=1.5in]{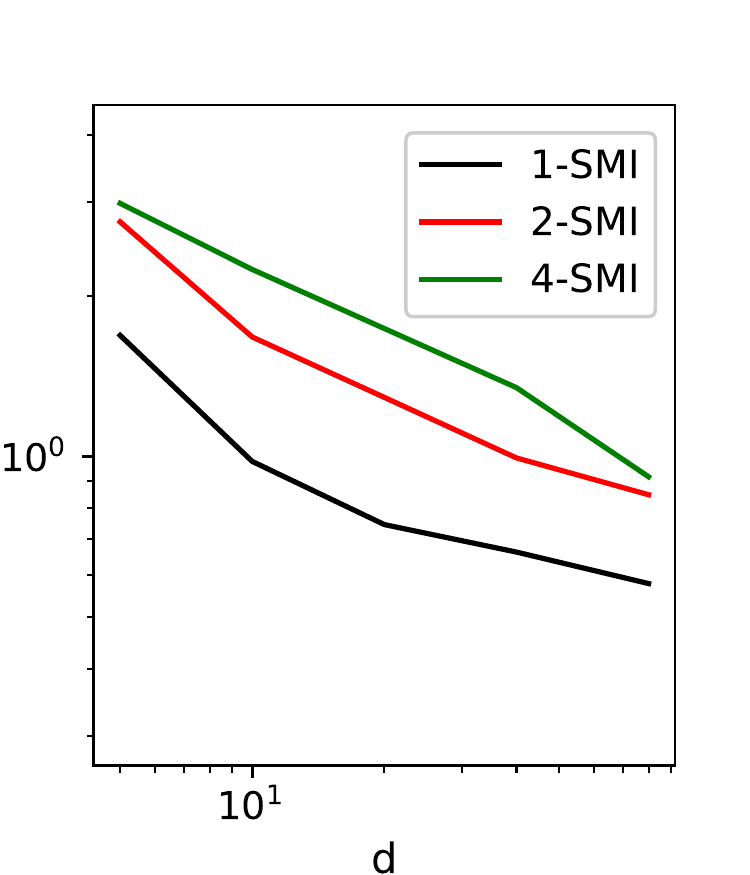}
\caption{Decay with dimension the population $k$-SMI (left) and the associated MC standard deviation (right).}
\label{FIG:MC}
\end{wrapfigure}

\paragraph{MC error and Gaussian {\boldmath $k$}-SMI rates.} Under the Gaussian setting described next, we illustrate the dependence on $k,d_x,d_y$ of (i) the population $k$-SMI expression in \cref{THM:Gaussian_ksmi}, and (ii) the associated MC estimation error from \cref{THM:MC_error}. Let $Z_1,Z_2 \sim \mathcal{N}(0,\mathrm{I}_d)$ and  $V \sim \mathcal{N}(0,\mathrm{I}_2)$ be independent, and set $X = \mathrm{P}_1 V + Z_1$ and $Y = \mathrm{P}_2 V + Z_2$, where $\mathrm{P}_1, \mathrm{P}_2\in\RR^{d\times 2}$ are projection matrices (with i.i.d. normal entries). We draw $m=10^3$ pairs of projection matrices $\{(\rA_j,\rB_j)\}_{j=1}^m$, and use the classic $k$-NN MI estimator of \cite{kraskov2004estimating} with $n=16\times 10^3$ samples of $(X,Y)$ to approximate the MI along each  projection pair, i.e., for each $j=1,\ldots,10^3$, we compute ${\sI}\big((\rA_j^\tr X)^n,(\rB_j^\tr Y)^n\big)$ . Note that the mean of ${\sI}\big((\rA_j^\tr X)^n,(\rB_j^\tr Y)^n\big)$ is the population $k$-SMI (which, in this Gaussian example, is given by \cref{THM:Gaussian_ksmi}), while its standard deviation is the constant in front of the $m^{-1/2}$ term in \eqref{EQ:MC_bound} of \cref{THM:MC_error}. Figure \ref{FIG:MC} plots the said mean and standard deviation of the projected MI terms ${\sI}\big((\rA_j^\tr X)^n,(\rB_j^\tr Y)^n\big)$. The rates of decay in both cases follow those predicted by Theorems \ref{THM:Gaussian_ksmi} and \ref{THM:MC_error}, respectively. This implies that $m$ need not be rapidly scaled up, even as the population $k$-SMI shrinks with increasing dimension.






\begin{figure}[htb]
\vspace{-3mm}
\centering


\begin{subfigure}[b]{.5\textwidth}
\begin{center}
\includegraphics[width=.95in]{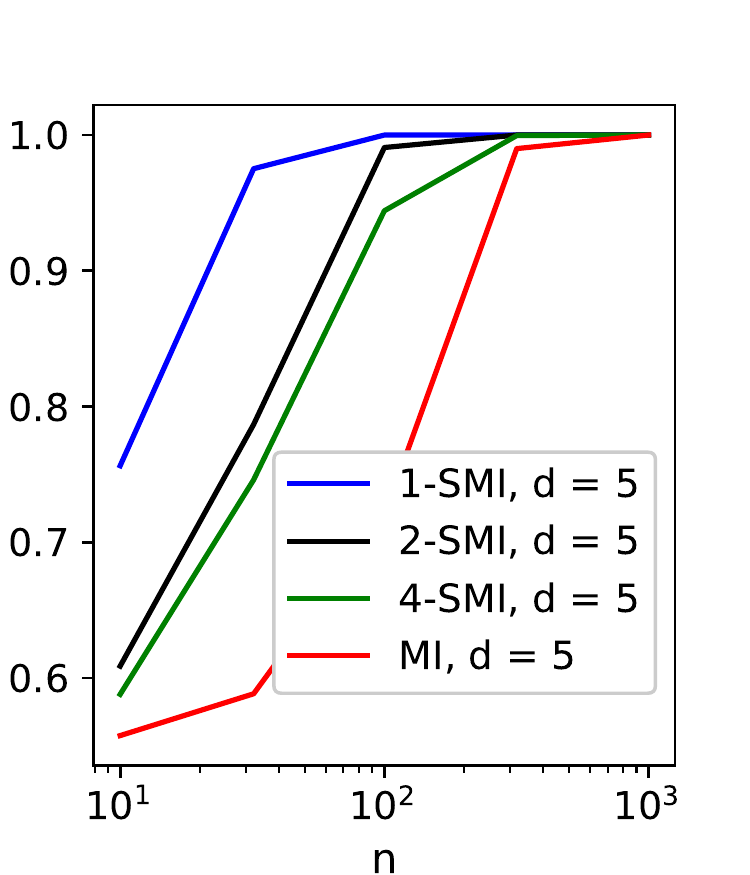}\hspace{-.1in}\includegraphics[width=.95in]{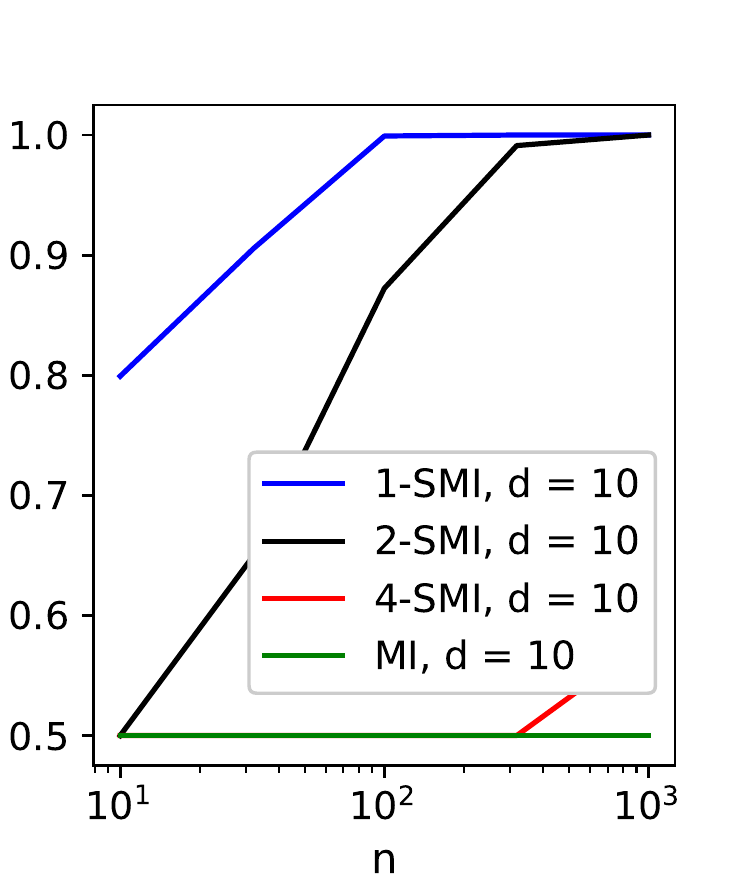}\hspace{-.1in}\includegraphics[width=.95in]{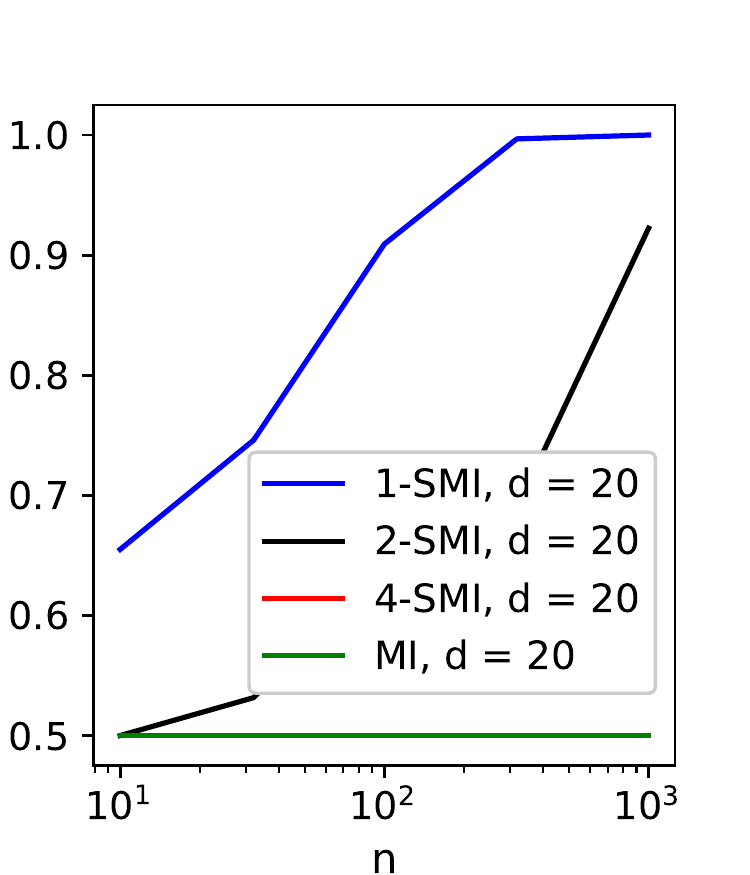}
\end{center}
\caption{$Y$ encodes single $X$ feature  (sinusoidal)}
\end{subfigure}
\begin{subfigure}[b]{.5\textwidth}
\begin{center}
\includegraphics[width=.95in]{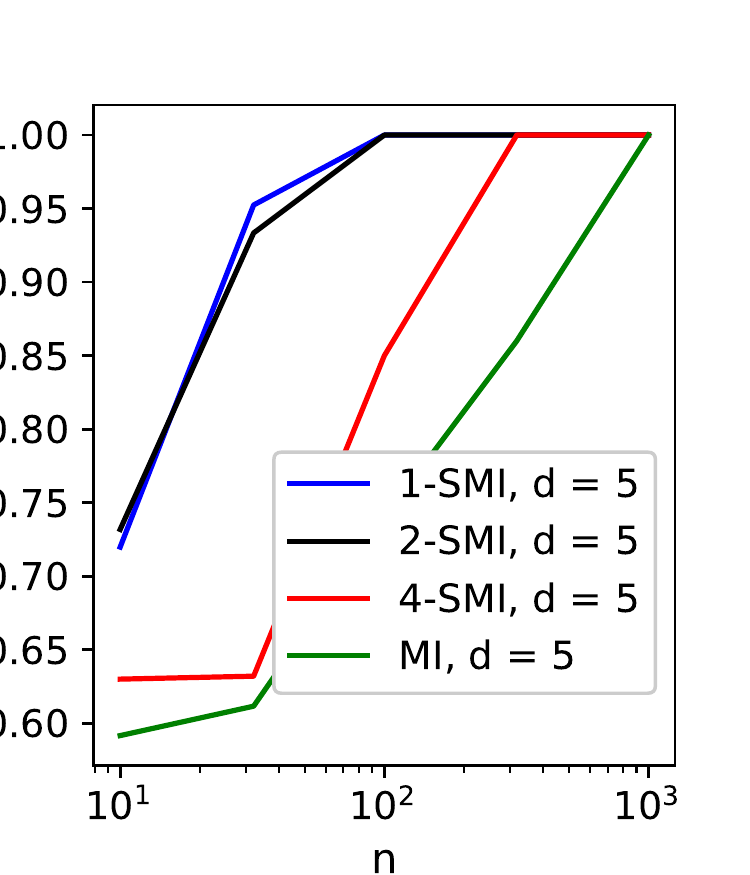}\hspace{-.1in}\includegraphics[width=.95in]{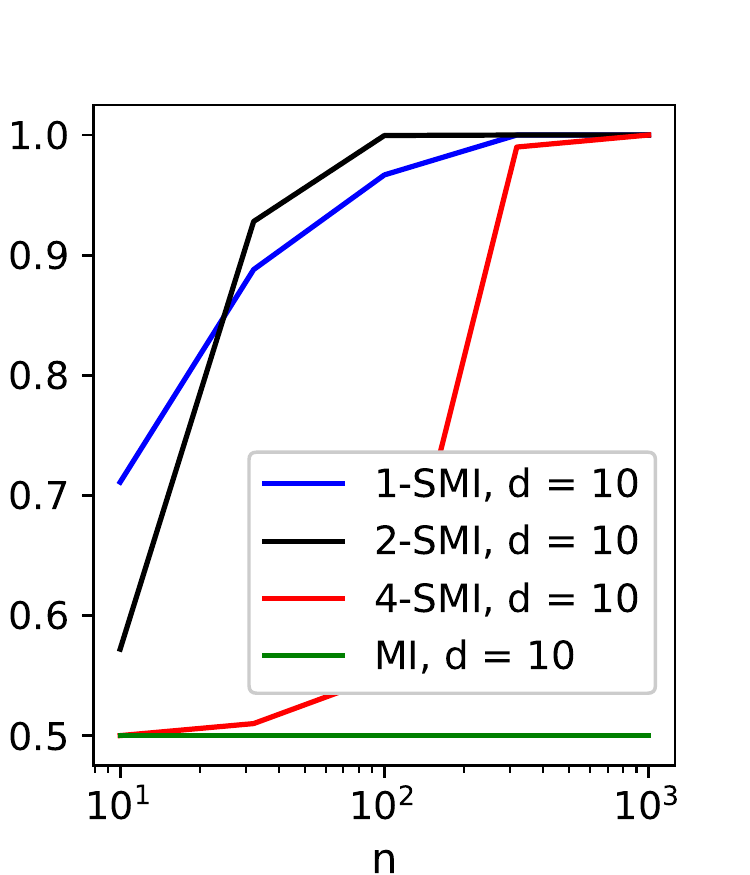}\hspace{-.1in}\includegraphics[width=.95in]{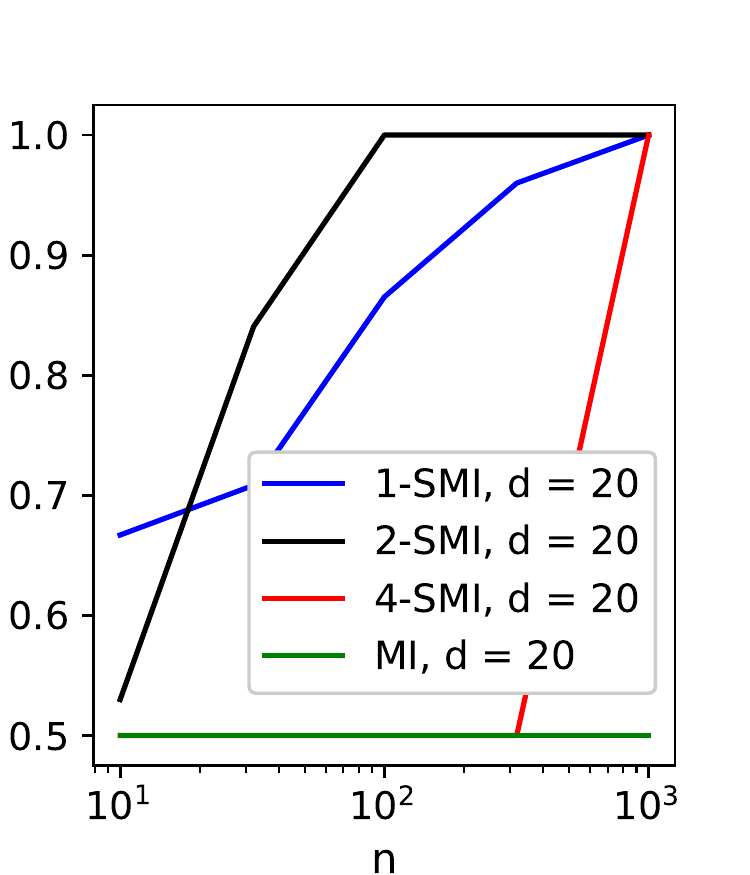}
\end{center}
\caption{Rank 2 common signal}
\end{subfigure}
\ \ \ \ 
\begin{subfigure}[b]{.5\textwidth}
\begin{center}
\includegraphics[width=.95in]{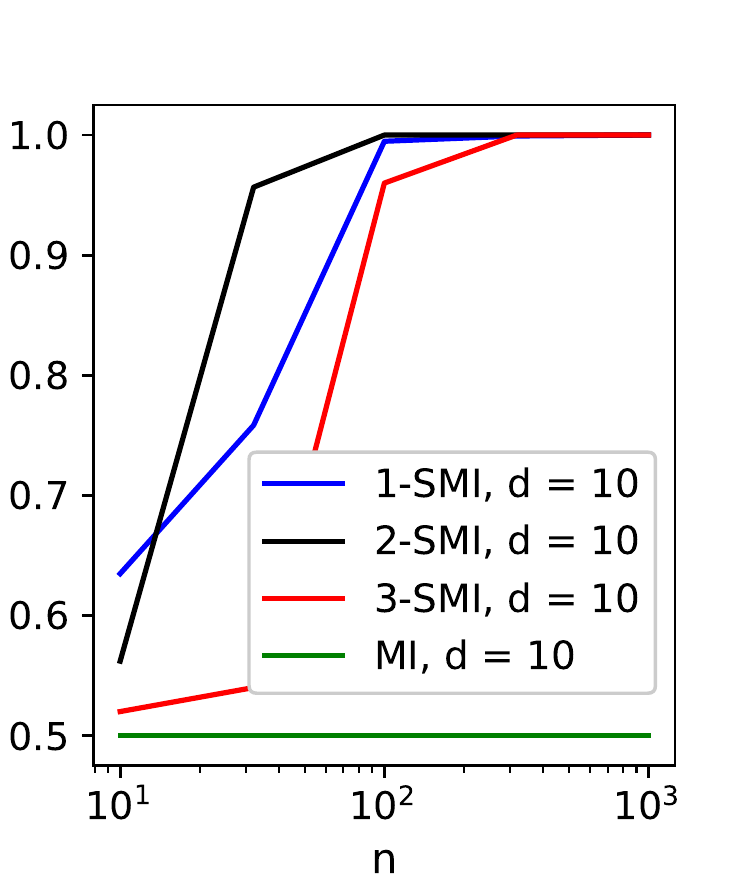}\hspace{-.1in}\includegraphics[width=.95in]{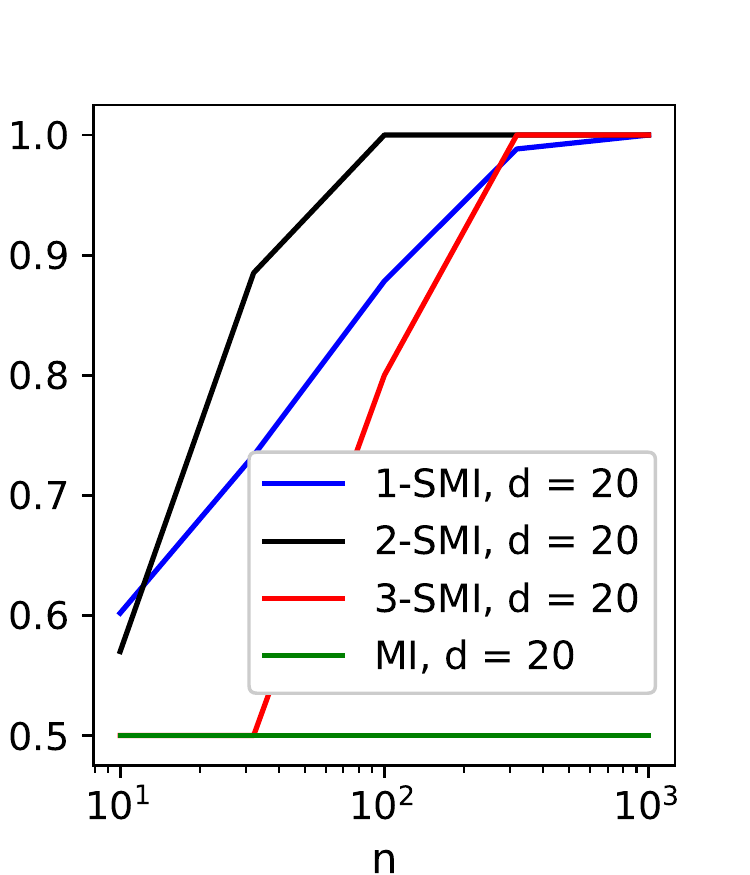}\hspace{-.1in}\includegraphics[width=.95in]{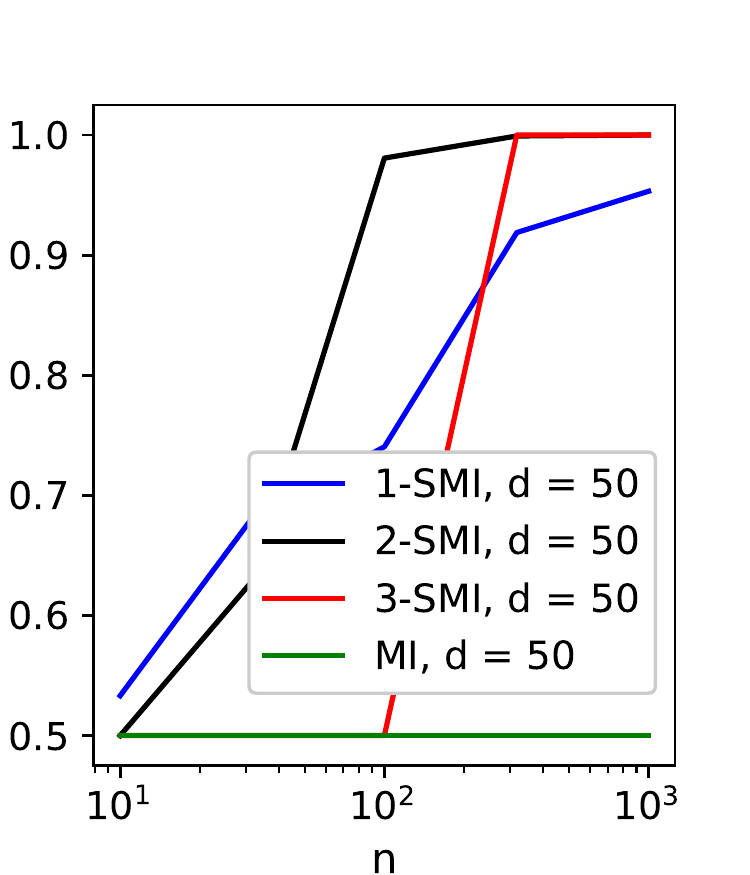}
\end{center}
\caption{Rank 3 common signal}
\end{subfigure}
\begin{subfigure}[b]{.5\textwidth}
\begin{center}
\includegraphics[width=.95in]{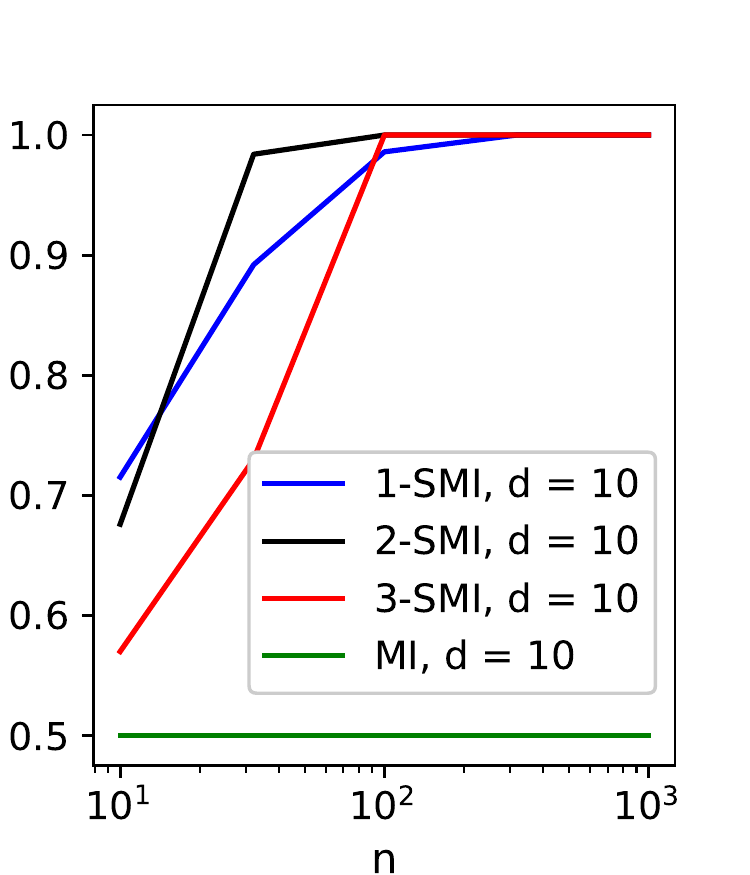}\hspace{-.1in}\includegraphics[width=.95in]{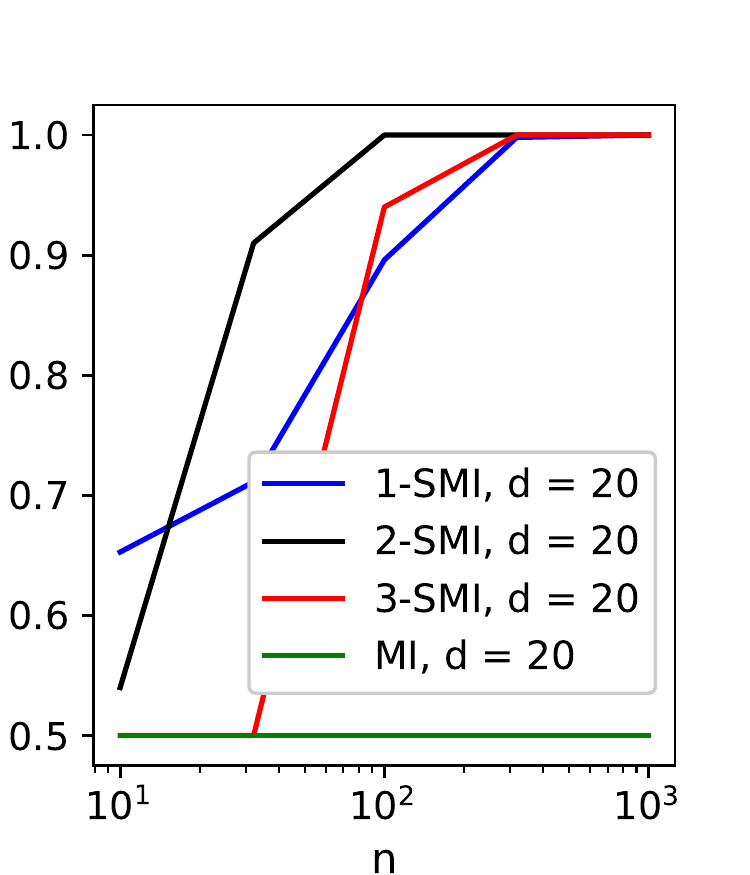}\hspace{-.1in}\includegraphics[width=.95in]{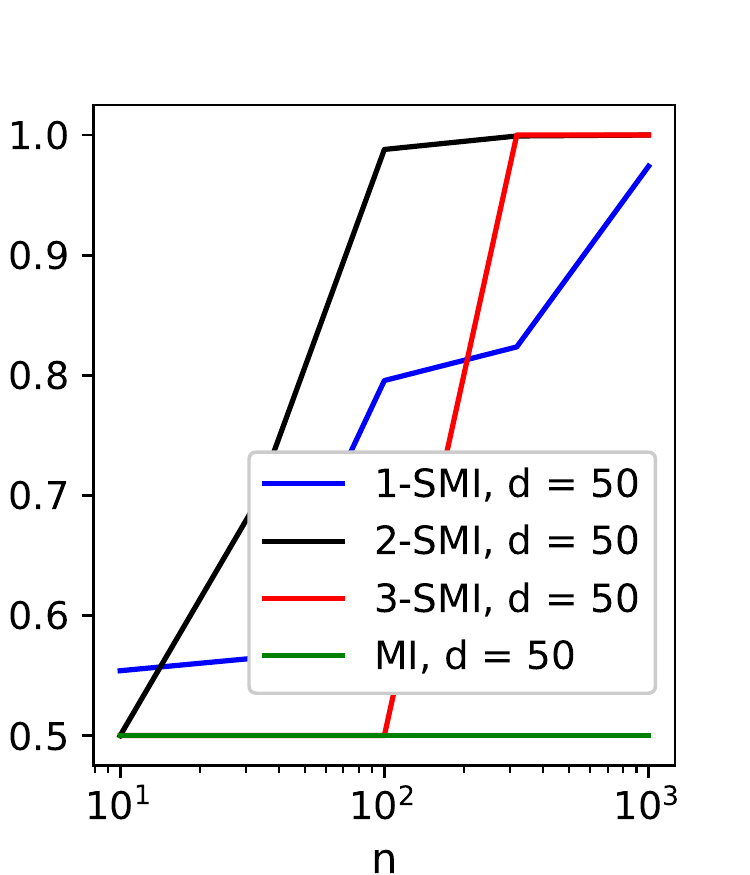}
\end{center}
\caption{Rank 4 common signal}
\end{subfigure}
\caption{Independence testing with $k$-SMI: AUC ROC versus sample size $n$ for different $k$ and $d$. 
}
\label{fig:IndepTest}
\end{figure}
\paragraph{Independence testing.} It was shown in \cite{goldfeld2021sliced} that SMI can be used for independence testing between high-dimensional variables, when classic MI is too costly to estimate. We revisit this experiment with $k$-SMI to demonstrate similar scalability and understand the effect of $k$. The test estimates $k$-SMI based on $n$ samples from $\mu_{XY}$ and then thresholds the value to declare dependence/independence. Two types of models for $(X,Y)$ are considered: (i) $X,Z \sim \mathcal{N}(0,\mathrm{I}_d)$ are independent and $Y = \frac{1}{\sqrt{2}}\big (\frac{1}{\sqrt{d}}\sin(\mathbf{1}^\intercal X)\mathbf{1}+Z\big)$ (i.e., $X$ and $Y$ share one sinusoidal feature), and (ii) the rank~2

\begin{figure}[t] 
\begin{center}
\footnotesize \hspace{.11in} $k=1$ \hspace{1.32in} $k=2$ \hspace{1.32in} $k=5$\\
\hspace{-.1in}\includegraphics[width=1.5in]{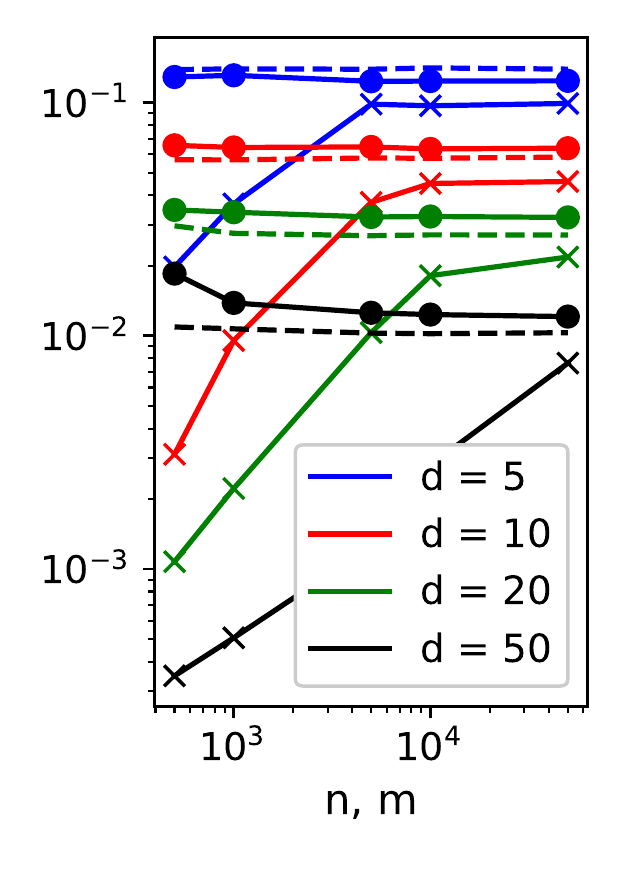}\hspace{.2in}\includegraphics[width=1.5in]{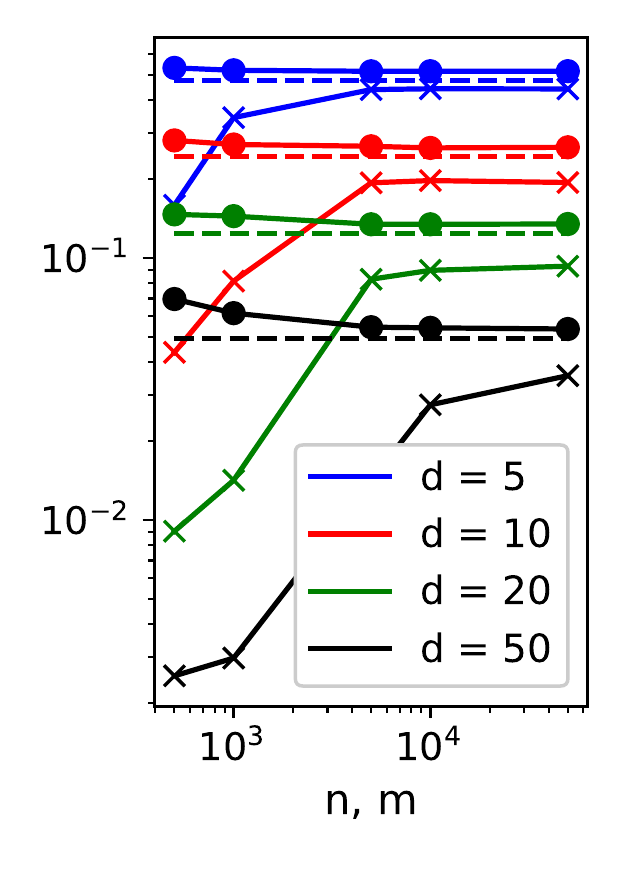}\hspace{.2in}\includegraphics[width=1.5in]{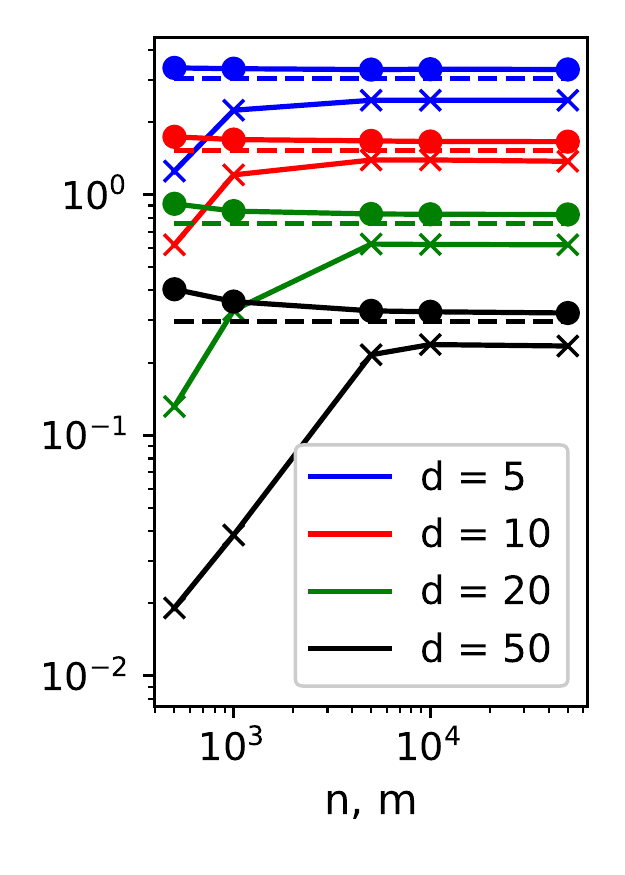}
\end{center}
\caption{Neural estimation rates: Dashed line shows the ground truth, circle line is the value of the parallel neural estimator from Section \ref{SUBSEC:SMI_NE}, and the cross line is the SMI neural estimator from~\cite{goldfeld2021sliced}. The parallel neural estimator converges at a faster rate for all considered $k$ and~$d$.}
\vspace{4mm}
\label{fig:Neural}
\end{figure}

common signal model from the previous paragraph, as well as its extension to ranks 3 and 4. Figure~\ref{fig:IndepTest} at the bottom of the previous page shows the area under the curve (AUC) of the receiver operating characteristic (ROC) as a function of $n$ for each of those models. Figure \ref{fig:IndepTest}(a) shows the results for Model (i), while Figures \ref{fig:IndepTest}(b)-(c) corresponds to Model (ii) with ranks 2, 3, and 4, respectively. The estimator $\widehat{\smi}_k^{m,n}$ from \eqref{EQ:SMI_est} is realized with $m=1000$ and $\hat{\sI}(\cdot,\cdot)$ as the Kozachenko–Leonenko estimator \cite{kraskov2004estimating}; the AUC ROC curves are computed from 100 random trials. For Figures 4(a) and 4(b), we vary the ambient dimension as $d=5,10,20$, while the projection dimension is $k=1,2,4,d$; note that $k=1$ corresponds to the SMI from \cite{goldfeld2021sliced} and $k=d$ to classic MI. In Figures 4(c) and 4(d) we consider, respectively, a common signal of rank 3 and 4. The ambient dimension is varied as $d=10,20,50$, while the projection dimension is $k=1,2,3,d$. Evidently, $k$-SMI-based tests perform well even when $d$ is large, while tests using classic MI fail. 1-SMI has a clear advantage in the model from Figure 4(a), where the common signal is 1-dimensional, but this is no longer the case for the models from Figures 4(b)-(d), where the shared structure is of higher dimension. Indeed, in Figure 4(b) we see that 2-SMI generally presents the best performance as it can better capture the underlying structure. For Figures 4(c) and 4(d), 3-SMI slightly outperforms 2-SMI for larger sample sizes, particularly in higher dimension. This highlights the potential gain of using higher $k$ values (to retain more information about the original signal, albeit at the cost of higher sample complexity) and the importance of adapting them to the intrinsic dimensionality of the model.

\paragraph{Neural estimation.}
Figure \ref{fig:Neural} (on the next page) illustrates the convergence of the $k$-SMI neural estimator\footnote{$m$ parallel 3-layer ReLU NNs were used, each with $30\cdot k$ hidden units in each layer.} from Section \ref{SUBSEC:SMI_NE} as $n = m$ increase together, for $X = Y \sim \mathcal{N}(0,\rI_d)$. For comparison, we include the original neural estimator of \cite{goldfeld2021sliced}, which uses a single neural net to approximate a shared DV potential.\footnote{A 3-layer ReLU NN was used with $20\cdot d$ hidden units in each layer.} While both neural estimators eventually converge to the ground truth, our parallel implementation converges much faster. Again note the clear decay of the true $k$-SMI as $d$ increases.

\paragraph{Sliced InfoGAN.} We demonstrate a simple application of $k$-SMI to modern machine learning. Recall the InfoGAN \cite{chen2016infogan}---a GAN variant that learns disentangled latent factors by maximizing a neural estimator of the MI between those factors and the generated samples. Figure \ref{fig:infogan}(left) shows InfoGAN results for MNIST,\footnote{Used experiment and code from \url{https://github.com/Natsu6767/InfoGAN-PyTorch}.} where 3 latent codes $(C_1,C_2,C_3)$ were used for disentanglement, with $C_1$ being a 10-state discrete variable and $(C_2,C_3)$ being continuous variables with values in $[-2,2]$. The shown images are generated by the trained InfoGAN, where each row of corresponds to a different values the discrete $C_1$, while columns corresponds to random $C_2,C_3$ values. Despite being completely unsupervised, $C_1$ has been successfully disentangled to encode the digits 0-9. Figure \ref{fig:infogan}(middle) shows the resulting generated images when the neural estimator for MI is replaced with a neural 1-SMI estimator with $m = 10^3$, and Figure \ref{fig:infogan}(right) for 5-SMI. 
Evidently, 1-SMI and 5-SMI successfully disentangle the latent factors, despite seeing only $10^3$ 1- (respectively 5-) dimensional projections of this very high-dimensional data.

\begin{figure}[t] 
\begin{center}
\tiny \hspace{-2mm} Original InfoGAN \hspace{1.1in} 1-Sliced InfoGAN\hspace{1.2in} 5-Sliced InfoGAN\\
\vspace{.5mm}
\includegraphics[width=1.55in]{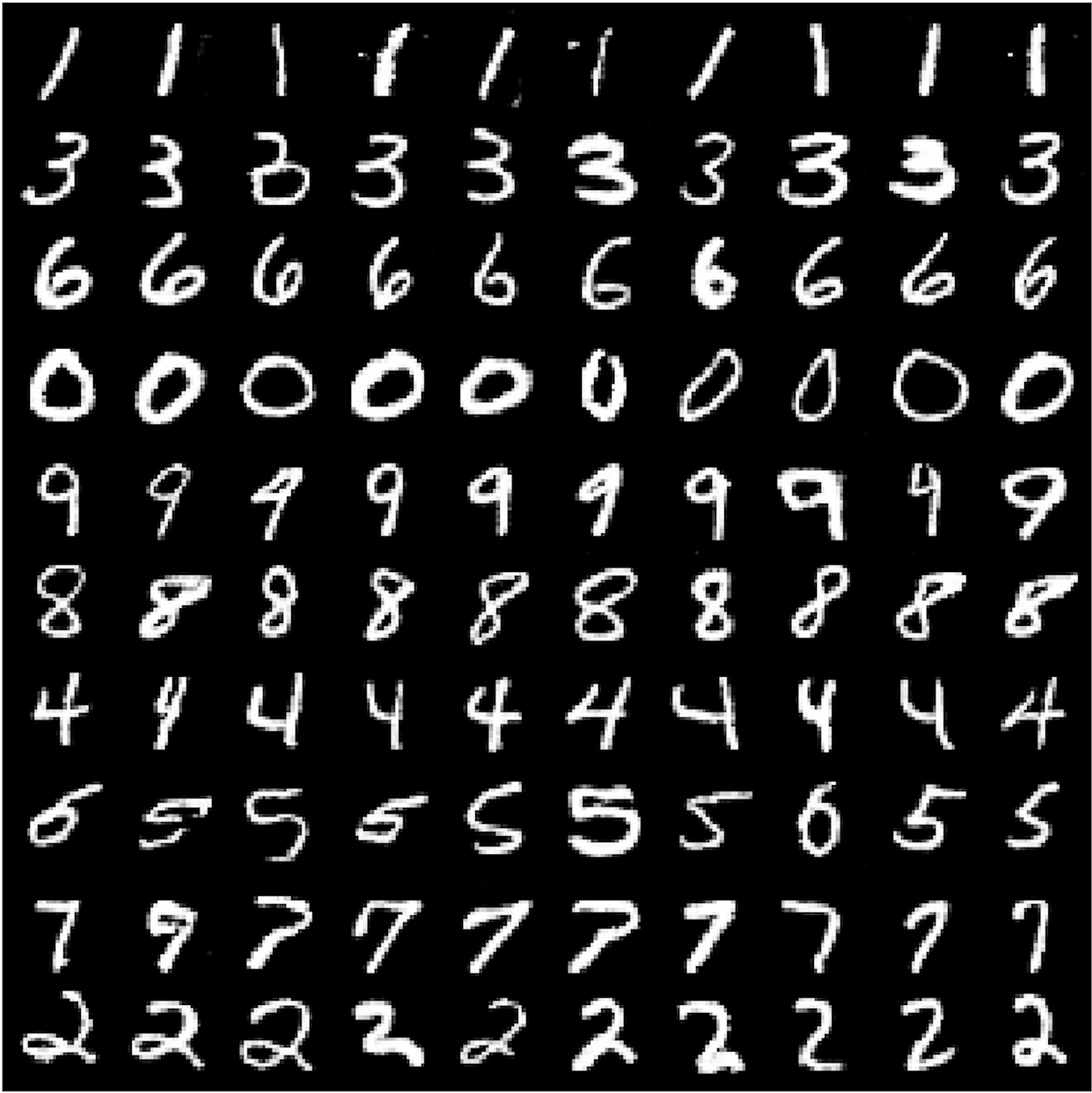}\hspace{0.25in}\includegraphics[width=1.55in]{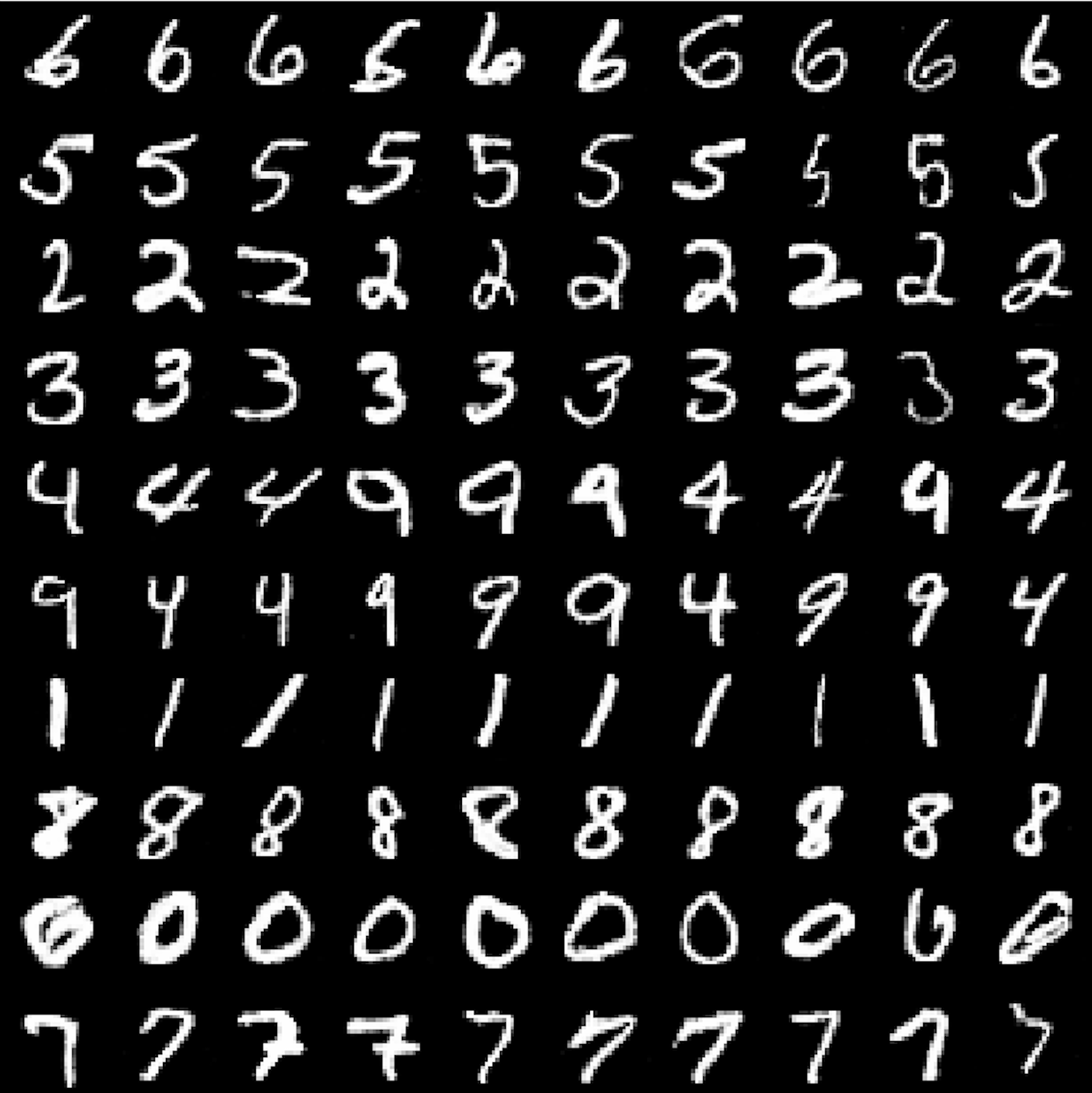}\hspace{0.25in}\includegraphics[width=1.55in]{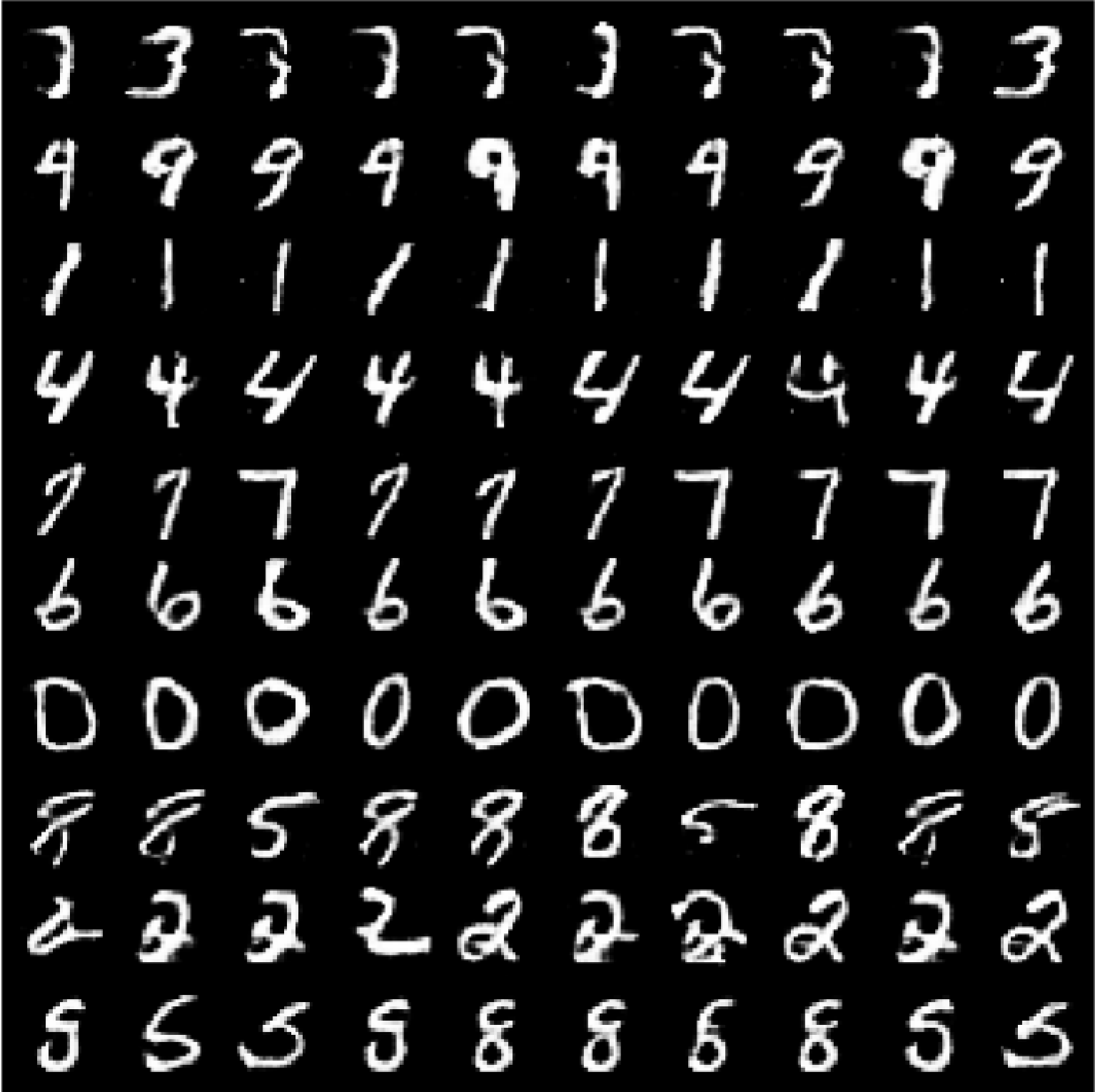}
\end{center}
\caption{MNIST images generated via InfoGAN using neural estimators of MI (left), 1-SMI (middle), and 5-SMI (right). The latent codes $C_1$ (encodes digits) is varied across rows, while columns correspond to (random) $C_2,C_3$ values. 
In all three cases, the latent codes are successfully disentangled.}
\label{fig:infogan}
\end{figure}

\section{Summary and Concluding Remarks}

This paper introduced $k$-SMI as a measure of statistical dependence defined by averaging MI terms between $k$-dimensional projections of the considered random variables. Our objective was to quantify and provide a rigorous justification for the perceived scalability of sliced information measures. We have done so by studying MC-based estimators of $k$-SMI, neural estimation methods, and asymptotics of $\ksmi(X;Y)$ under the Gaussian setting. Throughout, results with explicit dependence on $k,d_x,d_y$ were provided, revealing different gains associated with slicing, from the anticipated scalability to relaxed smoothness assumptions needed for neural estimation. Numerical experiments supporting our theory were provided, as well as a more advanced application to sliced infoGAN, showing that $k$-SMI can successfully replace classic MI even in applications with more intricate underlying structure. 

Future research directions, both theoretical and applied, are abundant. In particular, we seek to derive sharp rates of decay of the residual term in \eqref{EQ:kSMI_decomposition}, thereby establishing the Gaussian $k$-SMI as the leading term in that decomposition. Extensions of our results to the case when the projection dimensions for $X$ and $Y$ are different, i.e., $k_1\neq k_2$, may allow further flexibility and are also of interest. We also plan to explore non-linear dimensionality reduction maps, as in the generalized sliced Wasserstein distance setting \cite{kolouri2019generalized}, as well as non-uniform distributions over parameterizations of the projection functions (cf. \cite{nguyen2020distributional}). The max-SMI, where instead of averaging over $(\rA,\rB)$ we maximize over them, is another interesting avenue. 
On the application side, there are various machine learning models that utilize MI \cite{chen2016infogan,alemi2017deep,higgins2017beta,oord2018representation}; revisiting those with $k$-SMI is an appealing endeavor due to the expected gains from slicing and the formal guarantees our theory can provide for those systems.

\subsection*{Acknowledgements}
Z. Goldfeld is partially supported by NSF grants CCF-1947801,  CCF-2046018, and DMS-2210368, and the 2020 IBM Academic Award. G. Reeves is partially supported by NSF grant CCF-1750362.

\bibliographystyle{unsrt}
\bibliography{ref}

\appendix

\appendix

\section{Proofs of Results in the Main Text}

\subsection{Proofs for Section \ref{SUBSEC:continuity}}\label{APPEN:continuity_proofs} 

The HWI inequality of Otto and Villani \cite{otto2000generalization} is a functional inequality relating the entropy (H), quadratic transportation cost (W), and Fisher information (I), all defined w.r.t. a suitable reference measure that has bounded curvature. In deriving the classic result, a more general version of the HWI was established in \cite{otto2000generalization}---one that is particularly well suited to the application in this paper, where we consider the differences between two entropy terms. See also \cite[Proposition 1.5]{gentil2020entropic} for a recent derivation of the generalized inequality via a different argument based on an entropic interpolation of Wasserstein geodesics. 

The result reads as follows: let $\gamma_\sigma=\cN(0,\sigma^2\rI_d)$ denote the isotropic Gaussian measure on $\RR^d$ with variance $\sigma^2$ and consider $\mu, \nu \in \cP_2(\RR^d)$ with $\dkl(\nu\|\gamma_\sigma) < \infty$ for some $\sigma >0$,  then
\begin{align}
\dkl( \mu\| \gamma_\sigma) - \dkl( \nu\|\gamma_\sigma)  \le  \sW_2(\mu, \nu) \sqrt{ \sJ( \mu\|\gamma_\sigma)} - \frac{1}{2 \sigma^2 } \sW_2^2( \mu, \nu).  \label{eq:HWI} 
\end{align}
This HWI inequality is used to prove \cref{LEMMA:Wass_cont}, from which the Lipschitz continuity in \cref{PROP:ent_lip} readily follows.

\begin{proof}[Proof of \cref{LEMMA:Wass_cont}]
If $\mu$ has finite Fisher information then we have the well-known identities
\begin{align*}
\dkl( \mu \| \gamma_\sigma) &=  \frac{n}{2} \log( 2\pi \sigma^2 )  + \frac{1}{2 \sigma^2 } \EE_\mu \|X\|^2   - \sh( \mu ) \\
\sJ( \mu  \|  \gamma_\sigma) &= \sJ(\mu) - \frac{2n}{ \sigma^2} +  \frac{ 1 }{\sigma^4} \EE_\mu \| X\|^2.
\end{align*}
Making the change of variables $\lambda = \sigma^{-2}$ and swapping the roles of $\mu$ and $\nu$  leads to the following bound on the difference in differential entropy:
\[ \sh( \mu)  -  \sh( \nu) \le  \left(  \sJ(\nu)  - 2 n  \lambda   +  \lambda^2 \EE_\nu \|X\|^2\right)^{\frac 12} 
  \sW_2(\mu , \nu)    - \frac{\lambda}{2} \sW_2^2( \mu , \nu)+ \frac{\lambda}{2} \left( \EE_\mu \|X\|^2 - \EE_\nu \|X\|^2 \right). 
\]
As the left-hand side (LHS) does not depend on $\lambda$, by taking the $\lambda \to 0^+$ we obtain
\begin{equation}
 \sh( \mu)  -  \sh( \nu) \le  \sqrt{  \sJ(\nu) } \sW_2(\mu , \nu) \label{eq:HWI_simple}   
\end{equation}
as desired. To see that the constant cannot be improved, evaluate the above bound for $\mu=\gamma_a$ and $\nu=\gamma_b$, and consider the limiting case of $a/b\to 1^+$:
\[
\lim_{\frac{a}{b} \to 1^+}\frac{\sh(\gamma_a)-\sh(\gamma_b)}{\sqrt{\sJ(\gamma_b)}\sW(\gamma_a,\gamma_b)}=\lim_{\frac{a}{b} \to 1^+}\frac{\log\left(\frac{a}{b}\right)}{\left(\frac{a}{b}-1\right)}=1.
\]
\end{proof}


\begin{proof}[Proof of \cref{PROP:ent_lip}]
Because differential entropy is translation invariant we may assume without loss of generality that $\mu$ has zero mean. From the definition of the 2-Wasserstein distance, we obtain
\begin{equation}
\sW_2( \proj^A_\sharp \mu , \proj^B_\sharp \mu) \le \| \Sigma_{\mu}^{1/2} (A- B)\|_\rF.\label{EQ:W2_froben_bound}
\end{equation}
Combining this with \eqref{eq:HWI_simple} gives the first result.

To obtain a uniform bound on the Lipschitz constant, we use the fact that $\rJ_\rF(\proj^{\rB}_\sharp \mu) \preceq \rB^\tr \rJ_\rF( \mu) \rB$ for any matrix $\rB$ with orthogonal columns (cf. \cite[Equation (67)]{rioul2010information}). Thus, both $\|\rJ_\rF( \proj^\rA_\sharp\mu)\|_{\op}$ and $\|\rJ_\rF( \proj^\rB_\sharp\mu)\|_{\op}$ are bounded from above by the sum of the $k$ largest eigenvalues of $\rJ_\rF(\mu)$, and so
\[\|\rJ_\rF( \proj^\rA_\sharp\mu)\|_{\op}\vee\|\rJ_\rF( \proj^\rB_\sharp\mu)\|_{\op}\leq k\|\rJ_\rF(\mu)\|_{\op}.\]
Combining this with $\| \Sigma_{\mu}^{1/2} (A -B) \|_\rF \le \| \Sigma_{\mu}\|_{\op} \| A -B \|_\rF $ in \eqref{EQ:W2_froben_bound} completes the proof. 
\end{proof}

\subsection{Proof of \cref{PROP:kSMI_prop}}\label{APPEN:kSMI_prop_proof}

\paragraph{Proof of \ref{prpty:NonNeg}.} Non-negativity follows because $k$-SMI is an average of classic MI terms, which are non-negative. Nullification of $k$-SMI between independent $(X,Y)$ is also straightforward, since in this case $(\rA^\tr X,\rB^\tr Y)$ are independent for all $(\rA,\rB)\in\sti(k,d_x)\times\sti(k,d_y)$, which implies $\sI(\rA^\tr X;\rB^\tr Y)=0$, i.e., the integrand in the $k$-SMI definition is identically zero. For the opposite implication, as will be shown below, we have $\smi(X;Y)\leq \ksmi(X;Y)$, for any $1\leq k\leq d_x\wedge d_y$. Hence, if $\ksmi(X;Y)=0$ then $\smi(X;Y)=0$ and by Proposition 1 form \cite{goldfeld2021sliced} we have that $(X,Y)$ are independent.

\paragraph{Proof of \ref{prpty:Bounds}.} Throughout this proof we use our standard matrix notation (non-italic letter, such as $\rA$) to designate random matrices; for fixed matrices we add a tilde, e.g., $\tilde{\rA}$. Fix $1\leq k_1<k_2\leq d_x\wedge d_y$ and let $(\rA_1,\rB_1)\sim\sigma_{k_1,d_x}\otimes\sigma_{k_1,d_y}$ and $(\rA_2,\rB_2)\sim\sigma_{k_2,d_x}\otimes\sigma_{k_2,d_y}$. For each $\tilde{\rA}_2\in\sti(k_2,d)$, represent it as $\tilde{\rA}_2=[\tilde{\rA}_{21}\  \tilde{\rA}_{22}]$, where $\tilde{\rA}_{21}\in\sti(k_1,d)$ and $\tilde{\rA}_{22}\in\sti(k_2-k_1,d)$, and similarly for $\tilde{\rB}_2$. We now have 
\begin{align*}
\smi_{k_2}&(X;Y)\\
&=\sI(\rA_2^\tr X;\rB_2^\tr Y|\rA_2,\rB_2)\\
&=\sI(\rA_{21}^\tr X;\rB_{21}^\tr Y|\rA_2,\rB_2)+\sI(\rA_{22}^\tr X;\rB_2^\tr Y|\rA_2,\rB_2,\rA_{21}^\tr X)+\sI(\rA_{21}^\tr X;\rB_{22}^\tr Y|\rA_2,\rB_2,\rB_{21}^\tr Y)\\
&\geq \sI(\rA_1^\tr X;\rB_1^\tr Y|\rA_1,\rB_1)\\
&=\smi_{k_1}(X;Y),
\end{align*}
where the inequality uses the non-negativity of (conditional) MI and the fact that 
\[
\sI(\rA_{21}^\tr X;\rB_{21}^\tr Y|\rA_2,\rB_2)=\sI(\rA_{21}^\tr X;\rB_{21}^\tr Y|\rA_{21},\rB_{21})=\sI(\rA_1^\tr X;\rB_1^\tr Y|\rA_1,\rB_1).
\]
Indeed, the latter holds since $(\rA_{22},\rB_{22})$ are marginalized out in the conditioning and because $(\rA_{21},\rB_{21})\stackrel{d}{=}\rA_1,\rB_1)\sim\sigma_{k_1,d_x}\otimes \sigma_{k_1,d_y}$.

Lastly, supremizing the integrand in the $k$-SMI definition over all pair of matrices from the Stiefel manifold, we further obtain
\[\smi_{k_2}(X;Y)=\sI(\rA_2^\tr X;\rB_2^\tr Y|\rA_2,\rB_2)\leq \sup_{(\tilde{\rA},\tilde{\rB})\in\sti(k_2,d_x)\times\sti(k_2,d_y)}\sI(\tilde{\rA}^\tr X;\tilde{\rB}^\tr Y),\]
which concludes the proof.

\paragraph{Proof of \ref{prpty:KL}.} This follows because conditional mutual information can be expressed as \[\s{I}(X;Y|Z)=\EE_{\mu_Z}\Big[\s{D}_{\s{KL}}\big(\mu_{X,Y|Z}(\cdot|Z)\big\|\mu_{X|Z}(\cdot|Z)\otimes \mu_{Y|Z}(\cdot|Z)\big)\Big],\]
and because the joint distribution of $(\rA^\tr X,\rB^\tr Y)$ given $\{\rA=\tilde{\rA},\rB=\tilde{\rB}\}$, for fixed $(\tilde{\rA},\tilde{\rB})\in\sti(k,d_x)\times\sti(k,d_y),$ is $(\proj^{\tilde{\rA}},\proj^{\tilde{\rB}})_\sharp \mu_{X,Y}$, while the corresponding conditional marginals are $\proj^{\tilde{\rA}}_\sharp \mu_X$ and $\proj^{\tilde{\rB}}_\sharp \mu_Y$, respectively. Hence,
\[
    \ksmi(X;Y)=\dkl\big((\proj^\rA,\proj^\rB)_\sharp \mu_{XY}\big\|(\proj^\rA,\proj^\rB)_\sharp \mu_X\otimes\mu_Y\big| \rA,\rB\big)=\dkl\big(\mu_{\rA,\rB,\rA^\tr X,\rB^\tr Y}\big\|\mu_{\rA,\rB,\rA^\tr \tilde{X},\rB^\tr \tilde{Y}}\big)
\]
where the second step follows from the relative entropy chain rule, with
\begin{align*}
    (\rA,\rB,\rA^\tr X,\rB^\tr Y)&\sim \mu_{\rA,\rB,\rA^\tr X,\rB^\tr Y}=\sigma_{k,d_x}\otimes \sigma_{k,d_y}(\proj^\rA,\proj^\rB)_\sharp \mu_{XY}\\
    (\rA,\rB,\rA^\tr \tilde{X},\rB^\tr \tilde{Y})&\sim \mu_{\rA,\rB,\rA^\tr \tilde{X},\rB^\tr \tilde{Y}}=\sigma_{k,d_x}\otimes \sigma_{k,d_y}(\proj^\rA,\proj^\rB)_\sharp \mu_X\otimes \mu_Y.
\end{align*}
The variational form follows by applying the DV representation of relative entropy to the latter expression for $\ksmi(X;Y)$ (see Section \ref{SUBSEC:SMI_NE}).

\paragraph{Proof of \ref{prpty:entropy}.} Recall the definition of marginal and conditional $k$-sliced entropies: $\ksh(X):=\sh(\rA^\tr X|\rA)$ and $\ksh(X|Y):=\sh(\rA^\tr X|\rA,\rB,\rB^\tr Y)$. Given the representation of $k$-SMI as a conditional mutual information, we now have
\begin{align*}
    \ksmi(X;Y)=\sI(\rA^\tr X;\rB^\tr Y|\rA,\rB)=\sh(\rA^\tr X|\rA)-\sh(\rA^\tr X|\rA,\rB,\rB^\tr Y),
\end{align*}
where we have used independence of $\rB$ and $(\rA,\rA^\tr X)$ in the first conditional entropy term. The other decompositions follow in a similar fashion.


\paragraph{Proof of \ref{prpty:Chain}.} We only prove the small chain rule; generalizing to $n$ variables is straightforward.~Consider:
\begin{align*}
\ksmi(X,Y;Z)&=\sI(\rA^\tr X,\rB^\tr Y;\rC^\tr Z|\rA,\rB,\rC)\\
&=\sI(\rA^\tr X;\rC^\tr Z|\rA,\rB,\rC)+\sI(\rB^\tr Y;\rC^\tr Z|\rA,\rB,\rC,\rA^\tr X),
\end{align*}
where the last equality is the regular chain rule. Since $(X,Z,\rA,\rC)$ are independent of $\rB$, we have
\[\sI(\rA^\tr X;\rC^\tr Z|\rA,\rB,\rC)=\sI(\rA^\tr X;\rC^\tr Z|\rA,\rC)=\ksmi(X;Z),\]
while $\sI(\rB^\tr Y;\rC^\tr Z|\rA,\rB,\rC,\rA^\tr X)=\ksmi(Y;Z|X)$ by definition.

\paragraph{Proof of \ref{prpty:Tensor}.} By definition, 
\begin{align*}
    \ksmi(X_1,\mspace{-1mu}\dots\mspace{-1mu}, X_n; Y_1,\mspace{-1mu}\dots\mspace{-1mu}, Y_n) \mspace{-2mu}=\mspace{-2mu} \ksmi(\rA_1^\tr X_1,\mspace{-1mu}\dots\mspace{-1mu}, \rA_n^\tr X_n; \rB_1^\tr Y_1,\mspace{-1mu}\dots\mspace{-1mu}, \rB_n^\tr Y_n| \rA_1, \mspace{-1mu}\dots\mspace{-1mu}, \rA_n, \rB_1,\mspace{-1mu}\dots\mspace{-1mu}, \rB_n),
\end{align*}
where the $\rA_i$, $\rB_i$ are all independent and uniform on the respective Stiefel manifolds. Now, by mutual independence of the $\rA_i$, $\rB_i$ and $(X_i,Y_i)$ across $i$ and tensorization of MI, we have
\begin{align*}
\sI(\rA_1^\tr X_1,\dots, \rA_n^\tr X_n; \rB_1^\tr Y_1,\dots, \rB_n^\tr Y_n| \rA_1, \dots, \rA_n, \rB_1,\dots, \rB_n)&= \sum_{i=1}^n \s{I}(\rA_i^\intercal X_i; \rB_i^\tr Y_i|\rA_i, \rB_i)\\
&= \sum_{i=1}^n\ksmi(X_i;Y_i).
\end{align*}
\qed

\subsection{Proof of \cref{THM:MC_error}}\label{APPEN:MC_error_proof}
The proof of \cref{THM:MC_error} relies on the following technical lemmas concerning the Lipschitzness and variance of the function $i_{XY}:\sti(k,d_x)\times\sti(k,d_y)\to \RR$ defined as $i_{XY}(\rA,\rB):=\sI(\rA^\tr X;\rB^\tr Y)$.

\begin{lemma}[Lipschitzness of projected MI]\label{LEMMA:Lipschitzness}
For $\mu_{XY}\in\cP_2(\RR^{d_x}\times\RR^{d_y})$ with $\sJ(\mu_{XY})<\infty$, the function $i_{XY}:\sti(k,d_x)\times\sti(k,d_y)\to \RR$ is Lipschitz with respect to the Frobenius norm on the Cartesian product of Stiefel manifolds, with Lipschitz constant
\[
L_k(\mu_{XY})=3\sqrt{2k\|\rJ_\rF(\mu_{XY})\|_{\op}\big(\|\Sigma_X\|_{\op}\vee\|\Sigma_Y\|_{\op}\big)}.
\]
\end{lemma}

\begin{proof}
 Fixing $(\rA_1,\rB_1),(\rA_2,\rB_2)\in \sti(k,d_x)\times\sti(k,d_y)$, we have
\begin{align*}
    \big|i_{XY}&(\rA_1,\rB_1)-i_{XY}(\rA_2,\rB_2)\big|\\
    &\leq \big|\sh(\rA_1^\tr X)-\sh(\rA_2^\tr X)\big|+\big|\sh(\rB_1^\tr Y)-\sh(\rB_2^\tr Y)\big|+\big|\sh(\rA_1^\tr X,\rB_1^\tr Y)-\sh(\rA_2^\tr X,\rB_2^\tr Y)\big|.\numberthis\label{EQ:joint_ent}
\end{align*}
The differences of marginal entropy terms (i.e., the first two terms on the right-hand side (RHS) above) are controlled by $\sqrt{k \|\Sigma_X\|_{\op} \|\rJ_\rF(\mu_X)\|_{\op}} \|\rA_1-\rA_2\|_\rF$ and $\sqrt{k \|\Sigma_Y\|_{\op} \|\rJ_\rF(\mu_Y)\|_{\op}} \|\rB_1-\rB_2\|_\rF$, respectively, by applying \cref{PROP:ent_lip}. For the difference of joint entropies, we shall use \cref{LEMMA:Wass_cont}. To that end, note that
\begin{align*}
    \sW_2^2\big((\proj^{\rA_1},\proj^{\rB_1})_\sharp\mu_{XY},(\proj^{\rA_2},\proj^{\rB_2})_\sharp\mu_{XY}\big)&\leq\EE\left[\|(\rA_1-\rA_2)^\tr X\|^2+\|(\rB_1-\rB_2)^\tr Y\|^2\right]\\
    &\leq \big(\|\Sigma_X\|_{\op}\vee\|\Sigma_X\|_{\op}\big)\big(\|\rA_1-\rA_2\|^2_{\rF}+\|\rB_1-\rB_2\|^2_{\rF}\big),\numberthis\label{EQ:W2_bound}
\end{align*}
and observe that the Fisher information of the projected joint distribution can be controlled by the operator norm of the corresponding Fisher information matrix. Indeed, the Fisher information data processing inequality (cf. \cite[Equation (67)]{rioul2010information}) states that for any $(\rA,\rB)\in\sti(k,d_x)\times\sti(k,d_y)$, we have
\begin{equation}
    \sJ\big((\proj^{\rA},\proj^{\rB})_\sharp\mu_{XY}\big)=\Tr\left(\left[ \begin{array}{cc}
     \rA^\tr & 0 \\
     0 & \rB^\tr \\
   \end{array} \right]\rJ_\rF(\mu_{XY}) \left[ \begin{array}{cc}
     \rA & 0 \\
     0 & \rB \\
   \end{array} \right]   \right)\leq 2k\|\rJ_\rF(\mu_{XY})\|_{\op}.\label{EQ:FI_op_norm}
\end{equation}
Invoking \cref{LEMMA:Wass_cont}, while using the above along with \eqref{EQ:W2_bound},~gives
\begin{align*}
    \big|\sh(\rA_1^\tr X,\rB_1^\tr Y)&-\sh(\rA_2^\tr X,\rB_2^\tr Y)\big|\\
    &\leq \Big(2k\|\rJ_\rF(\mu_{XY})\|_{\op}\big(\|\Sigma_X\|_{\op}\vee\|\Sigma_Y\|_{\op}\big)\Big)^{1/2}\big(\|\rA_1-\rA_2\|^2_{\rF}+\|\rB_1-\rB_2\|^2_{\rF}\big)^{1/2}.
\end{align*}
Together with the marginal entropy bounds the fact that $\rJ_\rF(\mu_X)\vee\rJ_\rF(\mu_Y)\preceq \rJ_\rF(\mu_{XY})$ (which also follows from the data processing inequality), this implies the result.
\end{proof}

\begin{lemma}[Variance bound]\label{LEMMA:variance}
Let $(\rA,\rB)\sim\sigma_{k,d_x}\otimes\sigma_{k,d_y}$, then we have the variance bound
\[
\Var\big( i_{XY}(\rA,\rB)\big)\leq 24L^2_k(\mu_{XY})\left(d_x^{-1}+d_y^{-1}\right),
\]
where $L_k(\mu_{XY})$ is defined in \cref{LEMMA:Lipschitzness}.
\end{lemma}

\begin{proof}

Recall that the special orthogonal group $\mathbb{SO}(d) = \{\rU \in \RR^{d \times d} : \rU^\tr \rU = \rI_d,  \det(\rU) = 1\}$ is the set of $d \times d$ orthogonal matrices with determinant one.  The following result is consequence of concentration of measure on compact Riemannian manifolds (see Section~5 in \cite{anderson2010introduction}).

\begin{lemma} 
Let $f \colon \mathbb{SO}(d)  \to \RR$ be Lipschitz continuous with respect to the Frobenius norm with Lipschitz constant $L$, i.e., $|f(\rU)- f(\rV)| \le L \|\rU-\rV\|_\rF$ for all $\rU,\rV \in  \mathbb{SO}(d)$. If $d \ge 3$ and $\rU$ is distributed uniformly on $\mathbb{SO}(d) $ then $f(\rU)$ is sub-Gaussian with parameter $\sigma^2 =  4 L^2 /(d-2)$, i.e., 
\begin{align*}
\log \EE\left[  e^{\lambda (f(\rU) - \EE f(\rU))}\right] \le \frac{\lambda^2 \sigma^2}{2 }, \quad\forall \lambda>0.
\end{align*}
In particular, this implies that $\Var\big(f(\rU)\big) \le \sigma^2$.
\end{lemma}

For our purposes, this result provides concentration bounds with respect to functions defined on the Stiefel manifold. Observe that if $\rU = [u_1 \ldots u_d]$ is uniform on $\mathbb{SO}(d)$ then the $d \times k$ matrix $\rA =[u_1 \ldots u_k]$ is uniform on  $\mathrm{St}(k,d)$. Thus, if $g$  is a real-valued function on  $\mathrm{St}(k,d) $ that is Lipschitz continuous with constant $L$,  we can apply the above result to $f(\rU) = g( \rU [\rI_k, 0]^\tr)$ to conclude that $g(\rA)$ is sub-Gaussian with parameter $4 L^2 /(d-2)$, and hence 
\begin{equation}
    \Var\big(g(\rA)\big)\leq \frac{4 L^2}{d-2}.\label{EQ:var_bound}
\end{equation}

Now, to bound the variance of $i_{XY}$, recall that $(\rA,\rB)\sim\sigma_{k,d_x}\otimes\sigma_{k,d_y}$ are independent and uniformly distributed random matrices from the corresponding Stiefel manifold. By the Efron-Stein inequality  (cf. e.g., \cite[Theorem 3.3.7]{raginsky2013concentration}), the variance satisfies
\begin{align*}
\Var\big(i_{XY}(\rA,\rB)\big) \le \EE\big[\Var\big(i_{XY}(\rA,\rB)\big| \rB\big)\big]+ \EE\big[\Var\big(i_{XY}(\rA,\rB)\big|\rA\big)\big].
\end{align*}
Since $i_{XY}(\cdot, \cdot)$ is Lipschitz continuous in each of its arguments with the same constant, it follows from Lemma~\ref{LEMMA:Lipschitzness} that the terms on the RHS are bounded from above by $(2L_k(\mu_{XY}))^2/(d_x-2)$ and $(2L_k(\mu_{XY}))^2/(d_y-2)$, respectively. 

As the above requires $d_x,d_y>2$, we further note that for $\rA'$ an independent copy of $\rA$, we have 
\begin{align*}
    \EE\big[\Var\big(i_{XY}(\rA,\rB)\big| \rB\big)\big]\mspace{-1.5mu}&= \mspace{-1.5mu}
\frac{1}{2} \EE\left[ \big| i_{XY}(\rA,\rB) -  i_{XY}(\rA',\rB) \big|^2\right] \\
&\le  \frac{L_k^2(\mu_{XY})}{2} \EE \|\rA_1 - \rA_2\|^2_\rF \leq L_k^2(\mu_{XY}) \mspace{1.5mu} k d_x.
\end{align*}
Since $k \le d_x - 1$, it follows that
\[
\EE\big[\Var\big(i_{XY}(\rA,\rB)\big| \rB\big)\big] \le \left( d^2_x - d_x \wedge \frac{ 4}{ d_x - 2}  \right) L_k^2(\mu_{XY}) \le  \frac{12 L_k^2(\mu_{XY})}{d_x},
\]
and similarly for $\EE\big[ \Var\big(i_{XY}(\rA,\rB)\big|\rA\big)\big]$ with $d_y$ replacing $d_x$. The conclusion of \cref{LEMMA:variance} follows.
\end{proof}




\begin{proof}[Proof of Theorem \ref{THM:MC_error}]

Since $k$-SMI is invariant to translation (due to bijection invariance of MI), we may assume without loss of generality that $X$ and $Y$ are centered. The error is now decomposed as
\begin{align*}
  &\EE\left[ \left| \ksmi(X;Y) - \widehat{\smi}_k^{m,n} \right|\right]\\
  &\qquad\qquad\qquad\leq \EE\left[ \left| \ksmi(X;Y) - \frac{1}{m} \sum_{i=1}^{m} i_{XY}(\rA_i,\rB_i) \right|\right]+\EE\left[ \left|  \frac{1}{m} \sum_{i=1}^{m} i_{XY}(\rA_i,\rB_i) - \widehat{\smi}_k^{m,n} \right|\right].
\end{align*}

The first term on the RHS above corresponds to the MC error. By observing that $\ksmi(X;Y)=\EE\left[\frac{1}{m} \sum_{i=1}^{m} i_{XY}(\rA_i,\rB_i)\right]$ and using monotonicity of moments, we may upper bound it by $\big(\Var\big(i_{XY}(\rA,\rB)\big)/m\big)^{1/2}$. \cref{LEMMA:variance} then provides a bound on the variance.

The second term above is controlled by the $k$-dimensional MI estimation error $\delta_k(n)$ from \cref{ASSUMP:scalar_MI}, since 
\[
    \EE\left[ \left|  \frac{1}{m} \sum_{i=1}^{m} i_{XY}(\rA_i,\rB_i) - \widehat{\smi}_k^{m,n} \right|\right]\mspace{-2mu}\leq\mspace{-6mu} \sup_{\substack{\rA\in\sti(k,d_x)\\\rB\in\sti(k,d_y)}}\EE\left[\left|  i_{XY}(\rA,\rB) - \hat{\sI}\big((\rA^\tr X)^n,(\rB^\tr Y)^n\big)\right|\right]\leq \delta_k(n).
\]
Combining the two bounds produces the result.
\end{proof}

\subsection{Proof of \cref{THM:SMI_NE}}\label{APPEN:SMI_NE_proof}

The proof utilizes the result of Theorem 4 from \cite{sreekumar2022neural} for relative entropy neural estimation along with the sufficient conditions given in Proposition~7 therein (cf. \cite[Section 4.1.1]{sreekumar2022neural} for comments on the applicability of their Theorem 4 to the DV variational form). For completeness, we first restate those~results. Denote $\|\cZ\|_\infty := \sup_{z \in \cZ} \|z\|_{\infty}.$

\begin{proposition}[Sufficient conditions for relative entropy neural estimation (Theorem~4 and Proposition~2 of \cite{sreekumar2022neural})] \label{Prop:approximation error KL densities properties}
Fix $d,b,M \geq 0$ and set $s = \lfloor{d/2}\rfloor +3,$. Let $\cZ \subset \RR^d$ be compact, and $\mu,\nu \in \cP_{\s{ac}}(\cZ)$ 
have densities $f_\mu,f_\nu$ respectively. 
Suppose that $\dkl(\mu \|\nu) \leq M$ and that there exist $r_\mu, r_\nu \in \sC^s_{b}(\cU)$ for some open set $\cU \supset \cZ$, such that  $\log f_\mu= r_\nu|_{\cZ}$ and $\log f_\mu= r_\nu|_{\cZ}$. Then
\[\EE\left[\left|  \dkl(\mu \|\nu) - \hat{\dkl}_{\cG_{d, d}^\ell} (X^n, Y^n) \right| \right] \lesssim_{M,b,\|\cZ\|_\infty } d^{\frac{1}{2}}l^{-\frac{1}{2}} +d^{\frac{3}{2}} n^{-\frac{1}{2}},\]
where $\hat{\dkl}_{\cG_{d, d}^\ell} (X^n, Y^n) :=\sup_{g \in \cG_{d, d}^\ell} \frac 1n \sum_{i=1}^n g(X_i,Y_i)- \log\left(\frac 1n \sum_{i=1}^n e^{g(X_i,Y_{\sigma(i)})}\right)$. 
\end{proposition}

We use the above result to establish the following lemma that accounts for neural estimation of each projected MI term. Given the lemma, the result of \cref{THM:SMI_NE} follows by \cref{THM:MC_error}, with the RHS of \eqref{EQ:NE_unif_bound} in place of the $\delta_k(n)$ term therein.

\begin{lemma}[Neural estimation of  $i_{XY}(\rA,\rB)$] \label{LEM:Approximation of i_XY}
Let $\mu_{XY} \in \cF^k_{d_x,d_y}(M,b)$. Then uniformly in $(\rA, \rB) \in \sti(k,d_x) \times\sti(k,d_y)$, we have the neural estimation bound \begin{equation}
\EE\left[\left|  i_{XY}(\rA,\rB) - \hat{\sI}^{\ell}_{k,k}\big((\rA_j^\intercal X)^n,(\rB_j^\intercal Y)^n\big) \right| \right] \lesssim_{M,b,k,\|\cX \times \cY\|} k^{\frac 12} \ell^{-\frac{1}{2}} + k^{\frac 32}n^{-\frac{1}{2}}.\label{EQ:NE_unif_bound}
\end{equation}
\end{lemma}

\begin{proof}



The lemma is proven by showing that densities of $(\proj^\rA,\proj^\rB)_\sharp\mu_{XY}$ and $\proj^\rA_\sharp\mu_X \otimes \proj^\rB_\sharp\mu_Y$ satisfy the conditions of  \cref{Prop:approximation error KL densities properties}, whenever $\mu_{XY} \in  \cF^k_{d_x,d_y}(M,b)$.

Let $f$ be the density of $\mu_{XY}$ and set $f_{\theta}$, with $\theta:=(\theta_1, \theta_2) =  (\rA,\rB) \in \sti(k,d_x)\times\sti(k,d_y)$ as the density of projection $(\proj^\rA,\proj^\rB)_\sharp\mu_{XY}$ which is supported on $\cZ$. 
Let $\rA=[a_1 \ldots a_k]$, where $a_i \in \unitsphx$ with $\langle a_i,a_j\rangle =0, \ \forall i \neq j$, and denote $\cW_x=\{w \in \RR^{d_x}:\langle a_i,w\rangle=0,\,\forall i=1,\ldots,k\}$. Similarly, for $\rB=(b_1 \ldots b_k)$, set $\cW_y=\{w \in \RR^{d_y}: \langle b_i,w \rangle =0,\,\forall i=1,\ldots,k\}$. 

The density $f_{\theta}$ is given by
\[f_{\theta}(z_x,z_y)=\int_{\cW_x}  \int_{\cW_y}  f(\rA z_x +w_x, \rB z_y +w_y) \ dw_x \ dw_y,
\]
where we have denoted $z_{x,i}= \langle a_i, x \rangle$ and $z_{y,i}= \langle b_i, y \rangle$, for $i =1,\ldots,k$, and further defined $z_x=[z_{x,1} \ldots z_{x,k}]^\tr$  and $z_y=[z_{y,1} \ldots z_{y,k}]^\tr$ 

Given $\mu_{XY} \in  \cF^k_{d_x,d_y}(M,b)$,  there exists $r \in \sC^{s}_{b}(\cU)$ with $s=k+3$ for some open set $\cU \supset \cX \times \cY$, such that  $\log f= r|_{\cX \times \cY}$.
Choose $\cU^{'} \supset \cZ$ such that $\cU^{'}$ is the projection of the set $\cU$ on to the projection directions specified by $\rA, \rB$. Then also set 
 \[r_1(z_x,z_y)= \log \int_{\cW_x}  \int_{\cW_y}  \exp{\big(r(\rA z_x +w_x, \rB z_y +w_y) \big)} \ dw_x \ dw_y. \]
 which implies $r_1|_{\cZ}= \log f_{\theta}$.
 
 


To evaluate the derivative, we use the short hand notation $r_1 := \log  \left(\int \exp(r)\right)$, omitting the arguments of the functions $r,r_1$. Let $v \in \{z_{x,1}, \ldots, z_{x,k}, z_{y,1}, \ldots z_{y,k}  \}$ and $u\in \{x_1,\ldots, x_{d_x}, y_1,\ldots, y_{d_y} \}$, and consider

\begin{align*}
    \frac{\partial^{s}}{\partial^{s}v}r_1 & \stackrel{(a)}= \sum_{\cP_m^s} \frac{s!}{m_1! \ m_2! \ldots m_s!} \frac{(-1)^{M_{s}-1} \ (M_s-1)!}{ \big(  \int \exp(r) \big)^{M_s}}  \prod_{i=1}^{s} \frac{1}{(i!)^{m_i}} \Big(  \int \frac{\partial^{i}}{\partial^{i}v}\exp(r) \Big)^{m_i} \\ 
    & \stackrel{(b)}= \sum_{\cP_m^s} \frac{s!}{m_1! \ m_2! \ldots m_s!} \frac{(-1)^{M_{s}-1}\ (M_s-1)!}{ \big(\int \exp(r) \big)^{M_s}} \\
     & \quad \quad \quad  \quad \quad  \quad \times\prod_{i=1}^{s} \frac{1}{(i!)^{m_i}} \Big( \int \exp(r) \sum_{\cP_l^i} \frac{i!}{l_1! \ l_2! \ldots l_i!} 
    \prod_{k=1}^{i} \frac{1}{(k!)^{l_k}} \Big(  \frac{\partial^{k}}{\partial^{k}v}r \Big)^{l_k} 
    \Big)^{m_i} \\ 
    & \stackrel{(c)} \leq  \sum_{\cP_m^s} \frac{s!  \ (M_s-1)!}{m_1! \ m_2! \ldots m_s!} \prod_{i=1}^{s} \frac{1}{(i!)^{m_i}}  \left(\frac{ \int \exp(r) \sum_{\cP_l^i} \frac{i!}{l_1! \ l_2! \ldots l_i!} 
    \prod_{k=1}^{i} \frac{1}{(k!)^{l_k}} b^{l_k}   }{ \int \exp(r) }\right)^{m_i}  \\
    &  =  \sum_{\cP_m^s} \frac{s!  \ (M_s-1)!}{m_1! \ m_2! \ldots m_s! } \prod_{i=1}^{s} \frac{1}{(i!)^{m_i}}  \left( \sum_{\cP_l^i} \frac{i!}{l_1! \ l_2! \ldots l_i!} 
    \prod_{k=1}^{i} \frac{1}{(k!)^{l_k}} b^{l_k}   \right)^{m_i} \\
    & \stackrel{(d)} \leq c_s (b \vee b^s) 
\end{align*}
where:\\
(a) follows from Fa\`a di Bruno's formula with $M_s=\sum_{i=1}^s m_i$ and $\cP_m^s$ as the set of all $s$-tuples of non-negative integers $m_i$ satisfying $\sum_{i=1}^s i m_i=s$;\\ 
(b) uses the Fa\`a di Bruno's formula for the function $\exp(r)$, with $\cP_l^i$ defined similarly to~$\cP_m^s$; \\ 
(c) holds since  $|\frac{\partial^{k}}{\partial v^{k}} r \big | \leq b$, which comes from the fact that 
$\big |\frac{\partial^{k}}{\partial v^{k}} r \big | \leq |\frac{\partial^{k}}{\partial u^{k}} r \big | \leq b$ for $k \leq s$; the latter is a consequence of $r$ being $s$-times differentiable with derivatives bounded by $b$ and since $\big |\frac{\partial^k}{\partial v^k} u \big| \leq 1$, which holds because $x=\rA z_x +w_x, y=\rB z_y +w_y$ and thus $\frac{\partial}{\partial v} u$ is a constant (i.e., independent of $v$) upper bounded 1;\\ 
(d) identifies the dominating term as $b^{\sum_{i=1}^s \sum_{k=1}^i l_k m_i} \leq b \vee b^s$ and uses $c_s$ for a constant that depends only on $s$.

Conclude that $r_1 \in \sC^s_{b^\star}(\cU^{'})$ with $b^\star=c_s(b \vee b^{s})$.

Consider a similar derivation for the product of marginal densities. Let $f_{\theta_1}$ and $f_{\theta_2}$ denote the densities of $\proj^\rA_\sharp\mu_X$ and $\proj^\rB_\sharp\mu_Y $, respectively; the corresponding supports are $\cZ_1$ and $\cZ_2$, for which $\cZ= \cZ_1 \times \cZ_2$. Following steps as above, we can show that $\exists \ r_{\theta_1} \in \sC^s_{b^\star}(\cU^{'}_1), r_{\theta_2} \in \sC^s_{b^\star}(\cU^{'}_2) $ with $ \cU^{'}_1 \supset \cZ_1, \cU^{'}_2 \supset \cZ_2,$ such that $\log f_{\theta_1} ={r_{\theta_1}}|_{\cZ_1} $  and $\log f_{\theta_2} ={r_{{\theta_2}}}|_{\cZ_2}$. 

As the density of $\proj^\rA_\sharp\mu_X \otimes \proj^\rB_\sharp\mu_Y$ is $f_{\theta_1} f_{\theta_2}$, we choose $r_2(z_x,z_y)=r_{\theta_1} (z_x) + r_{\theta_2}(z_y)$. Accordingly, $\log f_{\theta_1} f_{\theta_2} = r_2|_{\cZ}$, and for $\cU^{'} =\cU^{'}_1 \times \cU^{'}_2 \supset \cZ$, we have
\[\|D^\alpha  r_2 \|_{\infty,\cU^{'}} 
     \leq \|D^\alpha  r_{\theta_1} \|_{\infty,\cU^{'}_1} +\|D^\alpha r_{\theta_2} \|_{\infty,\cU^{'}_2} \leq 2b^\star.
\]
This implies that $r_2\in \sC^s_{2b^\star}(\cU^{'})$, whereby
$\proj^\rA_\sharp\mu_X \otimes \proj^\rB_\sharp\mu_Y \in \cF^k_{d_x,d_y}(M,2b^\star)$. 

Since $\cZ \subseteq \RR^{2k}$ and $\mu_{XY} \in \cF^k_{d_x,d_y}(M,b)$, the above shows that  $(\proj^\rA,\proj^\rB)_\sharp\mu_{XY}$ and $\proj^\rA_\sharp\mu_X \otimes \proj^\rB_\sharp\mu_Y$ satisfy the smoothness requirement of 
\cref{Prop:approximation error KL densities properties} (the order should be at least $s = k+3$), with an expansion of smoothness radius to $2c_{k+3} (b \vee b^{k+3})$. 
For $k=1$ which corresponds to SMI, the expanded smoothness radius is $2 b^\star =154 (b \vee b^4)$.

Lastly, we note that
$ \| \cZ \|_\infty  \leq  \sup\limits_{(x,y) \in \cX \times \cY} \|\rA^\tr x, \rB^\tr y\| \leq  \| \cX \times \cY \|, $
where the last inequality is due to sub-multiplicative property of $\ell^2$-norm and
\[\left[ {\begin{array}{cc} 
    \rA^\tr& 0 \\
    0 & \rB^\tr \\
  \end{array} } \right] \left[ {\begin{array}{cc} 
    \rA& 0 \\
    0 & \rB \\
  \end{array} } \right] =\rI_{2k},\]
which results in the corresponding operator norm being 1. This completes the proof of \cref{LEM:Approximation of i_XY}. 
\end{proof}

\subsection{Proof of Theorem \ref{THM:Gaussian_ksmi}}\label{APPEN:Gaussian_decomp_proof}
We begin by recalling the setting of Theorem~\ref{THM:Gaussian_ksmi} as well as some basic properties of mutual information for Gaussian distributions. Let $(X,Y)\sim \gamma_{XY}=\cN(0,\Sigma_{XY})$ be jointly Gaussian random variables with positive definite covariance matrix
\begin{align*}
\Sigma_{XY} = \begin {pmatrix} \Sigma_X &  \rC_{XY} \\  \rC_{XY}^\tr &  \Sigma_Y \end{pmatrix} 
\end{align*}
The assumption that the covariance is positive definite means that the singular values of the correlation matrix defined by   $\rR := \Sigma_X^{-1/2} \rC_{XY} \Sigma_Y^{-1/2}$ are strictly less than one. The mutual information between $X$ and $Y$  depends only on the correlation matrix and is given by
 \begin{align*}
 \sI(X;Y) = - \frac{1}{2} \log \det( \Id_{d_x} - \rR \rR^\tr). 
 \end{align*}
Moreover, for a $d_x \times k$ matrix $\rA$ and $d_y \times k$ matrix  $\rB$, both with linearly independent columns, the mutual information between the $k$-dimensional Gaussian variables $\rA^\tr X$ and $\rB^\tr Y$ equals to
\begin{align*}
 \sI(\rA^\tr X;\rB^\tr Y) = - \frac{1}{2} \log \det( \Id_{k} - \tilde{\rR}\tilde{\rR}^\tr) , 
\end{align*}
where $\tilde{\rR} = \tilde{\rA}^\tr \rR \tilde{\rB}$ is the correlation matrix of the projected distribution and 
\begin{align}
\tilde{\rA} =  \Sigma_X^{1/2}\rA (\rA^\tr \Sigma_X \rA )^{-1/2} , \qquad \tilde{\rB} = \Sigma_Y^{1/2}\rB (\rB^\tr \Sigma_Y \rB )^{-1/2} \label{eq:AtBt}
\end{align}
The $k$-SMI is the expectation of this mutual information with respect to $(\rA,\rB)$ drawn from the uniform distribution on $ \mathrm{St}(k,d_x) \times \mathrm{St}(k,d_y)$ 
 
  \begin{remark}
If $\Sigma_X$ and $\Sigma_Y$  are approximately low rank then $\tilde{\rA}$ and $\tilde{\rB}$ are concentrated low-dimensional subspaces, which may or may not align with the dominant directions in the correlation matrix $\rR$. Therefore, in contrast to the mutual information, the $k$-SMI depends not only on the correlation matrix $\rR$ but also the marginal distributions of $X$ and $Y$. 
 \end{remark}


\begin{proof}[Proof of Theorem~\ref{THM:Gaussian_ksmi}]
The proof relies on several technical lemmas whose statements and proofs are deferred to the next section. The $k$-SMI for jointly Gaussian variables can be expressed as
\begin{align}
\ksmi(X,Y)  =  -  \frac{1}{2} \EE\left[ \log \det ( \Id_k - \tilde{\rR} \tilde{\rR}^\tr  )  \right]
\end{align} 
where $\tilde{\rR} =  \tilde{\rA}^\tr \rR \tilde{\rB}$ is the projected correlation matrix and $(\tilde{\rA}, \tilde{\rB})$ are defined as in \eqref{eq:AtBt} as a function of matrices $(\rA,\rB)$ drawn from the uniform distribution on $\mathrm{St}(k,d_x)\times \mathrm{St}(k,d_x)$. Note that  $\tilde{\rA}$ and $\tilde{\rB}$ are both on the Stiefel manifold,  and thus  $\|\tilde{\rA}\|_{\op} = \|\tilde{\rB}\|_{\op} = 1$. Accordingly, the  correlation matrix satisfies $\|\tilde{\rR} \|_{\op} \le  \|\rR \|_{\op} \le \rho$ a.s. Applying Lemma~\ref{lem:logdet_bnd} (see next section) to the positive definite matrix $\tilde{\rR}\tilde{\rR}^\tr$ and then taking expectation yields 
\begin{align*}
0 \le  \ksmi(X,Y)  -    \frac{1}{2}   \EE  \| \tilde{\rR}\|_\rF^2     \le  \frac{ \EE \| \tilde{\rR}\tilde{\rR}^\tr  \|^2_\rF  
}{2 (1- \rho^2) } 
\end{align*}
To establish the desired result we will characterize the leading order terms in $\EE  \| \tilde{\rR}\|_\rF^2$ and then show that the ratio  between $\EE \| \tilde{\rR}\tilde{\rR}^\tr  \|^2_\rF$ and $\EE  \| \tilde{\rR}\|_\rF^2$ converges to zero in the $d_x,d_y \to \infty$ limit. 

By the independence of $\tilde{A}$ and $\tilde{B}$, the expected squared Frobenius norm expands as 
\begin{align}
\EE  \|  \tilde{\rR} \|_\rF^2  & =\EE \,  \Tr\left(  \tilde{\rA}\tilde{\rA}^\tr  \rR \tilde{\rB} \tilde{\rB}^\tr \rR^\tr \right)   = \Tr\left( \EE [  \tilde{\rA}\tilde{\rA}^\tr ]   \rR \,   \EE [ \tilde{\rB} \tilde{\rB}^\tr] \rR^\tr\right). \label{eq:R_F} 
\end{align}
The matrices $\tilde{\rA}\tilde{\rA}^\tr$ and $\tilde{\rB}\tilde{\rB}^\tr$ are orthogonal projection matrices whose nonozero eigenvalues are equal to one. In the special case where $\Sigma_X$ and $\Sigma_Y$ are isotropic (i.e.\ proportional to the identity matrix), these matrices are distributed uniformly on the space of projection matrices of rank $k$. In the non-isotropic setting, however, these matrices are biased towards the directions in the covariances with large eigenvalues. An explicit expression for theirs means is provided in Lemma~\ref{lem:Pmean}, and simplified bounds are given in Lemma~\ref{lem:P_approx}, which shows that for all $\epsilon >0$, there exits a number $d = d(\epsilon, \kappa, k)$ such that for all $d_x, d_y \ge d$, we can write
\begin{align*}
\EE \, \tilde{\rA}\tilde{\rA}^\tr  = \frac{k}{\Tr(\Sigma_X)}  \Sigma_X (\Id_k + \Delta_x) , \qquad \EE \,  \tilde{\rB} \tilde{\rB}^\tr = \frac{k}{\Tr(\Sigma_Y)}  \Sigma_Y (\Id_k + \Delta_y). 
\end{align*}
for matrices $\Delta_x, \Delta_y$ that satisfy  $\|\Delta_x\|_{\op}, \|\Delta_{y}\|_{\op} \le \epsilon$. Combining these approximations with \eqref{eq:R_F}  and recalling that $\Sigma_X^{1/2} \rR \Sigma_Y^{1/2} = \rC_{XY}$, we conclude that 
\begin{align*}
\EE  \|  \tilde{\rR} \|_\rF^2    
=    \frac{ k^2 \| \rC_{XY} \|_\rF^2  }{ \Tr( \Sigma_X ) \Tr(\Sigma_Y)}\big(1 + o(1)\big) , \qquad d_x,d_y \to \infty.
\end{align*}

Finally, we need to show that ratio between $\EE \| \tilde{\rR}\tilde{\rR}^\tr  \|^2_F $ and $\EE \|\rR\|_\rF^2$ converges to zero. We begin by considering the lower bound
\begin{align*}
\|\tilde{\rR}\|_{\rF}
&  \ge \frac{ \| \rA^\tr \rC_{XY} \rB\|_\rF}{ \| (\rA^\tr \Sigma_X \rA )^{1/2} \|_{\op} \|(\rB^\tr \Sigma_Y \rB )^{1/2}\|_{\op}}
 \ge \frac{ \| \rA^\tr \rC_{XY} \rB \|_\rF}{ \|  \Sigma_X \|^{1/2}_{\op} \|\Sigma_Y\|^{1/2}_{\op}}
\end{align*}
as well as the upper bound
\begin{align*}
\|\tilde{\rR}\tilde{\rR}^\tr\|_{\rF}
& \le  \| (\rA^\tr \Sigma_X \rA )^{-1}\|_{\op} \| (\rB^\tr \Sigma_X \rB)^{-1}\|_{\op} \| \rA^\tr \rC_{XY} \rB \rB^\tr \rC_{XY}  \rA\|_\rF \\
& \le   \| \Sigma_X^{-1}\|_{\op} \|  \Sigma_Y^{-1}\|_{\op}  \|  \rA^\tr \rC_{XY} \rB \rB^\tr \rC^\tr_{XY} \rA \|_\rF.
\end{align*}
Note that matrices $\rA$ and $\rB$ in these bounds are the unbiased projections, which are uniformly distributed.  Since  $\EE \rA\rA^\tr  =(k/d_x)\Id_{d_x}$ and $\EE \rB\rB^\tr = (k/d_y) \Id_{d_y}$ one obtains
\[
\EE \|  \rA^\tr \rC_{XY} \rB \rB^\tr \rC^\tr_{XY} \rA \|_\rF = \frac{k^2}{ d_x d_y} \|\rC_{XY}\|_\rF^2
\]
Meanwhile, successive applications of  Lemma~\ref{lem:ASA}, first with respect to $\rA \rA^\tr$ and then  with respect to  $\rB\rB^\tr$, leads to
\begin{align*}
    \EE \|  \rA^\tr \rC_{XY} \rB \rB^\tr \rC_{XY}^\tr \rA\|_\rF^2  &\lesssim \frac{k^4}{d_x^2d_y^2} \left( \|\rC_{XY} \rC_{XY}^\tr\|_\rF^2 + \|\rC_{XY}\|_\rF^4 \right)\lesssim \frac{k^4  }{d_x^2d_y^2}  \|\rC_{XY}\|_\rF^4 
\end{align*}
Combining these upper and lower bounds and recalling that the condition numbers of $\Sigma_X$ and $\Sigma_Y$ are no greater than $\kappa$, we have
\begin{align*}
\EE   \|\tilde{\rR}\tilde{\rR}^\tr\|^2_{\rF} 
\lesssim \kappa^4\left( \EE \| \tilde{\rR} \|_\rF^2  \right)^2.
\end{align*}
In view of the fact that $\EE \| \tilde{\rR} \|_\rF^2 $ converges to zero, the proof is complete. 
\end{proof}


\subsection{Auxiliary results for the proof of  Theorem~\ref{THM:Gaussian_ksmi}}

\begin{lemma}\label{lem:logdet_bnd} 
If $\rM$ is a symmetric positive semidefinite matrix with $\|\rM\|_{\op} < 1$  then
\begin{align*}
0 \le  -   \log \det(\Id  - \rM) -  \Tr(\rM)    \le \frac{ \|\rM\|^2_\rF}{ 2 (1 - \| \rM\|_{\op} )} .
\end{align*}
\end{lemma}
\begin{proof}
The log determinant is given by  $ -  \log \det(\Id  - \rM)  -  \Tr(\rM)  = \sum_{i}  - \log(1 - \lambda_i) - \lambda_i$ where $0 \le \lambda_i  \le \| \rM\|_{\op}$ are the eigenvalues of $\rM$. Each summand satisfies the double inequality 
\begin{align*}
0 \le -  \log(1  - \lambda_i ) -\lambda_i  & = \int_0^{\lambda_i}  \frac{x}{ 1- x}  \, d x \le  \frac{\lambda_i^2}{2 (1 - \|\rM\|_{\op})}. 
\end{align*}
Summing over both sides and noting that $\|M\|_\rF^2 = \sum_{i} \lambda_i^2$ completes the proof. 
\end{proof}

\begin{lemma}\label{lem:ASA}
Let $\rP = \rA^\tr \rA$ where $\rA$ is distributed uniformly over $\mathrm{St}(k,d)$. 
For any $d \times d$ symmetric matrix $\rS$, we have
\begin{align*}
 \EE \, \Tr( \rP \rS \rP \rS)  &=  \frac{  k (   k d + d  -2 )  }{ d(d-1)  (d+2)} \Tr(\rS^2)  + \frac{  k ( d-k)  }{ d(d-1)  (d+2)} \Tr(\rS)^2\\
  \EE \,  \Tr( \rP \rS)^2  & = \frac{2 k(d-k   )}{d (d-1)(d+2) } \Tr(\rS^2)    +  \frac{  k ( kd + k - 2 )  }{ d(d-1)  (d+2)} \Tr(\rS)^2.
\end{align*}
\end{lemma}
\begin{proof}
Because the distribution of $\rP$ is invariant to orthogonal transformation of its rows and columns (i.e.,  $\rP$ is equal in distribution to $\rU \rP \rU^\tr$ for any $\rU \in \mathbb{O}(d)$), the quantities of interest are unchanged if  $\rS$ is replaced by a diagonal matrix containing its eigenvalues $\lambda_1, \dots, \lambda_n$. In particular, we have 
\begin{align*}
\Tr( \rP\rS \rP \rS  ) &\overset{d} {=} \Tr( \rP \diag(\lambda)  \rP \diag(\lambda) ) = \sum_{i,j = 1}^d \lambda_i \lambda_j \rP^2_{ij}\\
\Tr(\rP\rS)^2 &\overset{d}{=} \Tr( \rP \diag(\lambda) )^2 =  \sum_{i,j=1}^d\lambda_i \lambda_j  \rP_{ii}\rP_{jj}.
\end{align*}
A further consequence of the orthogonal invariance of $\rP$ is that its second order moments satisfy  $\EE[ \rP_{ii}^2 ]= \EE [\rP^2_{11}]$,  $\EE [\rP_{ij}^2] = \EE [\rP^2_{12}]$ and $\EE [\rP_{ii} \rP_{jj}] = \EE [\rP_{11} \rP_{22}]$ for all $1 \le i \ne j \le d$, and so the expectations can be simplified as follows: 
\begin{align}
\EE \, \Tr( \rP\rS \rP\rS)
& = \EE [ \rP_{11}^2 ]\sum_{i} \lambda^2_i  + \EE [\rP_{12}^2] \sum_{i \ne j} \lambda_i \lambda_j\notag \\
& =  \left(\EE [ \rP_{11}^2] - \EE[\rP_{12}^2 ] \right)  \Tr(\rS^2)  +  \EE \, [ \rP_{12}^2 ] \,  \Tr(\rS)^2\label{eq:TPSPS} \\
\EE \, \Tr(\rP\rS)^2  & = \EE [ \rP_{11}^2] \sum_{i} \lambda_i^2 + \EE[ \rP_{11} \rP_{22} ] \sum_{i \ne j}  \lambda_i \lambda_j  \notag\\
&=  \left( \EE  [\rP_{11}^2] -\EE[ \rP_{11} \rP_{22} ] \right)    \Tr(\rS^2)  +  \EE [\rP_{11} \rP_{22}] \,  \Tr(\rS)^2  \label{eq:TPS2}
\end{align}

Finally, we can determine coefficients in these expressions  by evaluating  \eqref{eq:TPSPS}  and \eqref{eq:TPS2}  for special choices of $\rS$. Recall that $\rP$ has $k$ nonzero eigenvalues all of which are equal to one. Therefore, if $\rS = \Id$, then $\Tr(\rS \rP) = k$ and $\Tr( \rS\rP\rS\rP)  = k^2$ a.s., and in view of \eqref{eq:TPSPS}  and \eqref{eq:TPS2},  we obtain
\begin{align*}
k &= d \EE[\rP_{11}^2], + d (d-1) \EE [\rP_{12}^2] \qquad 
k^2 = d \EE[\rP_{11}^2] + d (d-1) \EE [\rP_{11} \rP_{22}].
\end{align*}
Alternatively,  if $\rS = e_1 e_2 + e_2 e_1^\tr$ then $\EE \, \Tr(\rS\rP)^2 = \EE[(\rP_{12} + \rP_{21})^2] = 4 \EE[ \rP^2_{12}]$ and so \eqref{eq:TPS2} implies that 
\[
2 \EE[ \rP_{12}] =  \EE  [\rP_{11}^2] -\EE[ \rP_{11} \rP_{22}] .\]    
Solving these linear equations yields
\begin{align*}
    \EE [\rP_{11}^2]  = \frac{ k (k+2)}{ d (d+2)} , \qquad  \EE [ \rP_{12}^2]  = \frac{k-d}{ d (d-1) (d+2)} , \qquad \EE [\rP_{11}\rP_{22}] = \frac{ k(kd+k -2)}{ d (d-1) (d+2)}.
\end{align*}
Combining these expressions with  \eqref{eq:TPSPS}  and \eqref{eq:TPS2} gives the desired result. 
\end{proof}

 \begin{lemma}\label{lem:Pmean}
 Let $\rP = \Sigma^{1/2} \rA (\rA^\tr \Sigma \rA)^{-1} \rA^\tr \Sigma^{1/2}$ where $\Sigma$ is an deterministic $d \times d$ positive definite matrix with spectral decomposition $\Sigma =\sum_i \lambda_i u_i u_i^\tr $ and $\rA$ is distributed uniformly on $\mathrm{St}(k,d)$. Then, the mean of $\rP$ is given by  $\EE \rP = \sum_{i} \eta_i u_i u_i^\tr$ where
 \begin{align*}
     \eta_i = \EE   \left[  \frac{   \lambda_i \rZ^\tr_i \rW_i^{-1} \rZ_i }{  1  +    \lambda_i  \rZ^\tr_i \rW_i^{-1} \rZ_i } \right]
 \end{align*}
with  $\rZ_1,\dots, \rZ_d$ independent $\cN(0, \Id_k)$ variables and 
$\rW_i = \sum_{j \ne i} \lambda_j \rZ_j \rZ_j^\tr $.
 \end{lemma}
\begin{proof}
It is straightforward to show (see e.g., \cite[Theorem 3.2]{chikuse:1990}) that the distribution of the $n \times k$ matrix $\Sigma^{1/2} \rA(\rA^\tr \Sigma \rA)^{-1/2}$ 
is unchanged if the random matrix $\rA$ is replaced by Gaussian matrix $\rZ = [\rZ_1, \dots, \rZ_n]^\tr$ whose rows are independent $\cN(0,\Id_k)$ variables. Thus, letting $\rU = [u_1, \dots, u_n]$ and $\Lambda = \diag(\lambda_1, \dots, \lambda_n)$ be the be the eigenvectors and eigenvalues of  $\Sigma$ we have
\begin{align*}
\rU^\tr \rP\rU  \overset{d}{=} \Lambda^{1/2} \rZ (\rZ^\tr \Lambda \rZ)^{-1} \rZ^\tr \Lambda^{1/2}.
\end{align*}

In view of the above decomposition, we see that the $ij$-th entry of  $\rU^\tr \rP \rU$ is equal in distribution to $\lambda_i^{1/2} \lambda_j^{1/2} \rZ_i^\tr \left(  \lambda_i \rZ_i \rZ_i^\tr  + \rW_i  \right)^{-1}  \rZ_j$. For the off-diagonal entries, note that the distribution of $(\rZ_1, \dots, \rZ_n)$ is equal  to the distribution of $(\rZ_1, \dots, \rZ_{i-1}, S \rZ_i, \rZ_{i+1}, \dots, \rZ_n)$ where $S$ is an independent random variable distributed uniformly on $\{-1,1\}$. Making this substitution and then taking the expectation with respect to $S$ we see that the off-diagonal entries have mean zero.   The expression for the diagonal follows from applying  the matrix inversion lemma to $\lambda_i \rZ_i^\tr ( \lambda_i \rZ_i \rZ_i^\tr + \rW_i)^{-1} \rZ_i$.
 \end{proof}


\begin{lemma}\label{lem:P_approx} 
Consider the setting of Lemma~\ref{lem:Pmean}. There exists an absolute positive constant $C$ such that if
\begin{align*}
\frac{(2+k) \| \Sigma \|_{\op}  }{ \Tr(\Sigma) - \| \Sigma \|_{\op}} & \le \epsilon, \qquad 
 C\frac{\|\Sigma\|_{\op}}{ \Tr(\Sigma)} \left( k  + \sqrt{kd} + \log\left(\frac{2\Tr(\Sigma) \|\Sigma^{-1}\|_{\op} }{k\epsilon} \right) \right) \le \frac{\epsilon}{ 2 +\epsilon}
\end{align*}
for some $\epsilon >0$, then
\begin{align*}
\left|  \eta_i \cdot \frac{ \Tr(\Sigma) }{k \lambda_i  }  - 1 \right| \le \epsilon
\end{align*}
for all $1 \le i \le d$.
\end{lemma}
\begin{proof}
We begin with a lower bound on $\eta_i$.  For any nonzero vector $v \in \mathbb{R}^k$, the mapping $\rM   \mapsto  (v^\tr \rM^{-1} v )/(  1  +   v^\tr \rM^{-1} v )$ is convex over the cone of $k \times k$ positive semidefinite matrices. By Jensen's inequality,  the independence of $\rW_i$ and $\rZ_i$, and the fact that $\EE [\rW_i] = \tau_i \Id_k$ where $\tau_i := \sum_{j \ne i} \lambda_i = \Tr(\Sigma) - \lambda_i$,  we have
 \begin{align*}
\eta_i & \ge  \EE  \left[   \frac{   \lambda_i \rZ_i ( \EE[\rW_i])^{-1} ]\rZ_i }{  1  +    \lambda_i  \rZ_i  (\EE[\rW_i])^{-1} \rZ_i } \right]   =  \EE \left[    \frac{   \lambda_i \| \rZ_i\|^2 }{ \tau_j +    \lambda_i   \| \rZ_i\|^2}\right].
\end{align*}
To remove remove the expectation with respect to $\|Z_i\|^2$, we bound the RHS from below using
\begin{align*}
 \EE \left[    \frac{   \lambda_i \| \rZ_i\|^2 }{ \tau_j +    \lambda_i   \| \rZ_i\|^2}\right] & =      \frac{  k \lambda_i  }{ \Tr( \Sigma) } -   \frac{\lambda_i^2}{ \Tr(\Sigma)} \EE \left[    \frac{  \| Z_i\|^2 ( \ \|Z_i\|^2 - 1  )  }{  (\tau_j +  \lambda_i \|Z_i\|^2    ) }\right]
        \ge     \frac{  k \lambda_i  }{ \Tr( \Sigma) } -   \frac{ k (2+k) \lambda^2_i}{ \Tr(\Sigma) (\Tr(\Sigma) - \|\Sigma\|_{\op}) },
 \end{align*}
where the second step follows from $\EE\|\rZ_i\|^2 = k$ and $\EE  \|\rZ_i\|^4 = k (k+2)$.

Next we consider an upper bound. If we let $L_i := \min\{ u^\tr W u : u \in \mathbb{S}^{k-2}\}$ be the minimum eigenvalue of $k \times k$ symmetric matrix $W_i$ and then we can write
 \begin{align*}
\eta_i  & \le  \EE\left[    \frac{  \lambda_i   L_i^{-1} \|Z_i\|^2  }{  1+ \lambda_i   L_i^{-1} \|Z_i\|^2 }\right]  \le  \EE \left[     \frac{  k \lambda_i    }{  L_i  +   k \lambda_i   } \right] 
 \end{align*}
 where the second step follows from the Jensen's inequality and the independence of $Z_i$ and $L_i$.  By concentration of Lipschitz functions of Gaussian measure, one finds that that $L^{1/2}_i$ is sub-Gaussian with variance proxy $\max_{j \ne i} \lambda_j \le \|\Sigma\|_{\op}$ and this implies a sub-exponential tail bound for $L_i$ of the form%
 \begin{align*}
\PP\Big(  L_i \le \EE [L_i] - C'  \|\Sigma\|_{\op} t \Big)  \le 2 e^{ -  t} .
 \end{align*} 
 for some absolute constant $C' > 0$. 
 To obtain a lower bound on the expectation of $L_i$, recall that $\EE[ W_i] = \tau_i \Id_k$ where $\tau_i = \sum_{j \ne i} \lambda_j$.
 Noting that
 \begin{align*}
     \tau_i - L_i \le  | L_i - \tau_i| \le - \| W_i - \EE[ W_i] \|_{\op},
 \end{align*}
 and then taking the expectation of both sides leads to  $
    \EE [L_i]  \ge \tau_i -\EE \| W_i - \EE[ W_i] \|_{\op}$. 
At this point, we can apply Theorem 3.13 in \cite{cai2022non}, which gives
    \begin{align*}
  \EE \big \| W_i - \EE [ W_i] \big \|_{\op}&  \lesssim \sqrt{ k \sum_{j \ne i} \lambda_j^2 } + k \max_{j \ne i}  \lambda_j \le  \|\Sigma\|_{\op}   ( k  +1+ \sqrt{dk}),
    \end{align*} 
Combining these bounds and recalling that $\tau_i = \Tr(\Sigma) - \lambda_i$ yields
    \begin{align*}
\EE L_i  \ge  \Tr(\Sigma) -   C'' \|\Sigma\|_{\op}   ( k + \sqrt{dk}), 
  \end{align*} 
for some absolute constant $C''>0$.
Putting the pieces together, we have for all $t > 0$, 
\begin{align*}
    \eta_i & \le \EE\left[  \frac{k \lambda_i}{ L_i + k \lambda_i}  \mathbbm{1}_{\{ L_i \ge \EE[L_i] + C' \|\Sigma\|_{\op}  t\}}\right]+ \EE\left[  \frac{k \lambda_i}{ L_i + k \lambda_i} \mathbbm{1}_{\{ L_i < \EE[L_i] + C' \|\Sigma\|_{\op} t\}}\right]\\
   & \le \frac{ k \lambda_i}{ \EE[L_i ]- C' \|\Sigma\|_{\op}  t + k \lambda_i } + 2 e^{- t} \\
 & \le \frac{ k \lambda_i}{  \Tr(\Sigma) -   C'' \|\Sigma\|_{\op}   ( k + \sqrt{dk})  - C' \|\Sigma\|_{\op}  t + k \lambda_i } + 2 e^{- t} .
\end{align*}
where the last two lines hold provided that the denominator is strictly positive. Hence, if $t = \log(2\Tr(\Sigma) \|\Sigma^{-1}\|_{\op} /(k\epsilon))$ and 
\begin{align*}
   \Tr(\Sigma) -   C'' \|\Sigma\|_{\op}   ( k +\sqrt{dk})   - C' \|\Sigma\|_{\op} t \ge \frac{  \Tr(\Sigma) }{ 1 + \epsilon/2},
\end{align*}
then 
\begin{align*}
    \eta_i 
    \le (1 + \epsilon/2)  \frac{ k \lambda_i}{  \Tr(\Sigma)  }  +  \frac{ \epsilon k }{2 \Tr(\Sigma) \|\Sigma^{-1}\|_{\op} }  \le \frac{k \lambda_i}{\Tr(\Sigma)} (1+\epsilon)
\end{align*}
Simplifying the conditions leads to the stated bound. 
\end{proof}


\subsection{Proof of Decomposition in Equation \eqref{EQ:kSMI_decomposition}}\label{APPEN:kSMI_decomposition_proof}

Fix $\theta:=(\theta_1,\theta_2)=(\rA,\rB)\in\sti(k,d_x)\times\sti(k,d_y)$ and let $f_\theta$, $f_{\theta_1}$, and $f_{\theta_2}$ denote, respectively, the densities of $(\rA^\tr X,\rB^\tr Y)$, $\rA^\tr Z$, and $\rB^\tr Y$, where $(X,Y)\sim\mu_{XY}$. Similarly, we use $\varphi_\theta$, $\varphi_{\theta_1}$, and $\varphi_{\theta_2}$, for the densities when $(X,Y)$ are replaced with their Gaussian approximation $(X^*,Y^*)\sim \gamma_{XY}$. We may now decompose
\begin{align*}
\sI(\rA^\tr X;\rB^\tr Y)&=\int f_\theta(s,t)\log\left(\frac{\varphi_\theta(s,t)}{\varphi_{\theta_1}(s)\varphi_{\theta_2}(t)}\frac{f_\theta(s,t)}{\varphi_\theta(s,t)}\frac{\varphi_{\theta_1}(s)\varphi_{\theta_2}(t)}{f_{\theta_1}(s)f_{\theta_2}(t)}\right)ds\,dt\\
&=\EE_{\mu_{XY}}\left[\log\left(\frac{\varphi_\theta}{\varphi_{\theta_1}\varphi_{\theta_2}}\right)\right]+\dkl\big((\proj^\rA,\proj^\rB)_\sharp\mu_{XY}\big\|(\proj^\rA,\proj^\rB)_\sharp\gamma_{XY}\big)\\
&\mspace{270mu}-\dkl\big((\proj^\rA,\proj^\rB)_\sharp\mu_X\otimes \mu_Y \big\|(\proj^\rA,\proj^\rB)_\sharp\gamma_X\otimes \gamma_Y\big).
\end{align*}

Observing that $\log\left(\frac{\varphi_\theta}{\varphi_{\theta_1}\varphi_{\theta_2}}\right)$ depends only on the 2nd moment on the random variables and since the Gaussian approximation $(X^*,Y^*)\sim \gamma_{XY}$ was chosen to have the same covariance matrix as $(X,Y)\sim\mu_{XY}$, we may replace the distribution $\mu_{XY}$ w.r.t. which the expectation is taken with $\gamma_{XY}$. Doing so and taking an average over $(\rA,\rB)\in\sti(k,d_x)\times\sti{k,d_y}$, we obtain
\[
\ksmi(X;Y)=\ksmi(X^*;Y^*)+\EE\big[\delta(A,B)\big]
\]
where $\delta(\rA,\rB)$ is as defined under Equation \eqref{EQ:kSMI_decomposition}     in the main text.
\qed


\section{Bounds on Residual Term from Equation \eqref{EQ:kSMI_decomposition}}\label{APPEN:residual_bound}

Throughout this appendix we interchangeably denote information measures in terms of probability distribution or the corresponding random variables. For instance, we write $\sJ(X)$ or $\sJ(\mu)$ for the Fisher information of $X\sim \mu$, and $\sW_2(X,Y)$ or $\sW_2(\mu,\nu)$ for the 2-Wasserstein distance between $X\sim \mu$ and $Y\sim \nu$. We also define $\alpha(X):=  \frac{1}{d_X}\EE\big[\big| \|X\|^2 - \EE \|X\|^2 \big| \big]$ and $\smash{\beta_r(X) := \frac{1}{d_x} \big(\EE\big[\big| \langle X_1, X_2\rangle\big|^r\big]\big)^{1/r}}$, for $r=1,2$, where $X_1$ and $X_2$ are independent copies of $X\sim \mu_X$. The quantities $\alpha(Y)$  and $\beta_r(Y)$ are defined analogously. Note that $\beta_1(X) \leq \beta_2(X) = \frac{1}{d_x} \|\Sigma_X\|_\rF $. 

Due to translation invariance of $k$-SMI we may assume that $X$ and $Y$ are centered. Define the shorthand notation $\Theta =  \rA\oplus \rB$ and $Z=(X^\tr\ Y^\tr)^\tr$. Accordingly,  $\mu_Z=\mu_{XY}$ and we set $\gamma_Z=\gamma_{XY}=\cN(0,\Sigma_{XY})$ for the corresponding Gaussian; the Gaussian vector with distribution $\gamma_Z$ is denoted by $Z_*=(X_*^\tr\ Y_*^\tr)^\tr$. Slightly abusing notation we define $\proj^\Theta(z)=\Theta^\tr z=(x^\tr\rA\ y^\tr\rB)^\tr$.

To control the residual from \eqref{EQ:kSMI_decomposition}, we first bound it in terms of a certain MI term. Let $\rA_*$ and $\rB_*$ be matrices of dimension $d_x\times k$ and $d_y\times k$ with entries i.i.d. according to $\cN(0,1/d_x)$ and $\cN(0,1/d_y)$, respectively. Define $\Theta_*=\rA_*\oplus \rB_*$ and let $W=\Theta_*^\tr Z+\sqrt{t}N$, where $N\sim \cN(0,\rI_{2k})$. The following bound controls the residual in terms of $\sI(\Theta_*;W)$, plus a term that vanishes when $t$ is small and $d_x,d_y$ are large. The proof is deferred to \cref{APPEN:noisy_proof}.

\begin{lemma}[Residual bound via noisy MI]\label{lem:noisy}
Under the above model with $d_x\wedge d_y>k+1$ and for any $t>0$, we have 
\begin{align*}
&\EE\big[ \dkl( \proj^\Theta_\sharp \mu_Z \| \proj^{\Theta}\gamma_Z )\big]\leq \sI(\Theta_*;W)\\
&+\mspace{-2mu}k\sqrt{\mspace{-2mu}2\|\rJ_\rF(\mu_Z)\|_{\op}}\Bigg(\mspace{-3mu}\sqrt{t\mspace{-2mu}\left(\mspace{-1mu}\frac{d_x}{d_x\mspace{-3.5mu}-\mspace{-2.5mu}k\mspace{-2.5mu}+\mspace{-2.5mu}1}\mspace{-3mu}+\mspace{-3mu}\frac{d_y}{d_y\mspace{-3.5mu}-\mspace{-2.5mu}k\mspace{-2.5mu}+\mspace{-1.5mu}1}\mspace{-1mu}\right)}\mspace{-3mu}+\mspace{-3mu}\sqrt{\alpha(X) \mspace{-2mu}+\mspace{-2mu} \alpha(Y) \mspace{-2mu}+   \mspace{-2mu}\frac{ 2d_x\beta_2^2(X)}{\Tr( \Sigma_X)} \mspace{-2mu}+\mspace{-2mu}  \frac{ 2d_y\beta_2^2(Y)}{\Tr( \Sigma_Y)} }\Bigg)\mspace{-2mu}.
\end{align*}
\end{lemma}

Next, we bound the noisy MI term $\sI(\Theta_*;W)$. Let $\lambda_x = \frac{1}{d_x} \mathbb{E}\|X\|^2$ and $\lambda_y = \frac{1}{d_y} \mathbb{E}\|Y\|^2$, and for simplicity of presentation, henceforth assume that $\lambda=\lambda_x = \lambda_y$. This is without loss of generality since $k$-SMI is scale invariant.\footnote{This scaling does affect the $\alpha$, $\beta$ factors in the lemma but will not change their convergence properties so long as $\lambda_x$ and $\lambda_y$ scale at the same rate.}  Note that if $\lambda \in (0,\infty)$, then we have $0 \leq \alpha(X) \leq 2\lambda$ and $\lambda/\sqrt{d_x}\leq \beta_2(X) \leq \lambda $ (cf. \cite{reeves2017conditional}). Lastly, set $\bar{\alpha} = \max\{\alpha(X), \alpha(Y)\}$ and $\bar{\beta_r} = \max\{\beta_r(X),\beta_r(Y)\}$, for $r=1,2$. We prove the following result in \cref{APPEN:2dInformation}.

\begin{lemma}[Noisy MI bound]\label{thm:2dinf}
For any $t>0$ and $\epsilon \in (0,1]$, we have 
\begin{align*}
\sI(\Theta_*;W) \leq Ck \log\left(1+ \frac{\lambda}{t}\right){\frac{\bar{\alpha}}{\epsilon \lambda}} + C\left(\frac{1+\epsilon}{1-\epsilon}\right)^{\frac{k}{4}}\left( k^{\frac{3}{4}} \sqrt{\frac{\bar{\beta}_1}{\lambda}} + k^{\frac{1}{4}} \left(1 + \frac{2(1+\epsilon)\lambda}{t}\right)^{\frac{k}{2}}\frac{\bar{\beta}_2}{\lambda}\right).
\end{align*}
where $C$ is an absolute constant (in particular, $C=3$ is sufficient).
\end{lemma}

Combining Lemmas \ref{lem:noisy} and \ref{thm:2dinf}, yields a bound on $\EE\big[ \dkl( \proj^\Theta_\sharp \mu_Z \| \proj^{\Theta}\gamma_Z )\big]$ in terms $k,\ d_x,\ d_y,\ \lambda,\ \bar{\alpha}, \bar{\beta}$ and (arbitrary) $t>0$ and $\epsilon\in(0,1]$. To further simplify the subsequent expressions, suppose that $(\bar{\beta}_2/\lambda)^{2/(k+1)}\leq\frac{1}{2}$, and set\footnote{Our bounds only need $t^\ast$ to be strictly positive, which is always the case under the considered setting. Indeed, by the the Cauchy-Schwartz inequality $\bar{\beta}_2\leq \lambda$, with equality having probability zero since two independent copies of a continuous random variables are a.s. not linearly aligned.}
\[
{t^\ast} = {2(1+\epsilon)\lambda}\left( \left(\frac{\bar{\beta}_2}{\lambda}\right)^{-\frac{2}{k+1}}-1\right)^{-1} \leq 4(1+\epsilon)\lambda \left(\frac{\bar{\beta}_2}{\lambda}\right)^{\frac{2}{k+1}}.
\]
Inserting into the said bound, we obtain 
\begin{align*}
&\EE\big[ \dkl( \proj^\Theta_\sharp \mu_Z \| \proj^{\Theta}\gamma_Z )\big]\\
&\leq Ck \log\left(1+ \frac{1}{2(1+\epsilon)}\left( \left(\frac{\bar{\beta}_2}{\lambda}\right)^{-\frac{2}{k+1}}-1\right)\right){\frac{\bar{\alpha}}{\epsilon \lambda}}\\
&\qquad+ C\left(\frac{1+\epsilon}{1-\epsilon}\right)^{\frac{k}{4}}\left( k^{\frac{3}{4}} \sqrt{\frac{\bar{\beta}_1}{\lambda}} + k^{\frac{1}{4}} \left(\frac{\bar{\beta}_2}{\lambda}\right)^{\frac{1}{k+1}}\right)\mspace{-3mu}\\
&\qquad\qquad+ k\sqrt{2\|\rJ_\rF(\mu_Z)\|_{\op}}\left(\left(\frac{\bar{\beta}_2}{\lambda}\right)^{\frac{1}{k+1}}\sqrt{4(1+\epsilon)\lambda \left(\frac{d_x}{d_x-k+1}+\frac{d_y}{d_y-k+1}\right)}\right. \\&\qquad\qquad\qquad\qquad\qquad\qquad\qquad\qquad\qquad\qquad\qquad\left.+\sqrt{2\bar{\alpha}+  \frac{ 2d_x \beta_2^2(X)}{ \Tr( \Sigma_X)} + \frac{ 2d_x\beta_2^2(Y)}{ \Tr( \Sigma_Y)} }\right).
\end{align*}

We can now complete the bound on the residual term $\mathbb{E}[\delta_{XY}(\rA,\rB)]$ from  \eqref{EQ:kSMI_decomposition}. Recall the definition of $\Theta=\rA\oplus \rB$ and $Z=(X^\tr\ Y^\tr)^\tr$, we have 
\begin{align*}
\EE\big[\delta_{XY}(\rA,\rB)\big] &\leq \EE\Big[\dkl\big((\proj^\rA\mspace{-2mu},\proj^\rB)_\sharp\mu_{XY} \big\|(\proj^\rA\mspace{-2mu},\proj^\rB)_\sharp\gamma_{XY}\big)\Big]\\
&\leq Ck \log\left(1+ \frac{1}{2(1+\epsilon)}\left( \left(\frac{\bar{\beta}_2}{\lambda}\right)^{-\frac{2}{k+1}}-1\right)\right){\frac{\bar{\alpha}}{\epsilon \lambda}}\\
&\qquad  +  C\left(\frac{1+\epsilon}{1-\epsilon}\right)^{\frac{k}{4}}\left( k^{\frac{3}{4}} \sqrt{\frac{\bar{\beta}_1}{\lambda}} + k^{\frac{1}{4}} \left(\frac{\bar{\beta}_2}{\lambda}\right)^{\frac{1}{k+1}}\right)\mspace{-3mu}\\
&\qquad+ k\sqrt{2\|\rJ_\rF(\mu_Z)\|_{\op}}\left(\left(\frac{\bar{\beta}_2}{\lambda}\right)^{\frac{1}{k+1}}\sqrt{4(1+\epsilon)\lambda \left(\frac{d_x}{d_x-k+1}+\frac{d_y}{d_y-k+1}\right)}\right. \\&\qquad\qquad\qquad\qquad\qquad\qquad\qquad\qquad\qquad\left.+\sqrt{2\bar{\alpha}+  \frac{ 2d_x \beta_2^2(X)}{ \Tr( \Sigma_X)} + \frac{ 2d_x\beta_2^2(Y)}{ \Tr( \Sigma_Y)} }\right).
\end{align*}
Observe that this will typically converge to zero with increasing $d_x$, $d_y$. To better instantiate this regime, we revisit the concept of weak dependence, i.e. random vectors with weakly dependent entries \cite{reeves2017conditional} (essentially, a notion of approximate isotropy). 
The following proposition, whose proof is straightforward and hence omitted, provides explicit convergence rate for the residual subject to the weak dependence assumption. 
\begin{proposition}[Convergence rate under weak dependence]
Suppose that $\lambda$, $k$, $\|\rJ_\rF(\mu_Z)\|_{\op}$, $\frac{ d_x}{ \Tr( \Sigma_X)}$, $\frac{ d_y}{ \Tr( \Sigma_Y)}$ are $O(1)$ with respect to $d_x$ and $d_y$, and that there exists an absolute $C<0$ such that 
\[
\frac{\alpha(X)}{\lambda} \leq \frac{C}{\sqrt{d_x}}\quad,\quad \frac{\alpha(Y)}{\lambda} \leq \frac{C}{\sqrt{d_y}}\quad,\quad
\frac{\beta_2(X)}{\lambda} \leq \frac{C}{\sqrt{d_x}}\quad,\quad
\frac{\beta_2(Y)}{\lambda} \leq \frac{C}{\sqrt{d_y}}.
\]
Then, up to log factors,
\begin{align*}
\mathbb{E}[\delta_{XY}(\rA,\rB)] \lesssim& \frac{k}{\epsilon} \left(d_x^{-\frac{1}{2}} \mspace{-3mu}+\mspace{-3mu} d_y^{-\frac{1}{2}}\right) \mspace{-3mu}+\mspace{-3mu} \left(\frac{1+\epsilon}{1-\epsilon}\right)^{\frac{k}{4}}\mspace{-3mu}\left(k^{\frac{3}{4}}\left(d_x^{-\frac{1}{4}} \mspace{-3mu}+\mspace{-3mu} d_y^{-\frac{1}{4}}\right) \mspace{-3mu}+ \mspace{-3mu}k^{\frac{1}{4}} \left(d_x^{-\frac{1}{2(k+1)}} \mspace{-3mu}+\mspace{-3mu} d_y^{-\frac{1}{2(k+1)}}\right)\right) \\
&\mspace{230mu}+ k \left(d_x^{-\frac{1}{2(k+1)}} + d_y^{-\frac{1}{2(k+1)}}\right) + k\left(d_x^{-\frac{1}{4}} + d_y^{-\frac{1}{4}}\right),
\end{align*}
which, for $d_x = d_y = d$ increasing, decays to zero as $\tilde{O}\left(d^{-\frac{1}{4}} + d^{-\frac{1}{2(k+1)}}\right)$.
\end{proposition}

\subsection{Proof of \cref{lem:noisy}}\label{APPEN:noisy_proof}

We can represent the residual term from \eqref{EQ:kSMI_decomposition}, as
\begin{align*}
\EE\big[ \dkl( \proj^\Theta_\sharp \mu_Z \| \proj^{\Theta}\gamma_Z )\big] & =\sI(\Theta; \Theta^\tr Z) +\sh(\Theta^\tr Z^* | \Theta) - \sh(\Theta^\tr Z ) \\
& \le \sI(\Theta; \Theta^\tr Z) + \sh(\Theta^\tr Z^* ) - \sh(\Theta^\tr Z )\numberthis\label{EQ:residual_bound1}
\end{align*}

For the latter entropy difference we use the Wasserstein continuity result from \cref{LEMMA:Wass_cont}, to obtain
\begin{equation}
\sh(\Theta^\tr Z^* ) - \sh(\Theta^\tr Z )\leq \sqrt{\sJ(\Theta^\tr Z)}\sW_2(\Theta^\tr Z,\Theta^\tr Z).\label{EQ:Gauss_entropy_bound}    
\end{equation}

For the Fisher information term we use the data processing inequality \cite[Proposition 5]{rioul2010information}\footnote{The Fisher information $\sJ(W)$ is related to the parametric Fisher information $\sJ_\vartheta(W):=\mathrm{Var}\left(\frac{\partial}{\partial\vartheta}\log p_\vartheta(W)\right)$ as follows: if $\vartheta\in\RR^d$ is a location parameter, i.e., $p_\vartheta(w)=p(w+\vartheta)$, then $\sJ(W)=\sJ_\vartheta(W-\vartheta)$. Proposition 5 of \cite{rioul2010information} states that if $\vartheta\leftrightarrow W \leftrightarrow \widetilde{W}$ form a Markov chain, then $\sJ_\vartheta(\widetilde{W})\leq\sJ_\vartheta(W)$. Take $W=(\vartheta+\Theta^\tr Z,\Theta)$ and $\widetilde{W}=\vartheta+\Theta^\tr Z$, which clearly satisfy the said Markov chain, and invoke that result to obtain $\sJ(\Theta^\tr Z)=\sJ_\vartheta(\widetilde{W})\leq \sJ_\vartheta(W)=\int \sJ(\theta^\tr Z)d(\sigma_{k,d_x}\otimes\sigma_{k,d_y})(\theta)$. The latter equality is since $\frac{\partial}{\partial\vartheta}\log p_{\vartheta+\Theta^\tr Z,\Theta}(\cdot,\cdot)=\frac{\partial}{\partial\vartheta}\log p_{\vartheta+\Theta^\tr Z|\Theta}(\cdot|\cdot)$.} and the fact that $\Theta$ is an orthogonal matrix (i.e., $\Theta^\tr\Theta=\rI_{2k}$) to obtain
\begin{align*}
\sJ(\Theta^\tr Z)& \le \int \sJ( \theta^\tr Z) d(\sigma_{k,d_x}\otimes\sigma_{k,d_y})(\theta)\\
& \le \int  \Tr\big( \theta \rJ_\rF( Z) \theta^\tr\big) d(\sigma_{k,d_x}\otimes\sigma_{k,d_y})(\theta)\\
& \le 2k\|\rJ_\rF(Z)\|_{\op}.\numberthis\label{EQ:FI_bound}
\end{align*}


To treat the 2-Wasserstein distance, note that by orthogonal invariance of the projections, we see that the (unconditional) distribution of $\Theta^\tr Z$ satisfies 
\begin{align*}
    \Theta^\tr Z = \begin{pmatrix} \rA^\tr X \\ \rB^\tr  Y
    \end{pmatrix}  \overset{d}{=} \begin{pmatrix} \rA_1 \| X\|\\ \rB_1 \|Y\|
    \end{pmatrix} 
\end{align*}
where $\rA_1$ and $\rB_1$ are the first rows of $\rA$ and $\rB$, respectively. The same decomposition holds for  $\Theta^\tr Z^*$. Hence, 
\[\sW_2(\Theta^\tr Z, \Theta^\tr Z^*)=    \sW_2\mspace{-3mu}\left(\begin{pmatrix} \rA_1 \|X\|\\ \rB_1 \|Y\|
    \end{pmatrix}  \mspace{-3mu},\mspace{-3mu} \begin{pmatrix} \rA_1 \|X^*\|\\ \rB_1 \|Y^*\|
    \end{pmatrix}  \right)\le \sqrt{k} \,      \sW_2\mspace{-3mu}\left(\begin{pmatrix}  \|X\|/ \sqrt{d_x} \\ \|Y\|/\sqrt{d_y} 
    \end{pmatrix}\mspace{-3mu},\mspace{-3mu} \begin{pmatrix}  \|X^*\|/\sqrt{d_x} \\ \|Y^*\|/\sqrt{d_y} 
    \end{pmatrix}  \right),
\]
where the inequality follows from restricting to a coupling with the same $(\rA_1,\rB_1)$ and recalling that the entries of $\rA_1$ and $\rB_1$ have second moments of $1/d_x$ and $1/d_y$, respectively.

For any coupling of $(X,Y)$ and $(X^*,Y^*)$, we have 
\begin{align*}
     &\EE \left\| \begin{pmatrix}  \|X\|/ \sqrt{d_x} \\ \|Y\|/\sqrt{d_y} 
    \end{pmatrix}  -  \begin{pmatrix}  \|X^*\|/\sqrt{d_x} \\ \|Y^*\|/\sqrt{d_y} 
    \end{pmatrix}  \right\|^2\\
    &\qquad\qquad\qquad\le  \frac{2}{d_x} \EE  \left[ \left| \|X\| - \sqrt{\Tr(\Sigma_X)}  \right|^2 \right]  + \frac{2}{d_x} \EE  \left[ \left| \|X^*\| - \sqrt{\Tr(\Sigma_X)}  \right|^2 \right] \\
    &\qquad\qquad\qquad\qquad +  \frac{2}{d_x} \EE  \left[ \left| \|Y\| - \sqrt{\Tr(\Sigma_Y)}  \right|^2 \right]  + \frac{2}{d_y} \EE  \left[ \left| \|Y^*\| - \sqrt{\Tr(\Sigma_Y)}  \right|^2 \right] 
\end{align*}
where we have used the inequality $(a+b)^2 \le 2a^2 + 2b^2$. Note that for any random positive random variable $W$ we have
\begin{align*}
    \EE  \left[ \left| W - \sqrt{\EE W^2}  \right|^2 \right] \le \EE \left[ \left| W - \sqrt{\EE W^2}  \right|^2 \left(1 + \frac{W}{\sqrt{\EE W^2 }} \right)^2  \right] = \frac{ \Var(W^2)}{ \EE W^2 }
\end{align*}
Since $(X^*,Y^*)$ are Gaussian, their squared Euclidean norms can be expressed as the weighted sum of independent chi-squared variables, and one finds that  $\EE[ \|X^*\|^2] =   \Tr( \Sigma_X)$ 
and $\Var( \|X^*\|^2) =  2 \| \Sigma_X\|_{\rF}^2 $, and similarly for $Y^*$. Putting everything together, we obtain
\begin{align*}
    \sW^2_2(\Theta^\tr Z, \Theta^\tr Z^*) 
    &\le  2k \left( \alpha(X)  + \alpha(Y)+ \frac{ 2d_x\beta_2^2(X)}{ \Tr(\Sigma_X)}+\frac{2d_y\beta_2^2(X)}{\Tr( \Sigma_Y)} \right)\numberthis\label{EQ:W2_bound_final}
\end{align*}
where $\alpha(X)$ and $\alpha(Y)$ are defined in \cref{lem:noisy}.

It remains to transform the MI term $\sI(\Theta;\Theta^\tr Z)$ in \eqref{EQ:residual_bound1} into  $\sI(\Theta_*;\Theta_*^\tr Z+\sqrt{t}N)$, where $N\sim \cN(0,\rI_{2k})$ and $\Theta_*=\rA_*\oplus \rB_*$ with $\rA_*$ and $\rB_*$ matrices of dimension $d_x\times k$ and $d_y\times k$ and entries i.i.d. according to $\cN(0,1/d_x)$ and $\cN(0,1/d_y)$, respectively. Using the polar decomposition of Gaussian matrices, we know that $\smash{\rA_*\stackrel{d}{=}\rA (\rA_*^\tr\rA_*)^{1/2}}$, where $\rA\sim \sigma_{k,d_x}$, i.e., it is uniformly distributed over $\sti(k,d_x)$. A similar claim holds for $\rB_*$. 
By invariance of MI to invertible transformations ($\rA_*^\tr \rA_*$ and $\rB_*^\tr \rB_*$ are a.s. invertible), we have
\[\sI(\Theta;\Theta^\tr Z)=\sI(\rA,\rB;\rA^\tr X,\rB^\tr Y)=\sI(\rA_*,\rB_*;\rA_*^\tr X,\rB_*^\tr Y)=\sI(\Theta_*;\Theta_*^\tr Z).\]

Next, we introduce the noise into the latter MI as follows. Denote the distribution of $\Theta_*$ by $\gamma$ and consider
\begin{align*}
\sI(\Theta_*;\Theta_*^\tr Z)-\sI(\Theta_*; W)&=\sh(\Theta_*^\tr Z)-\sh(W)+\sh(W|\Theta_*)-\sh(\Theta_*^\tr Z|\Theta_*)\\
&\leq \int \sh(W|\Theta_*=\theta)-\sh(\Theta_*^\tr Z|\Theta_*=\theta) d\gamma(\theta)\\
&\leq \int \sqrt{\sJ(\theta^\tr Z)}\sW_2\big(\theta^\tr Z,\theta^\tr Z+\sqrt{t}N\big)d\gamma(\theta)\\
&\leq \sqrt{2kt\int \sJ(\theta^\tr Z)d\gamma(\theta)},\numberthis\label{EQ:noisy_MI_difference}
\end{align*}
where the first inequality follow since $\sh(W)\geq \sh(W|N)=\sh(\Theta_*^\tr Z)$ (as conditioning cannot increase differential entropy), the second inequality follows from  \cref{LEMMA:Wass_cont}, while the last step upper bounds the 2-Wasserstein distance by $\big(\EE[\|\sqrt{t}N\|^2]\big)^{1/2}= \sqrt{2kt}$ and applies Jensen's inequality. 

To bound the expected Fisher information, let $\theta^\dagger$ denote the pseudo-inverse of $\theta$. Using the data processing inequality once more, we have
\begin{align*}
\int \sJ( \theta^\tr Z)  d\gamma(\theta) & \le \int  \Tr\big( \theta^\dagger \rJ_\rF(Z) (\theta^\dagger)^\tr\big) d\gamma(\theta)\\
& = \int  \Tr\big(  \rJ_\rF( Z) (\theta \theta^\tr)^{-1}\big) d\gamma(\theta)\\
& = \|\rJ_\rF(Z)\|_{\op}\, \EE\big[\Tr\big((\Theta_*^\tr \Theta_*)^{-1}\big)\big]\\
& =  \|\rJ_\rF(Z)\|_{\op}\, \left(d_x\EE\,\big[\Tr\big((\tilde{\rA}_*^\tr \tilde{\rA}_*)^{-1}\big)\big]+d_y\EE\,\big[\Tr\big((\tilde{\rB}_*^\tr \tilde{\rB}_*)^{-1}\big)\big]\right),
\end{align*}
where $\tilde{\rA}_*$ and $\tilde{\rB}_*$ are random Gaussian matrices of dimensions $d_x\times k$ and $d_y\times k$, respectively, with i.i.d. $\cN(0,1)$ entries. Consequently, note that  $\tilde{\rA}_*^\tr\tilde{\rA}_*$ and $\tilde{\rB}_*^\tr\tilde{\rB}_*$ follow the $k \times k$ Wishart distribution with $d_x$ and $d_y$ degrees of freedom, respectively. For $d_x > k +1$ the mean of the inverse is $\EE\big[ (\rA^\tr_*\rA_*)^{-1}\big] = \frac{1}{d_x - k - 1} \rI_{2k}$ and so $\EE\big[\Tr\big( (\rA^\tr_*\rA_*)^{-1}\big)\big]  = \frac{k}{d_x - k - 1}$; cf. e.g., \cite{imori2020mean} (and similarly for $\tilde{\rB}_*$). Inserting this into \eqref{EQ:noisy_MI_difference} and combining with \eqref{EQ:Gauss_entropy_bound} and \eqref{EQ:W2_bound_final} yields the result.\qed

\subsection{Proof of Lemma \ref{thm:2dinf}}\label{APPEN:2dInformation}

The following bounds follow by the exact same argument of Lemmas 4 and 5 from \cite{reeves2017conditional}, respectively.

\begin{lemma}
\label{lemm:4}
We have
\[
\sI(\Theta_*;W) \leq \kappa \int_{\mathbb{R}^{2k}} \sqrt{\mathrm{Var}(p_{W|\Theta_*}(w|\Theta_*))} dz
\]
where $\kappa  = \sup_{x \in (0,\infty)} \log (1 + x) /\sqrt{x} \approx 0.80474$.
\end{lemma}

\begin{lemma}
\label{lemm:5}
Let $f:\RR^d\to\RR_{\geq 0}$ be a non-negative integrable function and denote its $p$th moment by $\eta_p[f]:=\int \|z\|^pf(z)dz$. If $\eta_{d-1}[f]$, $\eta_{d+1}[f] < \infty$, then
\[
\int \sqrt{f(z)} dz \leq \sqrt{\frac{2 \pi^{\frac{d}{2} + 1}}{\Gamma\left(\frac{d}{2}\right)}}\big(\eta_{d-1}[f] \eta_{d+1}[f]\big)^{\frac{1}{4}},
\]
where $\Gamma(z)$ is the Gamma function.
\end{lemma}

Let $\varphi_t$ denote the density of $\mathcal{N}(0,t \rI_{d})$; the dimension is suppressed and should be understood from the context, while the subscript is omitted when $t=1$. Define the following quantities:
\begin{align*}
m_p(W,\Theta_*) &:= \frac{ \int_{\mathbb{R}^{2k}} \|w\|^p \mathrm{Var}\big(p_{W|\Theta_*}(w|\Theta_*) \big)dw}{\big(\int_{\mathbb{R}^{k}} \varphi^2(w) dw \big) \big(\int_{\mathbb{R}^{k}} \|w\|^p \varphi^2(w) dw \big)}\\
M(W,\Theta_*) &:=\sqrt{m_{2k-1}(W,\Theta_*)m_{2k+1}(W,\Theta_*)}.
\end{align*}

The following lemma is adapted from Lemma 6 of \cite{reeves2017conditional} to accommodate our $M(W,\Theta_*)$, definition which slightly differs from theirs. 

\begin{lemma}\label{LEM:MIbound_moment}
If the conditional distribution of $W$ given $\Theta_*$, $p_{W|\Theta_*}$, is absolutely continuous w.r.t. $\leb$ and $M(W,\Theta_*) < \infty$, we have
\[
\sI(W; \Theta_*) \leq \kappa \left(\frac{3\pi k}{8}\right)^{\frac{1}{4}} \sqrt{M(W,\Theta_*)},
\]
where $\kappa$ is as defined in \cref{lemm:4}.
\end{lemma}
\begin{proof}
Lemmas \ref{lemm:4} and \ref{lemm:5} together imply 
\begin{align*}
    \sI(W;\Theta_*) &\leq \kappa \sqrt{\frac{2 \pi^{k+1}}{\Gamma(k)}}\Big(\eta_{2k-1}\big[\mathrm{Var}\big(p_{W|\Theta_*}(w|\Theta_*)\big)\big]\cdot\eta_{2k+1}\big[\mathrm{Var}\big(p_{W|\Theta_*}(w|\Theta_*) \big)\big]\Big)^{\frac{1}{4}},\\
    &= \kappa \sqrt{\frac{2 \pi^{k+1}}{\Gamma(k)}} \sqrt{M(W,\Theta_*) \int_{\mathbb{R}^{k}} \varphi^2(w) dw} \\ & \qquad \qquad\qquad\qquad \times \left(\int_{\mathbb{R}^{k}} \|w\|^{2k-1} \varphi^2(w) dw \int_{\mathbb{R}^{k}} \|w\|^{2k+1} \varphi^2(w) dw   \right)^{\frac{1}{4}}, \\
    &= \kappa \sqrt{M(W,\Theta_*)} \sqrt{2^{1 - 2k} \pi}  \left(\frac{\Gamma(\frac{3k-1}{2}) \Gamma(\frac{3k+1}{2})}{\Gamma^2(\frac{k}{2})\Gamma^2(k)}\right)^{\frac{1}{4}},\\
    &\leq \kappa \sqrt{M(W,\Theta_*)} ({6k  \pi^2})^{\frac{1}{4}} \left(2^{- 2k}\frac{\Gamma(\frac{3k-1}{2})}{\Gamma(\frac{k}{2})\Gamma(k)}\right)^{\frac{1}{2}},\\
    &\leq \frac{\kappa}{2} \sqrt{M(W,\Theta_*)} ({6 k \pi})^{\frac{1}{4}},
\end{align*}
where we have substituted in $M(W,\Theta_*)$ as defined above, have noted that $\int \|w\|^p \varphi^2(w) dw = (4\pi)^{-\frac{k}{2}} \Gamma(\frac{k+p}{2})/\Gamma(\frac{k}{2})$. The last step observes that $2^{- 2k}\frac{\Gamma(\frac{3k-1}{2})}{\Gamma(\frac{k}{2})\Gamma(k)}$ is a decreasing in $k \geq 1$. This can be verified by by using the fact that that as $\Gamma(z + 1) = z \Gamma(z)$, increasing $k$ by 2 will decrease $2^{- 2k}\frac{\Gamma(\frac{3k-1}{2})}{\Gamma(\frac{k}{2})\Gamma(k)}$.
\end{proof}

Given the bound in \cref{LEM:MIbound_moment}, we next bound the moment $M(W,\Theta_*)$. To that end, we control $\eta_p(W,\Theta_*)$. 
For convenience of notation, we set $Z = (Z_1^\tr\  Z_2^\tr)^\tr$, i.e., $Z_1 = X$, and $Z_2 = Y$, and $d_1 = d_x$, $d_2 = d_y$. We consider two independent copies of $Z$, denoted by $Z^{(1)}$ and $Z^{(2)}$. With this notation, we have the following lemma.
\begin{lemma}\label{lemm:7}
For any $p \geq 0$, we have
\begin{align*}
m_p(W,\Theta_*)&= \mathbb{E}\left[(V_{a,1} - R_1)^{-\frac{1}{2}}(V_{a,2} - R_2)^{-\frac{1}{2}}  \left( \left(\frac{V_{g,1}^2 - R_1^2}{V_{a,1} - R_1}\right)^{\frac{p}{2}} +  \left(\frac{V_{g,2}^2 - R_2^2}{V_{a,2} - R_2}\right)^{\frac{p}{2}}\right) \right] \\ 
&\qquad\qquad\qquad\qquad\qquad\qquad\qquad-\mathbb{E}\left[ V_{a,1}^{-\frac{1}{2}}V_{a,2}^{-\frac{1}{2}}  \left( \left(\frac{V_{g,1}^2 }{V_{a,1} }\right)^{\frac{p}{2}} +  \left(\frac{V_{g,2}^2 }{V_{a,2} }\right)^{\frac{p}{2}}\right)\right],
\end{align*}
 where 
 \begin{align*}
     V_{a,i} &= t + \frac{1}{2 d_i}\|Z_i^{(1)}\|^2 + \frac{1}{2 d_i}\|Z_i^{(2)}\|^2,\\
     V_{g,i} &= \sqrt{\left(t + \frac{1}{d_i} \|Z_i^{(1)}\|^2\right)\left(t + \frac{1}{d_i} \|Z_i^{(2)}\|^2\right)},\\
     R_i &= \frac{1}{d_i}\langle Z_i^{(1)}, Z_i^{(2)}\rangle.
 \end{align*}
\end{lemma}
\begin{proof}

 
By the definition of $W$, we have $p_{W|\Theta_*}(w|\theta) = \mathbb{E}\big[\varphi_{t}(w - \theta^\tr Z)\big]$, whereby
 \[
 p^2_{W|\Theta_*}(w|\theta) = \mathbb{E}\big[\varphi_t(w - \theta^\tr Z^{(1)})\varphi_t(w - \theta^\tr Z^{(2)})\big].
 \]
Taking the expectation over the distribution of $\Theta_*$ and swapping the order of expectation yields
 \[
 \mathbb{E}\big[p^2_{W|\Theta_*}(w|\Theta_*)\big] = \mathbb{E}\big[\nu(w, Z^{(1)}, Z^{(2)})\big]
 \]
 where $\nu(y,z^{(1)}, z^{(2)}) = \mathbb{E}[\varphi_t(w - \Theta_*^\tr z^{(1)})\varphi_t(w - \Theta_*^\tr z^{(2)})]$. Note that since $z^{(1)}, z^{(2)}$ are fixed, 
 \[
 \left[\begin{array}{c}\Theta_*^\tr z^{(1)} \\ \Theta_*^\tr z^{(2)}\end{array}\right] \sim \mathcal{N}(0,\Sigma)
 \]
 where
 \[
 \Sigma = \left[\begin{array}{cccc} \|z_1^{(1)}\|^2/d_1 & 0 & \frac{\langle z_1^{(1)}, z_1^{(2)}\rangle}{d_1}& 0\\
 0 & \|z_2^{(1)}\|^2/d_2 &0 & \frac{\langle z_2^{(1)}, z_2^{(2)}\rangle}{d_2}\\
 \frac{\langle z_1^{(1)}, z_1^{(2)}\rangle}{d_1} & 0 & \|z_1^{(2)}\|^2/d_1 & 0 \\
 0 & \frac{\langle z_2^{(1)}, z_2^{(2)}\rangle}{d_2} & 0 & \|z_2^{(2)}\|^2/d_2\end{array} \right] \otimes \rI_k.
 \]
 The proof of Lemma 7 from \cite{reeves2017conditional} shows that
 \begin{align}\label{eq:nu}
 \nu(w,z^{(1)}, z^{(2)}) = (2 \pi)^{2k} \left| \Sigma + t\, \rI_{4k}\right|^{-\frac{1}{2}} \exp\left( - \frac{1}{2} \left\|(\Sigma + t\, \rI_{4k})^{-\frac{1}{2}} \left[\begin{array}{c} w \\ w \end{array}\right]\right\|^2\right).
\end{align}
It is convenient to transform $\Sigma$ into a block-diagonal form. To that end, let us consider the (orthonormal) permutation matrix
\[
\rP = \left[\begin{array}{cccc} 1 & 0 & 0 &0 \\ 0 & 0 & 1 & 0 \\ 0 & 1 & 0 & 0\\ 0 & 0 & 0 & 1\end{array}\right],
\]
and set $\Sigma' = (\rP \otimes \rI_k)\Sigma (\rP \otimes \rI_k)^\tr$. This gives
 \[
 {\Sigma}' = \left[\begin{array}{cccc} \|z_1^{(1)}\|^2/d_1 &  \frac{\langle z_1^{(1)}, z_1^{(2)}\rangle}{d_1}& 0 & 0\\
  \frac{\langle z_1^{(1)}, z_1^{(2)}\rangle}{d_1} & \|z_1^{(2)}\|^2/d_1 & 0 & 0 \\
 0 & 0 &\|z_2^{(1)}\|^2/d_2 & \frac{\langle z_2^{(1)}, z_2^{(2)}\rangle}{d_2}\\
 0 & 0 & \frac{\langle z_2^{(1)}, z_2^{(2)}\rangle}{d_2} & \|z_2^{(2)}\|^2/d_2\end{array} \right] \otimes \rI_k.
 \]
Note that since $\rP$ is a permutation matrix, the eigenvalues of $\Sigma$ and $\tilde\Sigma$ are equal, hence $\left| \Sigma + t \rI\right| = \left| \tilde \Sigma + t \rI\right| = (v_{g,1}^2 - r_1^2)^k(v_{g,2}^2 - r_2^2)^k$. We also obtain
\[
\frac{1}{2}\left\|(\Sigma + t\, \rI_{4k})^{-\frac{1}{2}}\mspace{-3mu} \left[\begin{array}{c} w \\ w \end{array}\right]\right\|^2 \mspace{-6mu}=\mspace{-3mu} \frac{1}{2}\left\|(\tilde\Sigma + t\, \rI_{4k})^{-\frac{1}{2}}\mspace{-4mu} \left[\begin{array}{c} w_1 \\ w_1 \\ w_2 \\ w_2 \end{array}\right]\right\|^2
\mspace{-6mu}=\mspace{-3mu} \left(\frac{v_{a,1} - r_1}{v_{g,1}^2 - r_1^2}\right)\|w_1\|^2 + \left(\frac{v_{a,2} - r_2}{v_{g,2}^2 - r_2^2}\right)\|w_2\|^2,
\]
where $(v_{a,i}, v_{g,i}, r_i)$ are as defined in the lemma statement and we have observed that $\rP^\tr \rP = \rI_4$ and $(\rP \otimes \rI_{k}) \left[\begin{array}{c} w \\ w \end{array}\right] = \left[\begin{array}{c} w_1 \\ w_1 \\ w_2 \\ w_2 \end{array}\right]$. Substituting into \eqref{eq:nu} and simplifying yields
\[
\nu(w,Z^{(1)}, Z^{(2)}) = (V_{a,1} - R_1)^{-\frac{k}{2}}(V_{a,2} - R_2)^{-\frac{k}{2}} {U_1^{\frac{k}{2}} U_2^{\frac{k}{2}}} \varphi_{k,1}^2(U_1^{-\frac{1}{2}} y_1)\varphi_{k,1}^2(U_2^{-\frac{1}{2}} y_2)
\]
where $U_i = (V_{g,i}^2 - R_i^2)/(V_{a,i} - R_i)$. Then the $p$th moment of $\nu$ with respect to $w$ is (using change of variables) is give by
\begin{align*}
    &\eta_p\big[\mathbb{E} \big[p^2_{W|\Theta_*}(w|\Theta_*)\big]\big]
    \\
    &\qquad=\mathbb{E}\left[\int \|w\|^p \nu(w,Z^{(1)}, Z^{(2)} )dw\right]\\
    &\qquad= \mathbb{E}\left[(V_{a,1} - R_1)^{-\frac{k}{2}}(V_{a,2} - R_2)^{-\frac{k}{2}} \int w_1^p  {U_1^{\frac{k}{2}} U_2^{\frac{k}{2}}} \varphi_1^2(U_1^{-\frac{1}{2}} w_1)\varphi_1^2(U_2^{-\frac{1}{2}} z_2)dw\right. \\
    &\qquad\qquad\qquad\left.+ (V_{a,1} - R_1)^{-\frac{k}{2}}(V_{a,2} - R_2)^{-\frac{k}{2}} \int w_2^p  {U_1^{\frac{k}{2}} U_2^{\frac{k}{2}}} \varphi_{k,1}^2(U_1^{-\frac{1}{2}} w_1)\varphi_{k,1}^2(U_2^{-\frac{1}{2}} w_2)dw\right]\\
    &\qquad= \eta_0[\varphi^2]\eta_p[\varphi^2]\mathbb{E}\mspace{-3mu}\left[(V_{a,1} - R_1)^{-\frac{k}{2}}(V_{a,2} - R_2)^{-\frac{k}{2}}  \left( \left(\frac{V_{g,1}^2 - R_1^2}{V_{a,1} - R_1}\right)^{\frac{p}{2}} \mspace{-3mu}+\mspace{-3mu}  \left(\frac{V_{g,2}^2 - R_2^2}{V_{a,2} - R_2}\right)^{\frac{p}{2}}\right)\right]\mspace{-3mu}.\numberthis\label{eq:pythetap}
\end{align*}

Next, we find the $p$th moment of the unconditional squared density $p_{W}^2$. First, as in the conditional case, we write
$p_{W}^2(w) = \mathbb{E}[\tilde{\nu}(w,z^{(1)}, z^{(2)})]$, where \[\tilde{\nu}(w) = \mathbb{E}\big[\varphi_t(w - \Theta_*^{(1)} z^{(1)})\varphi_t(w - \Theta_*^{(2)} z^{(2)})\big]
\]
with $\Theta_*^{(i)}$, for $i=1.2$, being independent copies of $\Theta_*$. This independence in turn decorrelates $\Theta_*^{(1)} z^{(1)}$ and $\Theta_*^{(2)} z^{(2)}$, i.e., 
\[
 \left[\begin{array}{c}\Theta_*^{(1)} z^{(1)} \\ \Theta_*^{(2)} z^{(2)}\end{array}\right] \sim \mathcal{N}\big(0,\mathrm{diag}(\Sigma)\big).
 \]
Proceeding as in the correlated case above yields
 \[
     \eta_p\big[\mathbb{E}[p_{W}^2(w)]\big] = \eta_0[\varphi^2]\eta_p[\varphi^2]=\mathbb{E}\left[V_{a,1}^{-\frac{k}{2}}V_{a,2}^{-\frac{k}{2}}  \left( \left(\frac{V_{g,1}^2 }{V_{a,1} }\right)^{\frac{p}{2}} +  \left(\frac{V_{g,2}^2 }{V_{a,2} }\right)^{\frac{p}{2}}\right)\right],
 \]
and combining this with \eqref{eq:pythetap} gives
\begin{align*}
   m_p(W,\Theta_*) &= \frac{\eta_p\big[\mathrm{Var} \big(p^2_{W|\Theta_*}(w|\Theta_*)\big)\big]}{ \eta_0[\varphi^2]\eta_p[\varphi^2]}\\
    &= \mathbb{E}\left[(V_{a,1} - R_1)^{-\frac{k}{2}}(V_{a,2} - R_2)^{-\frac{k}{2}}  \left( \left(\frac{V_{g,1}^2 - R_1^2}{V_{a,1} - R_1}\right)^{\frac{p}{2}} +  \left(\frac{V_{g,2}^2 - R_2^2}{V_{a,2} 
    - R_2}\right)^{\frac{p}{2}}\right) \right] \\ 
    & \qquad \qquad\qquad\qquad\qquad\qquad\qquad -\mathbb{E}\left[ V_{a,1}^{-\frac{k}{2}}V_{a,2}^{-\frac{k}{2}}  \left( \left(\frac{V_{g,1}^2 }{V_{a,1} }\right)^{\frac{p}{2}} +  \left(\frac{V_{g,2}^2 }{V_{a,2} }\right)^{\frac{p}{2}}\right)\right].
\end{align*}
 
\end{proof}


It remains to bound the expectation in Lemma \ref{lemm:7}. 
We start with the following bound.
\begin{lemma}
\label{lem:mp}
For any $p \geq 0$, we have
\[
m_p(W, \Theta_*) \leq \mathbb{E}\left[ V_{a,2}^{-\frac{k}{2}}V_{a,1}^{-\frac{k-p}{2}} g_p\left(\frac{R_1}{V_{a,1}},\frac{R_2}{V_{a,2}} \right) 
+V_{a,1}^{-\frac{k}{2}}V_{a,2}^{-\frac{k-p}{2}} g_p\left(\frac{R_2}{V_{a,2}},\frac{R_1}{V_{a,1}} \right)\right],
\]
where $g_{k,p}: (-1,1) \rightarrow \mathbb{R}$ is given by $g_{k,p}(u,v) := (1 - v)^{-\frac{k}{2}}(1 - u)^{-\frac{k}{2}} ( 1+ u)^{\frac{p}{2}} - 1$.
\end{lemma}
\begin{proof}
Note that for $i,j=1,2$ with $i\neq j$, we have
\begin{align*}
    (V_{a,j} - R_j)^{-\frac{k}{2}}&(V_{a,i} - R_i)^{-\frac{k}{2}}\left(\frac{V_{g,i}^2 - R_i^2}{V_{a,i} - R_i}\right)^{\frac{p}{2}} - V_{a,i}^{-\frac{k}{2}}V_{a,j}^{-\frac{k}{2}} \left(\frac{V_{g,i}^2}{V_{a,i}}\right)^{\frac{p}{2}}\\
    &= V_{a,i}^{-\frac{k+p}{2}} V_{a,j}^{-\frac{k}{2}}V_{g,i}^{p}\left[ \left(1 - \frac{R_j}{V_{a,j}}\right)^{-\frac{k}{2}}\left(1 - \frac{R_i}{V_{a,i}}\right)^{-\frac{k}{2}} \left({1 + \frac{R_i}{V_{g,i}}}\right)^{\frac{p}{2}} - 1  \right] \\
    &= g_p\left(\frac{R_i}{V_{a,i}},\frac{R_j}{V_{a,j}}\right) V_{a,j}^{-\frac{k}{2}}V_{a,i}^{-\frac{k+p}{2}} V_{g,i}^p\\
    &\leq g_p\left(\frac{R_i}{V_{a,i}},\frac{R_j}{V_{a,j}}\right) V_{a,j}^{-\frac{k}{2}}V_{a,i}^{-\frac{k-p}{2}},
\end{align*}
where we have noted that $V_{g,i} \leq V_{a,i}$ since the geometric mean is upper bounded by the arithmetic mean.
Substituting into Lemma \ref{lemm:7} completes the proof.
\end{proof}


To make the expectation of $g_p$ tractable we next upper bound it by a quadratic function.

\begin{lemma}\label{lem:16}
For any $t>0$ and $(r_1,r_2)$ such that $|r_1| \leq c_1, |r_2| \leq c_2$, for some $c_1, c_2$, we have
\[
g_{p}\left(\frac{r_i}{t + c_i}, \frac{r_j}{t + c_j}\right) 
\leq \frac{k + p}{2}\frac{r_i}{t + c_i } + \frac{k}{2}\frac{r_j}{t + c_j} + 
t^{-k} (t + 2c_i)^{\frac{p}{2}} (t + c_i)^{\frac{ k-p}{2}}(t + c_j)^{\frac{ k}{2}}\left(\frac{r_i^2+r_j^2}{c_i^2\wedge c_j^2}\right),
\]
where $i,j=1,2$ with $i\neq j$.
\end{lemma}
\begin{proof}
Since $g_p(0,0) = 0$, we decompose 
\[
g_p(u,v) = \big(\nabla g_p(0,0)\big)^\tr \left[\begin{array}{c}u\\v\end{array}\right]  + h_p(u,v)(u^2 + v^2),
\]
where $h_p(u,v) = \left(g_p(u,v) - \nabla g_p(0,0)^\tr \left[\begin{array}{c}u\\v\end{array}\right]\right)/(u^2+v^2)$. It can be verified that $h_p$ is non-negative and nondecreasing in both arguments. Hence, for all $-1 < u \leq z_u < 1$, $-1 < v \leq z_v < 1$, we have
\[
g_p(u,v) \leq \nabla g_p(0,0)^\tr \left[\begin{array}{c}u\\v\end{array}\right]  + h_p(z_u,z_v)(u^2 + v^2).
\]
Furthermore, for $z_u,z_v > 0$,
\[
h_p(z_u,z_v) \leq \frac{g_p(z_u,z_v)}{z_u^2 + z_v^2} = \frac{1}{z_u^2+z_v^2}(1 - z_v)^{-\frac{k}{2}}(1 - z_u)^{-\frac{k}{2}} (1 + z_u)^{\frac{p}{2}}.
\]
Using $z_u = c_i/(t+c_i)$, $z_v = c_j/(t + c_j)$, we obtain
\begin{align*}
h_p(z_u,z_v) &\leq \frac{1}{2\min(z_u^2,z_v^2)}t^{-k}(t + c_j)^{\frac{k}{2}}(t + c_i)^{\frac{k-p}{2}} (t + 2 c_i)^{\frac{p}{2}},
\end{align*}
from which the result follows.

\end{proof}

Next, we provide a bound on $M(W,\Theta_*)$ subject to a.s. boundedness assumption on the squared norms of the random variables. The subsequently presented \cref{LEM:last} then relaxes this assumption to a bound on the MI term of interest.

\begin{lemma}\label{lem:M}
Suppose that $\lambda_{\min} \leq \frac{\|\rX\|^2}{d_x}\wedge\frac{\|\rY\|^2}{d_y} \leq \frac{\|\rX\|^2}{d_x}\vee\frac{\|\rY\|^2}{d_y}\leq \lambda_{\max}$ a.s. Then 
\[
M(W,\Theta_*) \leq 2^{\frac{1}{4}} \left(\frac{\lambda_{\max}}{\lambda_{\min}}\right)^{\frac{k}{2}}\left[4k \frac{\beta_1(Z_1) + \beta_1(Z_2)}{\lambda_{\min}} + 2\left(1 + \frac{2\lambda_{\max}}{t}\right)^k\frac{\beta_2^2(Z_1) + \beta_2^2(Z_2)}{\lambda_{\min}^2} \right].
\]
\end{lemma}
\begin{proof}
Using Lemma \ref{lem:16} and the definitions of $\beta_r^r$, $R_i$, and $V_{a,i}$, for $p = 2k - 1$, we have
\begin{align*}
    \mathbb{E}\left[V_{a,j}^{-\frac{k}{2}}V_{a,i}^{-\frac{k-p}{2}} g_p\left(\frac{R_i}{V_{a,i}},\frac{R_j}{V_{a,j}} \right) \right]
    \leq (t+\lambda_{\min})^{-\frac{k}{2}}(t +&\lambda_{\max})^{\frac{k-1}{2}} \left(\frac{3k-1}{2}\frac{\beta_1(Z_i)}{\lambda_{\min}} + \frac{k}{2} \frac{\beta_1(Z_j)}{\lambda_{\min}}\right) \\ 
    &+ t^{-k}(t + 2\lambda_{\max})^{\frac{2k-1}{2}}\frac{\beta_2^2(Z_i) + \beta_2^2(Z_j)}{\lambda_{\min}^2}.
\end{align*}
By \cref{lem:mp}, this yields
\begin{align*}
    m_{2k-1}(W,\Theta_*)\leq (t+\lambda_{\min})^{-\frac{k}{2}}&(t +\lambda_{\max})^{\frac{k-1}{2}}(4k-1) \frac{\beta_1(Z_1) + \beta_1(Z_2)}{\lambda_{\min}}
    \\ & \qquad \qquad\qquad + 2t^{-k}(t + 2\lambda_{\max})^{\frac{2k-1}{2}}\frac{\beta_2^2(Z_1) + \beta_2^2(Z_2)}{\lambda_{\min}^2}.
\end{align*}
Similarly for $p = 2k + 1$,
\begin{align*}
  m_{2k+1}(W,\Theta_*)\leq (t+\lambda_{\min})^{-\frac{k}{2}}&(t +\lambda_{\max})^{\frac{k+1}{2}}(4k+1) \frac{\beta_1(Z_1) + \beta_1(Z_2)}{\lambda_{\min}}
     \\ & \qquad \qquad\qquad + 2t^{-k}(t + 2\lambda_{\max})^{\frac{2k+1}{2}}\frac{\beta_2^2(Z_1) + \beta_2^2(Z_2)}{\lambda_{\min}^2} .
\end{align*}
By the definition of $M(W,\Theta_*) = \sqrt{m_{2k-1}(W,\Theta_*)m_{2k+1}(W,\Theta_*)}$, we obtain
\begin{align*}
    M(W,\Theta_*) \leq 2^{\frac{1}{4}} \left(\frac{\lambda_{\max}}{\lambda_{\min}}\right)^{\frac{k}{2}}\left[4k \frac{\beta_1(Z_1) + \beta_1(Z_2)}{\lambda_{\min}} + 2\left(1 + \frac{2\lambda_{\max}}{t}\right)^k\frac{\beta_2^2(Z_1) + \beta_2^2(Z_2)}{\lambda_{\min}^2} \right].
\end{align*}
where we have used the fact that the geometric mean is upper bounded by the arithmetic mean.


\end{proof}

The derivation is concluded by adapting Lemma 12 of \cite{reeves2017conditional} to our notation and setting. 
\begin{lemma}\label{LEM:last}
Let $\mathcal{E} \subseteq \mathbb{R}^{d_x + d_y}$ be measurable. Then 
\begin{align*}
    \sI(W;\Theta_*) \leq \frac{k}{2}\log\left(1 + \frac{\lambda}{t}\right) \left(\mu_Z(\mathcal{E}^c) + \frac{\alpha(X) + \alpha(Y)}{\lambda}\right) + \mu_Z(\mathcal{E})\sI(W;\Theta_* | Z \in \mathcal{E}).
\end{align*}
\label{lemm:E}
\end{lemma}
\begin{proof}
Letting $U = \mathds{1}_{\mathcal{E}}(Z)$, the MI chain rule gives
\[
    \sI(W, U; \Theta_*) = \sI(W;\Theta_*) + \sI(U; \Theta_* | W)= \sI(W;\Theta_* | U) + \sI(U;\Theta_*). 
\]
Since $\sI(U;\Theta_*) = 0$, we have $\sI(W;\Theta_*) \leq \sI(W;\Theta_* | U)$, and expanding the conditioning yields
\[
\sI(W;\Theta_*) \leq \mu_Z(\mathcal{E}^c)\sI(W;\Theta_* |Z \notin \mathcal{E}) + \mu_Z(\mathcal{E})\sI(W;\Theta_* | Z \in \mathcal{E}) .
\]
Recall that $W = \Theta_*^\tr Z + \sqrt{t}N$, where $\Theta_*$ and $N$ are independent Gaussians. Therefore, conditioned on $Z$, $(W, \Theta_*)$ is jointly Gaussian and we have
\begin{align*}
    &\sI(W;\Theta_* | Z \notin \mathcal{E})\\
    &\leq \sI(W;\Theta_* | Z, Z \notin \mathcal{E})\\
    &= \sh(W | Z,  Z\notin \mathcal{E}) - \sh(W | \Theta_*, Z,  Z\notin \mathcal{E})\\
    &= \frac{k}{2} \mathbb{E}\left[\left. \log\left(\frac{t + \frac{1}{d_x} \|X\|^2}{t}\right) + \log\left(\frac{t + \frac{1}{d_y} \|Y\|^2}{t}\right)\right| Z \notin{\mathcal{E}}\right]\\
    &= \frac{k}{2 \mu_Z(\mathcal{E}^c)} \left(  \mspace{-2mu} \mathbb{E}\left[\left( \log\left(\frac{t \mspace{-2mu}+\mspace{-2mu} \frac{1}{d_x} \|X\|^2}{t+\lambda}\right) \mspace{-2mu}+\mspace{-2mu} \log\left(\frac{t \mspace{-2mu}+\mspace{-2mu} \frac{1}{d_y} \|Y\|^2}{t+\lambda}\right)\right)\mspace{-2mu}{\mathds{1}_{\mathcal{E}^c}}(Z)\right] \mspace{-2mu}+\mspace{-2mu} \log\left(1\mspace{-2mu} +\mspace{-2mu} \frac{\lambda}{t}\right)\mu_Z(\mathcal{E}^c) \right)\\
    &\leq \frac{k}{2\mu_Z(\mathcal{E}^c)} \log\left(1 + \frac{\lambda}{t}\right)\left(\mu_Z(\mathcal{E}^c)+\frac{\alpha(X) + \alpha(Y)}{\lambda}\right),
\end{align*}
where the last inequality follows from Lemma 19 of \cite{reeves2017conditional}. Combining expressions yields the lemma.
\end{proof}

To use the bound on $M(W,\Theta_*)$ from Lemma \ref{lem:M}, we therefore let 
\[
\mathcal{E}_i = \left\{ w_i \in \mathbb{R}^{d_i}: \left|\frac{1}{d_i} \|z_i\|^2 - \lambda\right| \leq \frac{\epsilon}{2} \lambda\right\},
\]
where $\epsilon \in (0,1]$.
Markov's inequality implies $\mu_{Z_i}(\mathcal{E}_i^c) \leq \frac{2}{\epsilon \lambda} \alpha(Z_{i})$. 
Define $\mathcal{E} = \mathcal{E}_1 \times \mathcal{E}_2$, so that by the union bound $\mu_Z(\mathcal{E}^c) \leq \mu_{Z_1}(\mathcal{E}_1^c) + \mu_{Z_2}(\mathcal{E}_2^c) \leq \frac{4}{\epsilon \lambda}\alpha(Z_1) \vee \alpha(Z_2)$. Let $Z'$ be drawn according to the conditional distribution of $Z$ given $Z \in \mathcal{E}$, and set $W' = \Theta_*^\tr Z' + \sqrt{t} N$. By Lemma \ref{lem:M}, we have
\[
M(W', \Theta_*) \leq 2^{\frac{1}{4}} \left(\frac{1+\epsilon}{1-\epsilon}\right)^{\frac{k}{2}}\left[4k \frac{\beta_1(Z'_1) + \beta_1(Z'_2)}{\lambda} + 2\left(1 + \frac{2(1+\epsilon)\lambda}{t}\right)^k\frac{\beta_2^2(Z'_1) + \beta_2^2(Z'_2)}{\lambda^2} \right]. 
\]
Hence
\begin{align*}
\sI(W;\Theta_* | Z \in \mathcal{E}) &\leq \frac{\kappa}{2}(6\sqrt{2}k\pi)^{\frac{1}{4}}\left(\frac{1+\epsilon}{1-\epsilon}\right)^{\frac{k}{4}}\left[2\sqrt{2k} \sqrt{\frac{\beta_1(Z'_1)\vee \beta_1(Z'_2)}{\lambda}}\right.\\
&\qquad\qquad\qquad\qquad\qquad\qquad\qquad\ \ +\left. 2\left(1 + \frac{2(1+\epsilon)\lambda}{t}\right)^\frac{k}{2}\frac{\beta_2(Z'_1)\vee \beta_2(Z'_2)}{\lambda} \right],
\end{align*}
and applying Lemma \ref{lemm:E}, while noting that $\beta^r_r(Z_i') \leq \frac{\beta^r_r(Z_i)}{\mu_{Z_i}^2(\mathcal{E}_i)} \leq \frac{\beta^r_r(Z_i)}{\mu_{Z}^2(\mathcal{E})} $, for $i=1,2$ and $r=1,2$, yields the result of \cref{thm:2dinf}.

\end{document}